\documentclass[a4paper,abstracton,10pt]{article}

\usepackage{bm}
\usepackage[a4paper]{geometry}
\usepackage{booktabs}
\usepackage{stmaryrd}
\usepackage{pgfplots}
\usepackage{hyperref}
\usepackage{calc}
\usepackage{enumitem}
\usepackage{ulem}
\usepackage{mathtools}
\pgfplotsset{compat=newest}
\pgfplotsset{plot coordinates/math parser=false}

\usepackage{pifont}
\usepackage[numbers]{natbib}
\usepackage{color}

\definecolor{darkgreen}{rgb}{0.1,0.5,0.1}
\definecolor{darkblue}{rgb}{0.1,0.1,0.9}
\newcommand{\tocheck}[1]{{\color{darkgreen}{#1}}}

\newcommand{\sign}{\mathsf{sign}}
\newcommand{\spinv}{\mathsf{spinv}}

\newcommand{\E}{\mathbb{E}}

\newcommand{\range}{\mathcal{R}}

\newcommand{\ginvset}{\mathcal{G}}

\newcommand{\T}{\top}

\newcommand{\nullspace}{\ensuremath{\mathcal{N}}}

\newcommand{\vectorize}{\mathrm{vec}}

\newcommand{\ginv}[2]{\ensuremath{\mathsf{ginv}_{#1}(#2)}}

\newcommand{\ind}[1]{\mathbb{I}_{#1}}

\newcommand{\proj}{\ensuremath{\mathsf{proj}}}

\newcommand{\prob}{\mathbb{P}}

\newcommand{\dist}{\mathsf{dist}}

\newcommand{\erfc}{\mathrm{erfc}}

\newcommand{\setS}{{\cal S}}

\newcommand{\var}{\mathsf{Var}}

\usepackage{amsmath}
\usepackage{amssymb}
\usepackage{amsthm}
\usepackage{latexsym}
\usepackage{verbatim}
\usepackage{graphicx}
\usepackage{mdframed}

\definecolor{shade}{rgb}{0.93,0.93,0.93}

\newmdtheoremenv [backgroundcolor=shade, %
innertopmargin = -4pt , %
innerbottommargin =2pt , %
innerleftmargin = 1pt , %
innerrightmargin = 1pt, %
splittopskip = \topskip, %
skipbelow= 6pt, %
skipabove=6pt, %
topline=false,bottomline=false,leftline=false,rightline=false,]{theorem}{Theorem}[section]

\newmdtheoremenv[backgroundcolor=shade,%
innertopmargin = -4pt,%
innerbottommargin =2pt,%
innerleftmargin = 1pt,%
innerrightmargin = 1pt,%
splittopskip = \topskip,%
skipbelow= 6pt,%
skipabove=6pt,%
topline=false,bottomline=false,leftline=false,rightline=false,]{lemma}{Lemma}[section]

\newmdtheoremenv [backgroundcolor=shade, %
innertopmargin = -4pt , %
innerbottommargin =2pt , %
innerleftmargin = 1pt , %
innerrightmargin = 1pt, %
splittopskip = \topskip, %
skipbelow= 6pt, %
skipabove=6pt, %
topline=false,bottomline=false,leftline=false,rightline=false,]{corollary}{Corollary}[section]

\newmdtheoremenv [backgroundcolor=shade, %
innertopmargin = -4pt , %
innerbottommargin =2pt , %
innerleftmargin = 1pt , %
innerrightmargin = 1pt, %
splittopskip = \topskip, %
skipbelow= 6pt, %
skipabove=6pt, %
topline=false,bottomline=false,leftline=false,rightline=false,]{example}{Example}[section]

\newmdtheoremenv [backgroundcolor=shade, %
innertopmargin = -4pt , %
innerbottommargin =2pt , %
innerleftmargin = 1pt , %
innerrightmargin = 1pt, %
splittopskip = \topskip, %
skipbelow= 6pt, %
skipabove=6pt, %
topline=false,bottomline=false,leftline=false,rightline=false,]{definition}{Definition}[section]

\newmdtheoremenv [backgroundcolor=shade, %
innertopmargin = -4pt , %
innerbottommargin =2pt , %
innerleftmargin = 1pt , %
innerrightmargin = 1pt, %
splittopskip = \topskip, %
skipbelow= 6pt, %
skipabove=6pt, %
topline=false,bottomline=false,leftline=false,rightline=false,]{remark}{Remark}[section]

\newmdtheoremenv [backgroundcolor=shade, %
innertopmargin = -4pt , %
innerbottommargin =2pt , %
innerleftmargin = 1pt , %
innerrightmargin = 1pt, %
splittopskip = \topskip, %
skipbelow= 6pt, %
skipabove=6pt, %
topline=false,bottomline=false,leftline=false,rightline=false,]{proposition}{Proposition}[section]

{\begin{list}               
    {$\bullet$ \hfill}{
        \setlength{\leftmargin}{\parindent}
        \setlength{\parsep}{0.04\baselineskip}
        \setlength{\itemsep}{0.5\parsep}
        \setlength{\labelwidth}{\leftmargin}
        \setlength{\labelsep}{0em}}
    }
{\end{list}}

\providecommand{\cref}[1]{Chapter~\ref{chap:#1}}

\providecommand{\R}{\ensuremath{\mathbb{R}}}
\providecommand{\C}{\ensuremath{\mathbb{C}}}
\providecommand{\N}{\ensuremath{\mathbb{N}}}

\providecommand{\abs}[1]{\left|#1\right|}
\providecommand{\norm}[1]{\left\lVert#1\right\rVert}
\providecommand{\bignorm}[1]{\bigg\lVert#1\bigg\rVert}

\providecommand{\set}[1]{\left\lbrace#1\right\rbrace}

\providecommand{\bydef}{\overset{\mathrm{def}}{=}}

\providecommand{\parder}[2]{{\partial{#1} \over \partial{#2}}}
\providecommand{\parderr}[2]{{\partial^2{#1} \over \partial{#2}^2}}

\providecommand{\di}{\ensuremath{\text{d}}}
\providecommand{\e}{\ensuremath{\mathrm{e}}}

\renewcommand{\vec}[1]{\ensuremath{\mathbf{#1}}}
\providecommand{\mat}[1]{\ensuremath{\mathbf{#1}}}

\providecommand{\wt}[1]{\ensuremath{\widetilde{#1}}}

\providecommand{\mA}{\mat{A}} \providecommand{\mB}{\mat{B}}

\providecommand{\mI}{\mat{I}}

\providecommand{\mM}{\mat{M}}

\providecommand{\mW}{\mat{W}}

 \providecommand{\mG}{\mat{G}}

\providecommand{\mX}{\mat{X}}

\providecommand{\va}{\vec{a}} \providecommand{\vb}{\vec{b}}
 
\providecommand{\ve}{\vec{e}} 
 
\providecommand{\vg}{\vec{g}}
\providecommand{\vh}{\vec{h}} 
  
 \providecommand{\vp}{\vec{p}}
 
\providecommand{\vs}{\vec{s}}

\providecommand{\vu}{\vec{u}} \providecommand{\vw}{\vec{w}}

\providecommand{\vx}{\vec{x}} \providecommand{\vy}{\vec{y}}
\providecommand{\vz}{\vec{z}} 
 \providecommand{\vzero}{\vec{0}}
 \providecommand{\vv}{\vec{v}}

\DeclareMathOperator*{\argmin}{arg\,min}

\title{Beyond Moore-Penrose\\Part II: The Sparse Pseudoinverse}
\author{Ivan Dokmani\'c and R\'emi Gribonval}
\date{}

\begin{document}

\maketitle
\normalem


\begin{abstract}
This is the second part of a two-paper series on norm-minimizing generalized inverses. In Part II we focus on generalized inverses that are minimizers of entrywise $\ell^p$ norms, with the main representative being the \emph{sparse pseudoinverse} for $p = 1$. We are motivated by the idea to replace the Moore-Penrose pseudoinverse by a sparser generalized inverse which is in some sense well-behaved. Sparsity means that it is faster to multiply by resulting matrix; well-behavedness means that we do not lose much with respect to the least-squares performance of the MPP.

We first address questions of uniqueness and non-zero count of (putative) sparse pseudoinverses. We show that a sparse pseudoinverse is generically unique, and that it indeed reaches optimal sparsity for almost all matrices. We then turn to proving our main stability result: finite-size concentration bounds for the Frobenius norm of $\ell^p$-minimal inverses for $1 \leq p \leq 2$. Our proof is based on tools from convex analysis and random matrix theory, in particular the recently developed convex Gaussian min-max theorem. Along the way we prove several folklore facts about sparse representations and convex programming that were to the best of our knowledge thus far unproven.
\end{abstract}

\section{Introduction}
\label{sec:introduction}

A generalized inverse (or a pseudoinverse) of a rectangular matrix is a matrix that has some properties of the usual inverse of a regular square matrix. We call $\mX \in \C^{n \times m}$ a generalized inverse of $\mA \in \C^{m \times n}$ if it holds that $\mA \mX \mA = \mA$; we denote the set of all generalized inverses of $\mA$ by $\mathcal{G}(\mA)$.

The most common generalized inverse is the famous Moore-Penrose pseudoinverse (MPP). Some of the many uses of the MPP are to compute the least-squares fit to an overdetermined system of linear equations, or to find the shortest solution of an underdetermined system; many other related results are presented in Part I.

However, the MPP is only one out of infinitely many possible generalized inverses. By relaxing the requirements needed to get the MPP we free up degrees of freedom that can be optimized to promote other interesting properties. Our study focuses on a particular strategy to obtain alternative pseudoinverses: among all pseudoinverses, the MPP has the smallest Frobenius norm\footnote{It also minimizes a number of other norms (see Part I).}, so it seems natural to ask what happens if we replace the Frobenius norm by other matrix norms.

A generalized inverse of $\mA \in \C^{m \times n}$, $m < n$ that minimizes some arbitrary (quasi)norm $\nu$ is defined as
\begin{align}
  \ginv{\nu}{\mA} &\bydef \argmin_{\mX} \ \norm{\mX}_{\nu} \ \ \text{subject to} \ \ \mX \in \ginvset(\mA). \nonumber
\end{align}
Strictly speaking, $\ginv{\nu}{\mA}$ is a set, since the corresponding optimization program may have more than one solution. This, however, will not be the case for most studied norms and matrices (see, for example, Section \ref{sec:spinv_is_unique_and_optimal}) so we write ``$=$'' instead of ``$\in$'' and warn the reader when extra care is advised. We also note that if $\norm{\,\cdot\,}_\nu$ is a bona fide norm, then $\ginv{\nu}{\mA}$ involves solving a convex program. At least in principle, this means that it can actually be computed in the sense that any first-order scheme leads to a global optimum.

In this Part II, we specifically focus on entrywise $\ell^p$ norms\footnote{For brevity, we loosely call ``norm''  any quasi-norm such as $\ell^{p}$, $p<1$, as well as the ``pseudo-norm'' $\ell^{0}$.}  which are straightforward extensions of vector $\ell^p$ norms. Concretely, for $\mM \in \C^{m \times n}$ and $0 < p < \infty$ we have
\(
  \norm{\mM}_p \bydef \norm{\mathrm{vec}(\mM)}_p = \left( \sum_{ij} \abs{m_{ij}}^p \right)^{1/p},
\)
with $\mathrm{vec}(\,\cdot\,)$ denoting the concatenation of the columns of the argument matrix.
A particular entrywise norm is the Frobenius norm associated to $p = 2$ and mininimizing it gives the MPP. But our main motivation to look at entrywise norms is rather the case $p = 1$, which, as we show, leads to sparse pseudoinverses. 

The motivation is that applying a sparse pseudoinverse requires less computation than applying a full one
\cite{Dokmanic:2013bo,Li:2013cx,Casazza:ev,Krahmer:2012cy}. We could take advantage of this fact if we knew how to compute a sparse pseudoinverse that is in some sense stable to noise. Ignoring the last requirement, finding the sparsest pseudoinverse may be formulated as
\begin{equation}
    \label{eq:intro_l0}
    \ginv{0}{\mA} \bydef \argmin_{\mX}~\norm{\vectorize(\mX)}_{0}\ \text{subject to} \ \mX \in \ginvset(\mA),
\end{equation}
where $\norm{\, \cdot \,}_{0}$ counts the total number of non-zero entries
in a vector and $\vectorize(\cdot)$ transforms a matrix into a vector by stacking its columns. The non-zero count gives the naive complexity of applying
$\mX$ or its adjoint to a vector (for an illustration, see Figure \ref{fig:spinv}).

Any optimally sparse generalized inverse of $\mA$ in the sense of the $\ell^{0}$ norm \eqref{eq:intro_l0} is by definition in the set $\ginv{0}{\mA}$. When $\mA \in \C^{m \times n}$, $m < n$, has full column rank the condition $\mX \in \ginvset(\mA)$ is equivalent to $\mA\mX=\mI_{m}$, and computing an element from this set can be expressed columnwise, as a collection of $\ell^{0}$ minimization problems
\[
  [\ginv{0}{\mA}]_i = \argmin_{\mA \vx = \ve_i} \norm{\vx}_{\ell^0}.
\]
Even though optimization problems of this kind are NP-hard \cite{natarajan95:_spars,DMA97}, we will see that for most matrices $\mA$ finding \emph{a} solution is trivial and not very useful: just invert any full-rank $m \times m$ submatrix and set the rest to zero. What we mean by ``not useful'' is that the resulting matrix is in general poorly conditioned. 
On the other hand, the vast literature on conditions of equivalence between solutions of $\ell^{0}$ and $\ell^{1}$ minimization (see, e.g. \cite{Donoho:2006ci}) suggests that $\ell^{1}$ minimization is a good computationally tractable proxy. In fact, we will show that the sparsest pseudoinverse can be computed by convex optimization; that is, a convex program generically provides \emph{a} minimizer of the $\ell^{0}$ norm. Moreover, we will show that this minimizer is generically unique. This motivates the definition of the specific notation
\begin{equation}
  \label{eq:l1_instead_l0}
  \spinv(\mA) \bydef \ginv{1}{\mA} = \argmin_{\mX}~\norm{\vectorize(\mX)}_1\ \text{subject to} \ \mX \in \ginvset(\mA).
\end{equation}
Not only is \eqref{eq:l1_instead_l0} computationally tractable but we will show  that it leads to well-behaved matrices that are indeed sparse. Intuitive reasoning is as follows: an $m \times n$ matrix $\mA$ is generically full rank, hence the constraint $\mX \in \ginvset(\mA)$ is generically equivalent to $\mA \mX = \mI_{m}$. As this is a shorthand notation for $m^2$ linear equations, it constrains $m^2$ degrees of freedom. The matrix $\mX$ has $nm$ entries, leaving us with $nm - m^2$ degrees of freedom, which will (hopefully) be set to zero by $\ell^1$ minimization.

\begin{figure}
\centering
\includegraphics[scale=.8]{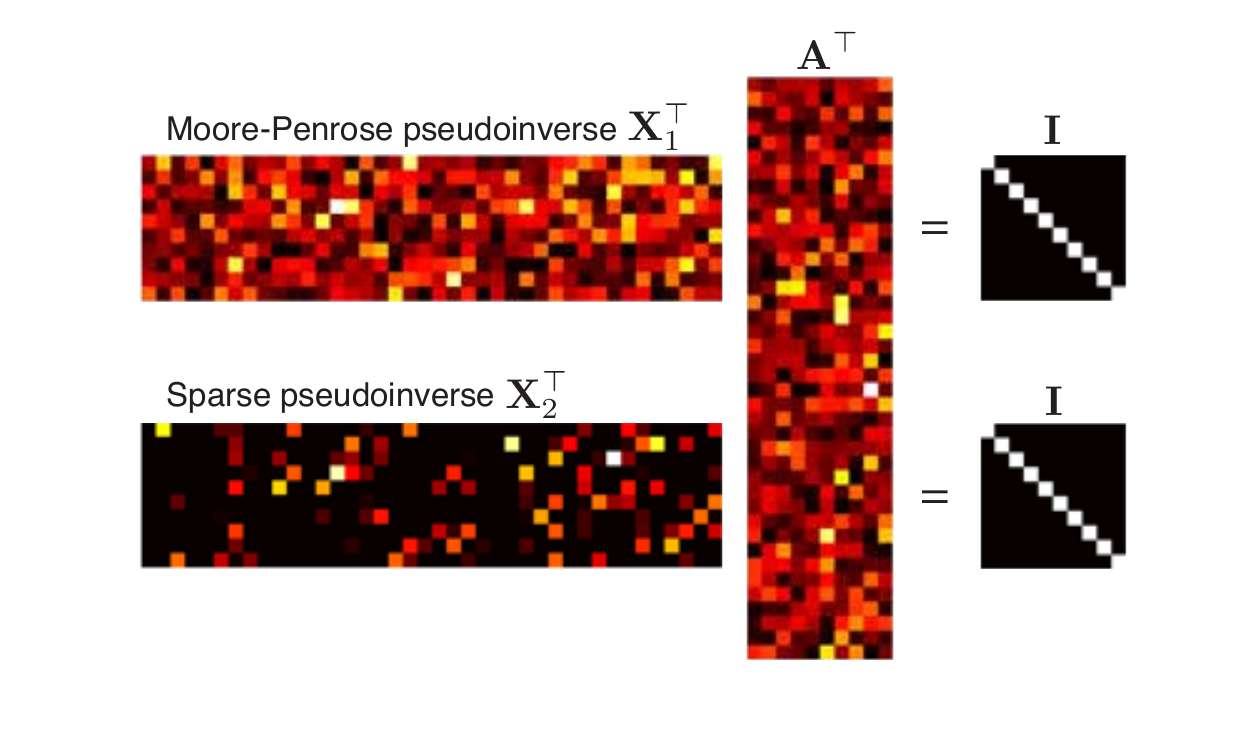}
\caption{Illustration of the sparse pseudoinverse: $\mA$ is a realization of a $10 \times 30$ iid Gaussian matrix; the MPP $\mX_1 = \mA^\dag$ is a full matrix; the sparse pseudoinverse $\mX_2 = \spinv(\mA)$ has only $10^2$ non-zeros among $10\times30$ entries.}
\label{fig:spinv}
\end{figure}

\subsection{Prior Art}

Several earlier work address sparsity of generalized inverses and dual frames \cite{Casazza:ev, Krahmer:2012cy, Li:2013cx}. These works concentrate on existence results and explicit constructions of sparse frames and sparse dual frames with prescribed spectra. Krahmer, Kutyniok, and Lemvig \cite{Krahmer:2012cy} establish sharp bounds on the sparsity of dual frames, showing that generically, for $\mA \in \C^{m \times n}$, the sparsest dual has $mn - m^2$ zeros. Li, Liu, and Mi \cite{Li:2013cx} provide bounds on the sparsity of duals of Gabor frames which are better than generic bounds. They also introduce the idea of using $\ell^p$ (for $p = 1)$ minimization to compute these dual frames, and  show that under certain conditions, the $\ell^1$ minimization yields the sparsest possible dual Gabor frame. Further examples of non-canonical dual Gabor frames are given by Perraudin \emph{et al.}, who use convex optimization to derive dual Gabor frames with more favorable properties than the canonical one \cite{Perraudin:2014we}, particularly in terms of time-frequency localization.

Our stability results fit squarely within random matrix theory. However, a perusal of the corresponding rich literature indicates that results on finite-size concentration bounds for norms of pseudoinverses are scarce. One notable exception is \cite{Halko:2011kg} which gives the expected Frobenius norm of the MPP and a corresponding upper bound on a probability of large deviation; we note that our techniques allow us to get a finite-size concentration result for the MPP for a range of deviations complementary to the one considered in \cite{Halko:2011kg}. It is also worth noting that there are a number of results for square matrices \cite{Rudelson:2008gv,Vershynin:2014jv}.

Finally, the sparse pseudoinverse was previously studied in \cite{Dokmanic:2013bo}, where it was shown empirically that the minimizer is indeed a sparse matrix, and that it can be used to speed up the resolution of certain inverse problems.

\subsection{Our Contributions and Paper Outline}

We address fundamental questions of uniqueness and stability of the \emph{sparse pseudoinverse}. First, in Section \ref{sec:spinv_is_unique_and_optimal}, we show that minimizing the entrywise $\ell^1$ norm of the generalized inverse generically results in a maximally sparse inverse. By exploiting connections between basis pursuit and basis pursuit denoising/lasso we also show that the optimum is generically unique. We then prove in Section~\ref{sub:stability} that unlike other, simpler strategies that yield a sparsest generalized inverse, $\ell^1$ minimization produces a \emph{good}, well-conditioned matrix. 

We measure well-behavedness by the Frobenius norm, 
which arises naturally in the context of linear inverse problems with noisy measurements: for white noise the output MSE is directly proportional to the Frobenius norm. Our main result is a 
characterization of the typical Frobenius norm of the $\ell^p$-minimal generalized inverse for random Gaussian matrices for $1 \leq p \leq 2$. As a corollary we show the the Frobenius norm of the sparse pseudoinverse $\ginv{1}{\mA}$ is indeed controlled.

The proof uses the \emph{convex Gaussian min-max theorem} \cite{Thrampoulidis:2016vo,Oymak:2013vm,Thrampoulidis:2015vf} which has thus far found major use in quantifying the performance of regularized M-estimators such as the lasso \cite{Thrampoulidis:2016vo}. Unlike these previous applications, we give finite-size concentration bounds rather than ``in probability'' asymptotic results.

Along the way we collect a number of useful byproducts related to matrix analysis, sparse representations, and convex programming. For example, we prove that basis pursuit generically has a unique minimizer---a folklore fact mentioned repeatedly in the literature, but of which we could find no proof. The same goes for the fact that $\ell_p$ minimization for $0 \leq p \leq 1$ has a sparse solution which is to the best of our knowledge yet another piece of thus far unproven folklore.


\section{A Sparse Pseudoinverse is (Generically) Unique and Optimally Sparse}

\label{sec:spinv_is_unique_and_optimal}

As the section title promises, we first settle the questions of uniqueness and non-zero count of a ``sparse'' pseudoinverse $\spinv(\mA)$. We rely on two groups of facts: first, existing results on sparse dual frames which characterize maximal levels of sparsity attainable by generalized inverses \cite{Krahmer:2012cy,Casazza:ev,Li:2013cx}, and second, new results on constrained $\ell^p$ minimization which we put forward in Section \ref{sub:vector_facts}.

Our main result in this section can be informally stated as
\begin{theorem}[Informal]
For a generic $\mA \in \R^{m \times n}$ with $m < n$, $\spinv(\mA)$ is unique, $\spinv(\mA) \in \ginv{0}{\mA}$, and it has exactly $m$ non-zeros per column.
\end{theorem}
We begin by proving two preliminary results about vectors.

\subsection{Vector Facts}
\label{sub:vector_facts}

First, we show that the set of $\ell^{p}$ minimizers for $0 \leq p \leq 1$ always contains an $m$-sparse point. While this is known folklore, we could only find a proof for $p=1$ \cite[Section 1.4]{Elad:2010wo} so we give one for $0 \leq p \leq 1$ below, courtesy of Simon Foucart:
\begin{lemma}
  \label{lem:zeros_in_l1}
  Consider $0 \leq p \leq 1$ and let $\mA \in \R^{m \times n}$ with $m<n$ and $\vy \in \range(\mA)$. Then the
  set of minimizers of $\norm{\vz}_{p}$ subject to the
  constraint that $\mA \vz = \vy$ always contains a point with at most $m$ non-zero entries.
\end{lemma}
\begin{proof} 
Let $\mathcal{Z}(\vy)$ be the (non-empty, closed) set of all solutions to $\vy=\mA \vz$ that reach the minimum $\ell^{p}$ norm solution. Let $\vz \in \mathcal{Z}(\vy)$ be an element of this set reaching the smallest $\ell^{0}$ norm: $\norm{\vz}_{0} \leq \norm{\vz'}_{0}$ for all $\vz' \in \mathcal{Z}(\vy)$.
For the sake of contradiction, assume that $\norm{\vz}_{0} > m$ and denote by $S$ the support of $\vz$. By assumption the size of $S$ exceeds $m$, the number of rows of $\mA$, hence the columns of $\mA$ indexed by $S$ are linearly dependent and there exists a nonzero vector $\vh$ of the null space of $\mA$ that is supported in $S$. Considering $\vz' = \vz+\epsilon \vh$ we have $\mA \vz'  = \vy$. 
There is a critical $\epsilon_{o}>0$ and $s \in \{\pm 1\}$ such that: $\sign(\vz+\epsilon \vh) = \sign(\vz)$ when $\abs{\epsilon} < \epsilon_{o}$; and $\norm{\vz + s\epsilon_{o} \vh}_{0} < \norm{\vz}_{0}$ (corresponding to the cancellation of at least one component of $\vz$ by adding $s \epsilon_{o} \vh$). In the remainder of the proof we take $s=1$ with no loss of generality.

For $p = 0$, the contradiction follows immediately. For $0 < p \leq 1$, we let $q \bydef 1/p$ and $q^{*} = 1/(1-p)$ so that $1/q+1/q^*=1$. Then by H{\"o}lder's inequality, when $\abs{\epsilon} < \epsilon_{o}$
\begin{eqnarray*}
 \norm{\vz'}_{p}^{p} 
 &=& \sum_{j \in S} \tfrac{\abs{z'_{j}}^{p}}{\abs{z_{j}}^{p(1-p)}} \cdot \abs{z_{j}}^{p(1-p)}
 \leq 
\left(\sum_{j \in S} \left(\tfrac{\abs{z'_{j}}^{p}}{\abs{z_{j}}^{p(1-p)}}\right)^{q}\right)^{1/q} \cdot \left(\sum_{j \in S} \left(\abs{z_{j}}^{p(1-p)}\right)^{q^*}\right)^{1/q^*}\\
&=&
\left(\sum_{j \in S} \tfrac{\abs{z'_{j}}}{\abs{z_{j}}^{1-p}}\right)^{p} \cdot \left(\sum_{j \in S} \abs{z_{j}}^{p}\right)^{1-p}
=
\left( \sum_{j \in S} \sign(z_{j}+\epsilon h_{j})\,\tfrac{z_{j}+\epsilon h_{j}}{\abs{z_{j}}^{1-p}}\right)^{p} \cdot \norm{\vz}_{p}^{p(1-p)}\\
&=&
\left( \sum_{j \in S} \sign(z_{j})\,\tfrac{z_{j}+\epsilon h_{j}}{\abs{z_{j}}^{1-p}}\right)^{p} \cdot \norm{\vz}_{p}^{p(1-p)}
=
\Big( \norm{\vz}_{p}^{p} + \epsilon \underbrace{\sum_{j \in S} \sign(z_{j})\,\tfrac{h_{j}}{\abs{z_{j}}^{1-p}}}_{\bydef \sigma}\Big)^{p} \cdot \norm{\vz}_{p}^{p(1-p)}
\end{eqnarray*}
Hence

\begin{itemize}
  \item if $\sigma \neq 0$,
  choosing $\abs{\epsilon} < \epsilon_{o}$ with $\sign(\epsilon) = -\sign(\sigma)$ 
  we obtain $\norm{\vz'}_p^{p} < \norm{\vz}_p^{p}$; a contradiction.
  \item if $\sigma = 0$,
  choosing $\epsilon = s \epsilon_{o}$, by continuity we get $\norm{\vz'}_{p}^{p} \leq \norm{\vz}_{p}^{p}$ (hence $\vz' \in \mathcal{Z}(\vy)$), yet $\norm{\vz'}_{0} < \norm{\vz}_{0}$; another contradiction.
\end{itemize}
\end{proof}

For the second fact, related to uniqueness, we need the notion of general position. To this end we quote two definitions from~\cite{Tibshirani:2013ce}.
\begin{definition}
  The $m \times n$ matrix $\mA \in \R^{m \times n}$ has columns {\em in general position} if for any $k < m$ the affine span of any $k+1$ vectors $\pm \va_{i_{1}}, \ldots, \pm \va_{i_{k+1}}$ does not contain any vector $\pm \va_{i}$, $i \notin \{i_{1},\ldots,i_{k+1}\}$.
\end{definition}
\begin{definition}
  The $m \times n$ matrix $\mA \in \R^{m \times n}$ is {\em in general position} with respect to the vector $\vy \in \R^{m}$ if for any $k < m$, $\vy$ does not belong to the linear span of any $k$ columns of $\mA$.
\end{definition}

We now show that for a generic $\mA \in \R^{m \times n}$ and $\vy \in \R^{m}$, the problem
\begin{equation}
    \label{eq:l1exact}
    \min_{\vx:\mA \vx = \vy} \norm{\vx}_1
\end{equation}
has a unique minimizer. Similarly to Lemma \ref{lem:zeros_in_l1}, this fact seems to be a piece of folklore. But although it has been mentioned in  the literature (e.g. \cite{Dossal:2012cd}) we could find no
proof, so we produce one based on the properties of the lasso.

\begin{lemma} \label{le:uniqueL1}
  Assume that $\mA \in \R^{m \times n}$ has columns in general position, and that $\mA$ is in general position with respect to $\vy \in \R^{m}$.     Then~\eqref{eq:l1exact} has a unique solution.
\end{lemma}

\begin{proof}[Proof of Lemma~\ref{le:uniqueL1}]
  The proof is a combination of a sufficient condition for uniqueness of the squared lasso \cite{Fuchs:2004dd, Tibshirani:2013ce}:
  \begin{equation}
    \label{eq:squarelasso}
    \min_{\vx} ~ \tfrac{1}{2} \norm{\vy - \mA \vx}_2^2 + \lambda \norm{\vx}_1
  \end{equation}
  with a technique to map solutions of \eqref{eq:l1exact} to those of
  \eqref{eq:squarelasso} for small $\lambda$ \cite{Fuchs:2004dd}. We start by writing out the KKT optimality conditions for the problem
  \eqref{eq:squarelasso}:
  \begin{equation}
    \label{eq:lassokkt}
    \mA^\T(\vy - \mA \vx^*) = \lambda \vh,
    ~\text{where}~
    h_i \in
    \begin{cases}
      \set{\sign(x^*_i)}, &x_i^* \neq 0, \\
      [-1, 1], &x_i^* = 0
    \end{cases}.
  \end{equation}
  Next, the so-called equicorrelation set $I = I(\vy,\mA,\lambda)$ and equicorrelation sign $\vs = \vs(\vy,\mA,\lambda)$ are defined as
  \begin{equation}
    I 
    \bydef \set{i : \abs{\va_i^\T(\vy - \mA \vx^*)} = \lambda},
    \qquad
    \vs \bydef \sign(\mA^\T_{I}(\vy-\mA \vx^{*})). 
  \end{equation}
  The set $I$ and the vector $\vs$ are unique in the sense that all solutions $\vx^{*}$ of \eqref{eq:squarelasso} have the same equicorrelation set and equicorrelation sign. This follows from the strict convexity of the quadratic loss which, combined with the convexity of the $\ell^{1}$ norm implies (see, e.g.,  \cite[Lemma 1]{Tibshirani:2013ce}) that all solutions lead to the same $\mA \vx^{*}$. From these definitions and \eqref{eq:lassokkt} we see that the solutions $\vx^*$ of \eqref{eq:squarelasso} are characterized by
  \begin{eqnarray}
    \label{eq:lassozerocompl}
    \vx_{I^c}^* &=& \vzero,\\
    \label{eq:lassozero}
    \mA_{I}^\T(\vy-\mA_{I}\vx_{I}^*) &=& \lambda \vs.
  \end{eqnarray}
  Since $\mA$ has columns in general position, as shown in \cite[Lemma 3]{Tibshirani:2013ce}, the matrix $\mA_{I}$ has a trivial nullspace for any $\vy$ and $\lambda>0$, hence 
  \begin{equation}
    |I| \leq m. \label{eq:CardEquicorrelationSet}
  \end{equation}
  It follows that any solution $\vx^*$ of \eqref{eq:squarelasso} can be written as 
  \begin{equation}
    \label{eq:solsquarelasso}
    \vx^*_I = \mA_{I}^\dagger [\vy - \lambda (\mA_{I}^\T)^\dagger \vs ], 
  \end{equation}
  and  the solution  is indeed unique, with $\vs = \sign(\vx_{I}^{*})$. 

  Note that the solution~\eqref{eq:solsquarelasso} of~\eqref{eq:squarelasso} in general does not satisfy $\mA \vx^* = \vy$. We now use a construction by Fuchs \cite[Theorem 2]{Fuchs:2004dd} to relate solutions of~\eqref{eq:squarelasso} to those of~\eqref{eq:l1exact}. 

  Denote $\mathcal{X}$ the convex set of solutions of~\eqref{eq:l1exact}. It is classical exercise in convex analysis \cite[Section 1.4]{Elad:2010wo} to show that this set contains at least one vector whose support $J$ is associated to linearly independent columns of $\mA$, i.e., $\mA_{J}$ has a trivial null space (see also Lemma~\ref{lem:zeros_in_l1} applied to $p=1$). Hence, the set \[\mathcal{J} \bydef \{J = \text{support}(\vx^{\sharp}):\ \vx^{\sharp} \in \mathcal{X},\ \nullspace(\mA_{J}) = \{0\}\}
  \]
  is  non empty. Moreover, if \eqref{eq:l1exact} has a non-unique set of minimizers (i.e., if $\mathcal{X}$ is {\em not} a singleton), then $\mathcal{J}$ contains at least {\em two} distinct sets $J_{1} \neq J_{2}$. To see this, first observe that 
  $\mathcal{X}$ is convex and compact so it is the convex hull of its extreme points.
  Let $\vx_e$ be an element of $\mathcal{X}$ with support $J$, and assume $\mA_J$ has a non-trivial nullspace. Similarly as in Lemma \ref{lem:zeros_in_l1}, there exists a vector $\vh \in \nullspace(\mA)$ entirely supported in $J$. Then, for any $\epsilon$, $\vx_e + \epsilon \vh$ is a feasible point for the optimization \eqref{eq:l1exact} and there exists a critical $\epsilon_0 > 0$ such that for all $\abs{\epsilon} < \epsilon_0$, $\sign(\vx_e + \epsilon \vh) = \sign(\vx_e) \bydef \vs$. For $\abs{\epsilon} < \epsilon_0$, reasoning as in Lemma \ref{lem:zeros_in_l1} we can compute   $\norm{\vx_e + \epsilon \vh}_1 = \norm{\vx_e}_1 + \epsilon \sum_{i=1}^n s_i h_{i}$.
  If $\sigma \bydef \sum_{i=1}^n s_i h_{i} \neq 0$ then by setting $\epsilon = - \sign(\sigma) \epsilon_0 / 2$  we get $\norm{\vx_e + \epsilon \vh}_1 = \norm{\vx_e}_1 - \tfrac{1}{2} \epsilon_0 \abs{\sigma} < \norm{\vx_e}_1$ which is a contradiction since $\vx_e + \epsilon \vh$ is feasible.  It therefore must hold that $\sigma = 0$, meaning that $\norm{\vx_e + \epsilon \vh}_1 = \norm{\vx_e - \epsilon \vh}_1 = \norm{\vx_e}_1$ when $\abs{\epsilon} < \epsilon_{0}$, or in other words, $(\vx_e \pm \epsilon \vh) \in \mathcal{X}$. As 
  \[
    \vx_{e} =  \tfrac{1}{2} (\vx_e + \epsilon \vx_k) + \tfrac{1}{2} (\vx_e - \epsilon \vx_k),
  \]
  the element $\vx_e$ can be written as a convex combination (with non-zero coefficients) of two distinct points in $\mathcal{X}$ which shows that it is not an extreme point of $\mathcal{K}$. By contraposition, this establishes that any extreme point $\vx$ of $\mathcal{X}$ has a support $J(\vx)$ such that $\mA_{J(\vx)}$ has full column rank, that is, $J(\vx) \in \mathcal{J}$. 

  If $\mathcal{J}$ is a singleton with one element $J$ then $\mathcal{X}$ has a single extreme point with non-zero entries $\mA_J^\dag \vy$. Consequently $\mathcal{X}$ is a singleton as the convex hull of that single extreme point. By contraposition, if $\mathcal{X}$ is not a singleton, it has multiple extreme points and $\mathcal{J}$ is not a singleton.

  Consider any $J \in \mathcal{J}$ and $\vx^\sharp$ the corresponding solution of \eqref{eq:l1exact}, characterized by $J  = \text{support}(\vx^{\sharp})$, $\vx^{\sharp}_{J} = \mA_{J}^{\dagger} \vy$ and $\vx^{\sharp}_{J^{c}} = \vzero$. For any $\lambda >0$ we construct $\vx^\star$ (note the $\star$ instead of $*$) as follows:
  \begin{equation}
    \label{eq:lassomapping}
    \vx^\star_J \bydef \vx^{\sharp}_J - \lambda (\mA_J^\T \mA_J)^{-1} \sign(\vx^{\sharp}_J),~
    \vx^\star_{J^c} \bydef \vzero
  \end{equation}
  There is a critical value $\lambda_{J} > 0$ such that for $\lambda < \lambda_{J}$ we have $\sign(\vx^\star_J) = \sign(\vx^{\sharp}_J)$, so that 
  \begin{eqnarray*}
    \mA_{J}^\T(\vy-\mA\vx^{\star})
    &=& \mA_{J}^\T(\mA_{J} \vx^{\sharp}_{J}-\mA_{J} \vx^{\star}_{J}) = \lambda \mA_{J}^\T\mA_{J}(\mA_J^\T \mA_J)^{-1} \sign(\vx^{\sharp}_J) 
    = \lambda \sign(\vx^{\sharp}_J)
    = \lambda \sign(\vx^{\star}_J).
  \end{eqnarray*}
  This shows that $J \subseteq I(\vy,\mA,\lambda)$. Since $\mA$ is in general position with respect to $\vy$ and $\vy \in \range(\mA_{J})$, we have $|J| \geq m$. We also know that, since $\mA$ has columns in general position, $|I| \leq m$ by~\eqref{eq:CardEquicorrelationSet}. Hence $I = J$. This shows that $\mathcal{J}$ is made of a {\em single} support set $J$, which matches the equicorrelation set $I(\vy,\mA,\lambda)$ for any small enough $\lambda$. As a consequence, $\mathcal{X}$ is a singleton and~\eqref{eq:l1exact} has a unique solution.
\end{proof}

\subsection{Generic Uniqueness and Optimal Sparsity of Spinv}

In the previous subsection we proved two facts about (vector) $\ell^1$ minimization: that the optimizer is generically sparse and that it is unique. We now combine these results with known facts about optimal sparsity of dual frames. Let us first state a simple matrix consequence of Lemma \ref{lem:zeros_in_l1}:

\begin{lemma} 
  \label{le:zeros_in_spinv}
  Let $\mA \in \R^{m \times n}$, $m<n$, and $0 \leq p \leq 1$. Then there exists $\mX \in
  \ginv{p}{\mA}$ with at most $m^{2}$ nonzero entries (more precisely, with at
  most $m$ nonzero elements per column). In particular, there is such an $\mX \in \spinv(\mA) = \ginv{1}{\mA}$.
\end{lemma}
\begin{remark}
  One consequence of Lemma~\ref{le:zeros_in_spinv} is that $\spinv(\mA)$ is in general {\em not} the MPP $\mA^{\dagger}$.
\end{remark}

\begin{proof}
Minimization for $\ginv{p}{\mA}$ can be decoupled into $\ell^p$ minimizations
for every column,
\begin{equation}
  \label{eq:spinv_columnwise}
  \ginv{p}{\mA}_j = \argmin_{\mA \vx = \ve_j} \, \norm{\vx}_{p}.
\end{equation}
By Lemma \ref{lem:zeros_in_l1}, the set of minimizers for every column of
$\ginv{p}{\mA}$ (the set of solutions to \eqref{eq:spinv_columnwise}) contains a point with at most 
$m$ nonzeros. Because $\ginv{p}{\mA}$ has $m$ columns, there exists an element of $\ginv{p}{\mA}$ with at most $m^{2}$ non-zero entries.
\end{proof}

A corollary of Lemma~\ref{le:zeros_in_spinv} is that generically,
$\spinv(\mA)$ is {\em the} sparsest pseudoinverse of $\mA$. To see this, we invoke a
result by Krahmer, Kutyniok and Lemvig \cite{Krahmer:2012cy}:
\begin{lemma}[{\cite[Theorem 3.6]{Krahmer:2012cy}}]
    \label{lem:krahmer_measure_zero}
    Let $\mathcal{F}(m, n)$ be the set of full rank matrices in $\R^{m \times
    n}$, and denote by $\mathcal{N}(m, n)$ the subset of matrices in
    $\mathcal{F}(m, n)$ whose sparsest generalized inverse has $m^2$
    non-zeros. Then 
    \begin{enumerate}[label=(\roman*)]
        \item Any matrix in $\mathcal{F}(m, n)$ is arbitrarily close to a
        matrix in $\mathcal{N}(m, n)$,
        \item The set $\mathcal{F}(m, n) \ \setminus \ \mathcal{N}(m, n)$ has
        measure zero.
    \end{enumerate}
\end{lemma}
In particular, the sparsest generalized inverse of many random matrices (e.g., iid Gaussian) will have $m^2$ non-zeros with probability 1.

Lemmata~\ref{le:zeros_in_spinv} and \ref{lem:krahmer_measure_zero} imply that for almost all matrices the set of sparsest generalized inverses intersects the set of $\ell^1$-norm minimizing generalized inverses:

\begin{theorem}\label{th:spinvmsparseprevious}
    The set of matrices $\mA$ in $\mathcal{F}(m, n)$ such that $\spinv(\mA)$ does not
    intersect $\ginv{0}{\mA}$ has measure zero.
\end{theorem}

With all these results in hand we can prove a stronger claim which is our main result. Namely, the $\spinv$ is generically unique and $m^2$-sparse.
\begin{theorem} \label{th:spinvmsparse}
Assume that $\mA \in \R^{m \times n}$ has columns in general position, and that $\mA$ is in general position with respect to the canonical basis vectors $\ve_{1},\ldots,\ve_{m}$. Then the sparse pseudoinverse $\spinv(\mA)$ of $\mA$ contains a single matrix whose columns are all exactly $m$-sparse.
\end{theorem}
\begin{corollary}\label{th:spinvmsparsegen}
Except for a measure-zero set of matrices, the sparse pseudoinverse $\spinv(\mA)$ of $\mA \in \R^{m \times n}$ contains a single matrix whose columns are all exactly $m$-sparse.
\end{corollary}

\begin{proof}[Proof of Theorem~\ref{th:spinvmsparse}]
  To finish the proof of Theorem \ref{th:spinvmsparse}, we apply Lemma \ref{lem:zeros_in_l1} and Lemma \ref{le:uniqueL1} to the sparse pseudoinverse by setting $\vy = \ve_i$ to obtain $\vx_i$, $i \leq 1 \leq m$.
\end{proof}
\begin{proof}[Proof of Corollary~\ref{th:spinvmsparsegen}]
  When $\mA$ is a random matrix drawn from a probability distribution which is absolutely continuous with respect to the Lebesgue measure, we have with probability one that: a) its columns are in general position; b) it is in general position with respect to any finite collection of vectors $\{\vy_{1},\ldots,\vy_{\ell}\}$. We can thus apply Theorem~\ref{th:spinvmsparse}.
\end{proof}

\section{Numerical Stability}
\label{sub:stability}

\begin{figure}[t!]
\centering
\includegraphics{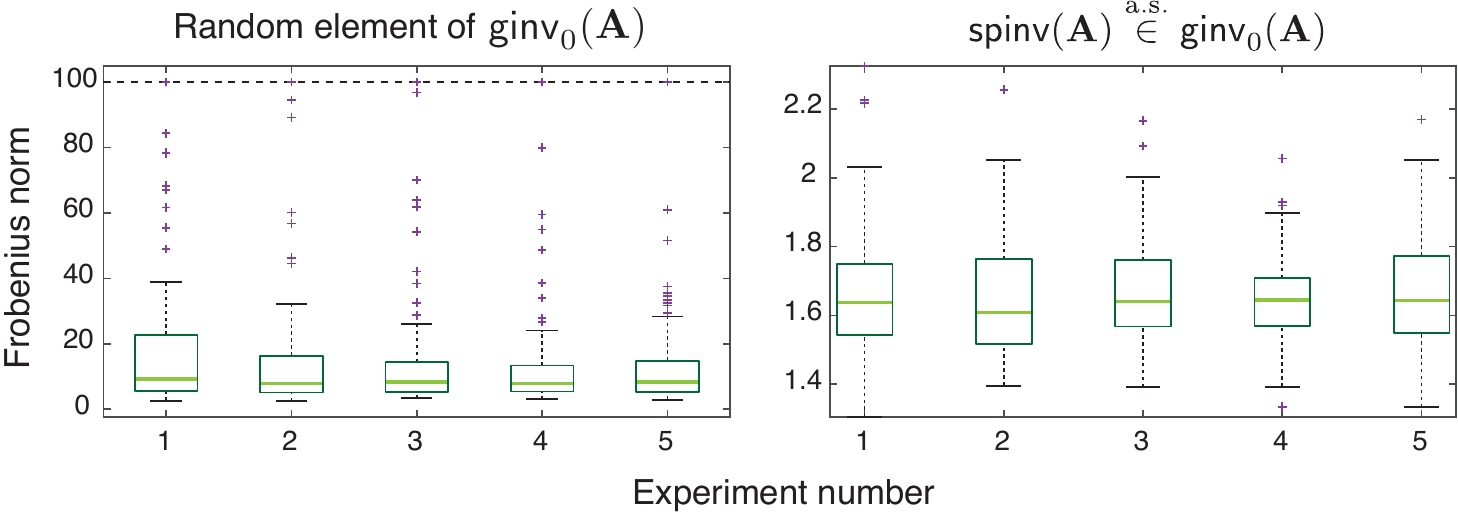}
\caption{Frobenius norm of a random element of $\ginv{0}{\mA}$ for $\mA$ a random Gaussian matrix of size $20 \times 30$ and a sparse pseudoinverse $\spinv(\mA)$ (generically in $\ginv{0}{\mA}$; cf. Corollary \ref{th:spinvmsparsegen}) of the same matrix. The random element of $\ginv{0}{\mA}$ is computed by selecting $m$ columns of $\mA$ at random and inverting the obtained submatrix. The plot shows results of 5 identical experiments, each consisting of generating a random $20 \times 30$ matrix and computing the two inverses 100 times. In the first experiment, the outlier norms extend to 2000 so they were clipped at 100. Green lines denote medians, and boxes denote the second and the third quartile.}
\label{fig:minor}
\end{figure}

For matrices in general position (that is, for most matrices), there is a simpler way to obtain a generalized inverse with the minimal number $m^2$ of non-zeroes than solving for $\spinv$---just invert any full-rank $m \times m$ submatrix of $\mA$. But there is a good reason why minimizing the $\ell^1$ norm is a better idea than inverting a submatrix: it gives much better conditioned matrices.

To understand precisely what we mean by better conditioned, consider an overdetermined inverse problem $\vy = \mB \vx + \vz$, where $\mB \in \C^{n \times m}$ is full rank with $m < n$, and $\vz$ is white noise. For any matrix $\mW \in \ginvset(\mB)$ we have that
\begin{equation}
    \E_{z}[\norm{\mW \vy - \vx}_2^2] = \E_{\vz}[\norm{\mW \vz}_2^2] \, \propto\, \norm{\mW}_F^2. 
\end{equation}
Thus the influence of noise on the output is controlled by the Frobenius norm of $\mW$ which makes it clear that it is desirable to use generalized inverses with small Frobenius norms. The right figure of merit is how much larger the Frobenius norm of our generalized inverse is than that of the MPP which attains the smallest one.

Alas, Figure \ref{fig:minor} shows that for a simple inversion of an $m \times m$ submatrix, this Frobenius norm can be quite large. Worse, as figure shows, the variance of the norm of this inverted minor is large (note that the ordinate axis was clipped). On the contrary, if we carefully select a particular element of $\ginv{0}{\mA}$, namely the sparse pseudoinverse, then it appears that we get a well-controlled Frobenius norm; in other words we get a well-conditioned generalized inverse. 
The goal of this section is to make the last statement rigorous by developing concentration results for the Frobenius norm of $\spinv(\mA)$ when $\mA$ is an iid Gaussian random matrix. In particular, we will prove the following result:

\begin{theorem}
    \label{thm:frob_of_l1}
    Let $\mA \in \R^{m \times n}$ be a standard iid Gaussian matrix, $1 \leq m \leq n$. Define
    $\delta \bydef (m-1) / n \in (0, 1)$ and for $1 \leq p \leq 2$ define the function
    \begin{eqnarray}
    \label{eq:DefExpSquaredDistance}
    D(t) = D_{p}(t;n) &\bydef& \tfrac{1}{n} \left[\E_{\vh \sim {\cal N}(0, \mI_n)} \,
    \dist(\vh, \norm{\,\cdot\,}_{p^{*}} \leq t)\right]^{2} \in (0, 1),
    \end{eqnarray}
    Let $t^* = t^{*}_{p}(\delta,n)$  be the unique solution of $\delta = D(t) -\tfrac{t}{2} D'(t)$ on $(0,\infty)$ and denote 
    \[
    \alpha^* = \alpha^*_{p}(\delta, n) \bydef \sqrt{\frac{D(t^{*})}{\delta(\delta - D(t^{*}))}}. 
    \]
    If there exist $\gamma(\delta)$ and $ N(\delta)$ such that $-t^{*}D'_{p}(t^{*};n) \geq \gamma(\delta) >0$ for all $n \geq N(\delta)$, then for any $n \geq \max(2/(1-\delta),N(\delta))$ we have: for any $0<\epsilon \leq 1$
    \begin{equation}
        \label{eq:stability_main_bound}
        \prob \left[ \abs{ \tfrac{n}{m}\norm{\ginv{p}{\mA}}_F^2 -  (\alpha^*)^2} \geq \epsilon(\alpha^*)^2 \right] \leq \   \frac{n}{C_{1}\epsilon}\e^{-C_{2}n\epsilon^{4}},
    \end{equation}
    where the constants $C_{1},C_{2}>0$ may depend on $\delta$ but not on $n$ or $\epsilon$.
\end{theorem}

The statement of Theorem \ref{thm:frob_of_l1} may be difficult to parse at first glance. Luckily, it allows us to obtain more explicit results for the two most interesting cases: $p=1$ and $p=2$. For $p=2$ we get a result on the Frobenius norm of the Moore-Penrose pseudoinverse. It  is complementary to a known large deviation bound \cite[Proposition A.5; Theorem A.6]{Halko:2011kg} obtained by a completely different technique.
\begin{corollary}[$p=2$]\label{cor:P2}
With the notations of the above theorem, for $0<\delta<1$ and $n \geq N(\delta)$ we have: for any $0<\epsilon \leq 1$
\begin{equation}
        \prob \left[ \abs{ \tfrac{n}{m}\norm{\mA^{\dagger}}_F^2 -  (\alpha^*)^2} \geq \epsilon(\alpha^*)^2 \right] \leq \   \frac{n}{C_{1}\epsilon}\e^{-C_{2}n\epsilon^{4}},
    \end{equation}
    where the constants $C_{1},C_{2}>0$ may depend on $\delta$ but not on $n$ or $\epsilon$, and $\alpha^{*}=\alpha^{*}_{2}(\delta;n)$ with
    \begin{equation}
    \lim_{n \to \infty} \alpha_{2}^{*}(\delta;n) = \frac{1}{\sqrt{1-\delta}}.
    \end{equation}
\end{corollary}

\begin{remark}
This corollary covers ``small'' deviations $(0<\epsilon \leq 1)$. In contrast, the result of \cite{Halko:2011kg} establishes that $\mathbb{E}\ \tfrac{n}{m}\norm{\mA^{\dagger}}_F^2 = \tfrac{n}{n-m-1} = \tfrac{1}{1-\delta}$ and that for any $\tau \geq 1$,
\[
\prob \left[ \tfrac{n}{m}\norm{\mA^{\dagger}}_F^2 \geq \tfrac{12 n}{n-m} \tau\right] \leq 4 \tau^{-(n-m)/4} = 4 e^{-n \tfrac{(1-m/n) \log \tau}{4}}.
\]
For large $n$ we have $\tfrac{12 n}{n-m} \tau \approx \tfrac{12\tau}{1-\delta} \approx 12\tau (\alpha^{*}_{2})^{2}$ and $\tfrac{(1-m/n) \log \tau}{4} \approx \tfrac{(1-\delta) \log \tau}{4}$, hence this provides a bound for $1+\epsilon := 12 \tau \geq 12$ with an exponent $\log \tau$, of the order of $\log \epsilon$ (instead of $\epsilon^{4}$ we get for $0<\epsilon \leq 1$). Furthermore, unlike  \cite{Halko:2011kg}, we give a two-sided bound; that is, we also show that the probability of $\tfrac{n}{m}\norm{\mA^\dag}_F^2$ being much smaller than $(\alpha^*)^2$ is exponentially small.
\end{remark}

The most interesting corollary is for $p = 1$. 

\begin{corollary}[$p=1$]\label{cor:P1}
With the notations of the above theorem, for $0<\delta<1$ and $n \geq N(\delta)$ we have: for any $0<\epsilon \leq 1$
\begin{equation}
        \prob \left[ \abs{ \tfrac{n}{m}\norm{\spinv(\mA)}_F^2 -  (\alpha^*)^2} \geq \epsilon(\alpha^*)^2 \right] \leq \     \frac{n}{C_{1}\epsilon}\e^{-C_{2}n\epsilon^{4}},
    \end{equation}
    where the constants $C_{1},C_{2}>0$ may depend on $\delta$ but not on $n$ or $\epsilon$, and $\alpha^{*}=\alpha^{*}_{1}(\delta;n)$ with
    \begin{equation}
    \lim_{n \to \infty} \alpha_{1}^{*}(\delta;n)
    = \sqrt{\frac{1}{\sqrt{\tfrac{2}{\pi}} \e^{-\tfrac{(t_1^*)^2}{2}} t^{*}_{1} - \delta(t_1^*)^2} -\tfrac{1}{\delta}}
    \end{equation}
    $t^*_{1} = \sqrt{2} \cdot \erfc^{-1}(\delta)$.
\end{corollary}

Results of Corollaries \ref{cor:P2} and \ref{cor:P1} are illustrated\footnote{For reproducible research, code is available online at \url{https://github.com/doksa/altginv}.} in Figures \ref{fig:frobinv1} and \ref{fig:frobinv2}. Figure \ref{fig:frobinv1} shows the shape of the limiting $(\alpha_{p}^*)^{2}$ as a function of $\delta$, as well as empirical averages of 10 realization for different values of $n$ and $\delta$. As expected, the limiting values get closer to the empirical result as $n$ grows larger; for $n=1000$ the agreement is near-perfect.

In Figure \ref{fig:frobinv1} we also show the individual realizations for different combinations of $n$ and $\delta$. As predicted by the two corollaries, the variance of both $\tfrac{n}{m}\norm{\mA^\dag}_F^2$ and $\tfrac{n}{m}\norm{\spinv(\mA)}_F^2$ reduces with $n$. For larger values of $n$, all realizations are very close to the limiting $(\alpha_{p}^*)^{2}$.

It is worth noting that the $\spinv$ and the MPP exhibit rather different qualitative behaviors. The Frobenius norm of the MPP monotonically decreases as $\delta$ gets smaller, while that of the $\spinv$ turns up below some critical $\delta$. This may be understood by noting that for small $\delta$, the support of the $\spinv$ is concentrated on very few entries, so those have to be comparably large to produce the ones in the identity matrix $\mI = \mA \mX$. However, the $\ell^2$ norm is more punishing for large entries.

\begin{figure}[h!]
\centering
\includegraphics[width=.7\linewidth]{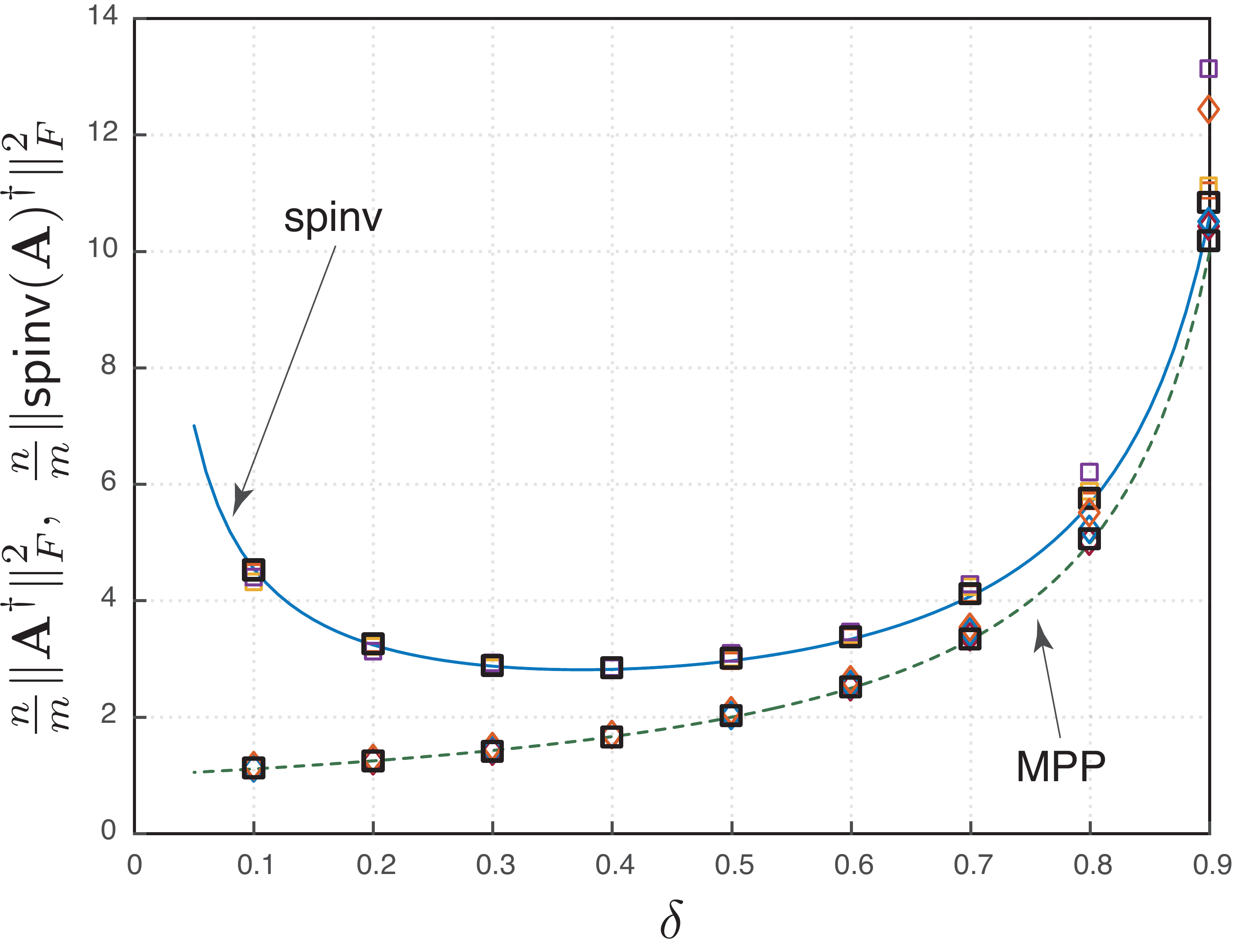}
\caption{Comparison of the limiting $(\alpha_{p}^*)^2$ with the mean of $\tfrac{n}{m}\norm{\mA^\dag}_F^2$ and $\tfrac{n}{m}\norm{\spinv(\mA)}_F^2$ for 10 realizations of $\mA$. Empirical results are given for $\delta \in \set{0.1, 0.2, \ldots, 0.9}$ and $n \in \set{100, 200, 500, 1000}$. Black squares represent the empirical result for $n=1000$; colored squares represent the empirical mean of $\tfrac{n}{m}\norm{\mA^\dag}_F^2$ for $n \in \set{100, 200, 500}$, with the largest discrepancy (purple squares) for $n=100$; colored diamonds represent the empirical mean of $\tfrac{n}{m}\norm{\spinv(\mA)}_F^2$ with the largest discrepancy (orange diamonds) again for $n=100$.}
\label{fig:frobinv1}
\end{figure}

\begin{figure}[h!]
\centering
\includegraphics[width= \linewidth]{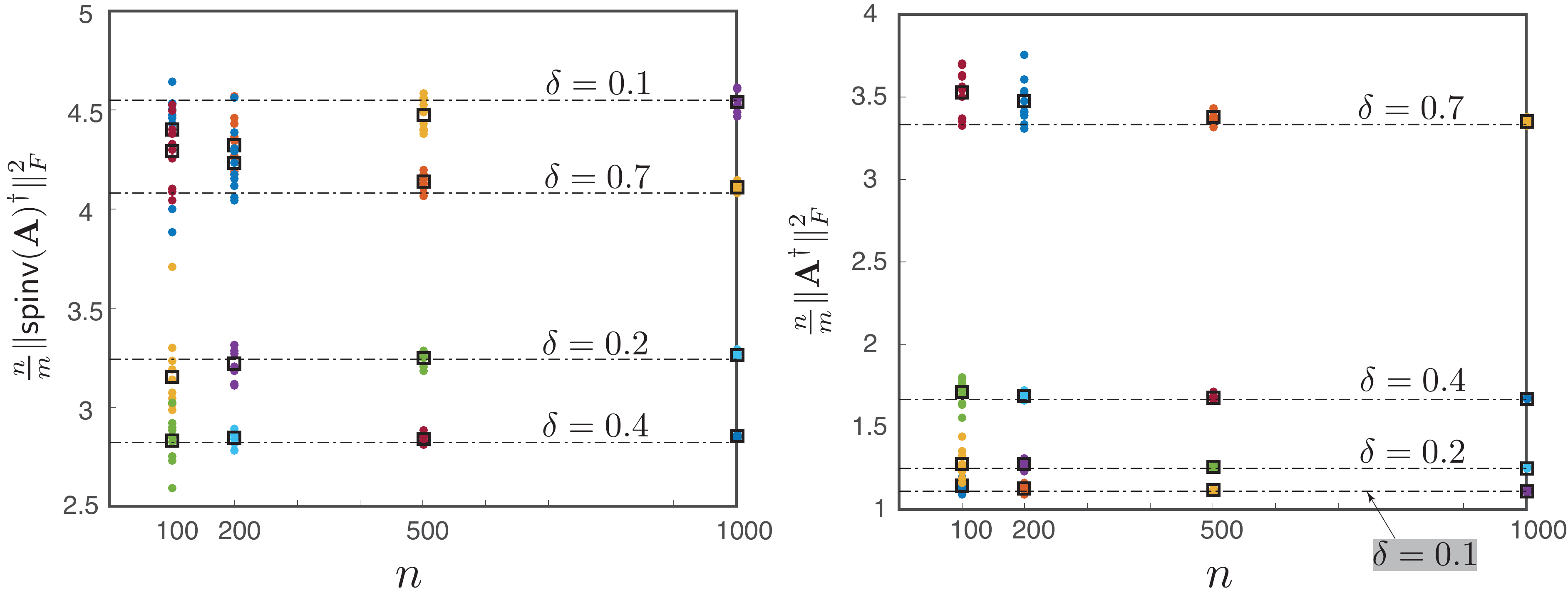}
\caption{Comparison of the limiting $(\alpha_{p}^*)^2$ with the value of  $\tfrac{n}{m}\norm{\spinv(\mA)}_F^2$ (left) and $\tfrac{n}{m}\norm{\mA^\dag}_F^2$ (right) for 10 realizations of $\mA$. Results are shown for $\delta \in \set{0.1, 0.2, 0.4, 0.7}$ and different values of $n$. Values for individual realizations are shown with colored dots with different color for every combination of $n$ and $\delta$. Horizontal dashed lines indicate the limiting value for the considered values of $\delta$.}
\label{fig:frobinv2}
\end{figure}

Several remarks are due:
\begin{enumerate}
    \item The bound \eqref{eq:stability_main_bound} and the corresponding bounds in Corollaries \ref{cor:P1} and \ref{cor:P2} involve $\epsilon^4$ instead of the usual $\epsilon^2$ implying a higher variance of $\tfrac{n}{m}\norm{\ginv{p}{\mA}}_F^2$: to guarantee a given probability in the right hand side of \eqref{eq:stability_main_bound}, $\epsilon$ should be of the order $n^{-1/4}$ instead of the usual $n^{-1/2}$. This seems to be a consequence of the technique used to lower bound $\inf_{\abs{\alpha-\alpha^{*}} \geq \epsilon} \kappa(\alpha)-\kappa(\alpha^{*})$ in Lemma \ref{le:ArgMinKappa} which relies on strong convexity of $\kappa$. A more refined bound may give the result with better error bars, albeit also complicate the analysis.
    
    \item Corollaries \ref{cor:P2} and \ref{cor:P1} specialize Theorem \ref{thm:frob_of_l1} to cases $p = 1$ and $p = 2$ proving that the Frobenius norm of the corresponding generalized inverses indeed concentrates, and giving a closed-form limiting value of the optimal $\alpha^*$. It would seem rather natural that an interpolation to $p \in (1, 2)$ is possible, although $\alpha^*_{p}$ would be specified implicitly and have to be computed numerically. Whether an extension to $p > 2$ is possible is less clear. Numerical evidence suggests that it is, but we leave the corresponding theory to future work.

    \item With additional work, one could characterize the rate of convergence of $\alpha^{*}_{p}(\delta;n)$ towards its limit in Corollaries \ref{cor:P2} and \ref{cor:P1}. We leave these characterizations to future work.
\end{enumerate}

\begin{proof}[Proof of Corollaries~\ref{cor:P2}-\ref{cor:P1}]
For $p \in \{1,2\}$, using Lemmas~\ref{le:boundp2}-\ref{le:boundp1} we lower bound  $-t^{*} D'_{p}(t^{*}) \geq \gamma(\delta)$ for all $n$ above some $N(\delta)$, and control $\lim_{n \to \infty} \alpha_{p}^{*}(\delta;n)$.
Applying Theorem~\ref{thm:frob_of_l1} yields the conclusion.
\end{proof}

\subsection{Proof of the Main Concentration Result, Theorem \ref{thm:frob_of_l1}}

We prove Theorem \ref{thm:frob_of_l1} by first noting that the matrix optimization for $\ginv{p}{\mA}$ decouples into $m$ vector optimizations, and then using the following vector result for each column of $\ginv{p}{\mA}$:

\begin{lemma}%
    \label{lem:concentration_column}%
    With notations and assumptions as in Theorem \ref{thm:frob_of_l1}, we have for $0<\epsilon' \leq 1$ and $n \geq \max(2/(1-\delta),N(\delta))$
    \[\prob
    \left[ \abs{ \sqrt{n} \norm{\vx^*} - \alpha^*} \geq \epsilon' \alpha^{*} \right] \leq \
    \frac{1}{K_{1}\epsilon'} \e^{-K_{2} n \epsilon'^{4}}.
    \]
    where $K_{1}, K_{2}>0$ may depend on $\delta$ but not on $n$ or $\epsilon'$.
\end{lemma}

The proof of this result is given in Section \ref{sub:proof_of_concentration}. It is based on the convex Gaussian min-max theorem \cite{Thrampoulidis:2016vo}, cf Appendix \ref{append:CGMT}. However, because the squared Frobenius norm is a sum of $m$ squared $\ell^2$ column norms, an ``in probability'' result on vector norms would not suffice. To address this shortcoming we developed a number of technical lemmas that lead to a stronger concentration result. With Lemma \ref{lem:concentration_column} we now continue working to prove Theorem \ref{thm:frob_of_l1}.

By definition, we have
\begin{equation*}
    \ginv{p}{\mA} = \argmin_{\mA \mX = \mI} \ \norm{\mathrm{vec}(\mX)}_p = \argmin_{\mA \mX = \mI} \ \norm{\mathrm{vec}(\mX)}_p^p,
\end{equation*}
where we assume that the solution is unique (for $p = 1$ this is true with
probability 1, so we can condition on this event). This optimization decouples
over columns of $\mX$: denoting $\mX^* = \ginv{p}{\mA}$ we have for the $i$th
column that $\vx^*_i = \argmin_{\mA \vx = \ve_i} \norm{\vx}_p$. 

Lemma \ref{lem:concentration_column} tells us that $\sqrt{n} \norm{\vx^*_{i}}_2$ remains close to $(\alpha^{*})$. However, to exploit the additivity of the squared Frobenius norms over columns a more useful statement would be that $n \norm{\vx^*_{i}}_2^2$ remains close to $(\alpha^{*})^{2}$. To show that this is indeed the case, we note that $(n \norm{\vx^*}^2 - (\alpha^*)^2) = (\sqrt{n} \norm{\vx^*} - \alpha^*)(\sqrt{n} \norm{\vx^*} + \alpha^*)$ so we can write for any $b > 0$
\[
    \prob \left[ \abs{ n \norm{\vx^*}^2 - (\alpha^*)^2} \geq \epsilon' (\alpha^*)^2 \right]
    \leq
    \prob \left[ \abs{ \sqrt{n} \norm{\vx^*} - (\alpha^*)} \geq \epsilon' (\alpha^*)^2 / b \right]
    + \prob \left[ (\sqrt{n} \norm{\vx^*} + \alpha^*) \geq b \right].
\]
By taking $b = 3 \alpha^*$, we bound the second term as 
\[
    \prob \left[ (\sqrt{n} \norm{\vx^*} + \alpha^*) \geq b \right]
    \leq
    \prob \left[ (\sqrt{n} \norm{\vx^*} - \alpha^*) \geq 1 \alpha^* \right] 
    \leq
    \frac{1}{K_{1}} \e^{-K_{2} n }
\]
using Lemma \ref{lem:concentration_column}. To bound the first term we again use Lemma \ref{lem:concentration_column} to finally obtain
\[
    \prob \left[ \abs{ n \norm{\vx^*}^2 - (\alpha^*)^2} \geq \epsilon (\alpha^*)^2 \right]
    \leq
    \frac{1}{K_{1}} \e^{-K_{2} n } + \frac{3}{K_{1}\epsilon} \e^{-K_{2} n (\epsilon / 3)^4}
    \leq
    \frac{1}{C_{1}\epsilon} \e^{-C_2 n \epsilon^4}
\]
with an appropriate choice of $C_{1},C_{2}$.

This characterizes the squared $\ell^2$ norm of one column of the MPP. The squared Frobenius norm is a sum of $m$ such terms $\sum_{i=1}^m \norm{\vx_i}^2$ which are not independent. Our goal is to show that
\begin{equation*}
    \frac{n}{m} \norm{\mX^*}_F^2 = \frac{n}{m} \sum_{i=1}^m \norm{\vx^*_i}^2 
\end{equation*}
stays close to $(\alpha^*)^{2}$ as well. We work as follows: 

\begin{eqnarray*}
    \prob\left[\abs{\tfrac{n}{m}  \norm{\mX^\star}_F^2 - (\alpha^*)^2} > \epsilon(\alpha^*)^2\right]
    &=& \prob\left[\abs{\sum_{i=1}^m (n\norm{\vx^\star}^2 - (\alpha^*)^2} > m \epsilon(\alpha^*)^2\right]\\ 
    &\leq& \prob\left[\sum_{i=1}^m \abs{n \norm{\vx^\star}^2 - (\alpha^*)^2} > m \epsilon(\alpha^*)^2\right] = (*).
\end{eqnarray*}
Now observe that if the sum of $m$ terms is to be larger than $m \epsilon$,
then at least one term must be larger than $\epsilon$, so we can continue writing:

\begin{equation}
\begin{aligned}
    (*) \leq \prob \left[ \, \exists \, k \in \set{1, \ldots, m}: \abs{n \norm{\vx_k^\star}^2 - (\alpha^*)^2} > \epsilon(\alpha^*)^2 \right]
    &= \prob \left[ \bigcup_{i=1}^m \set{\abs{n \norm{\vx_i^\star}^2 - (\alpha^*)^2} > \epsilon(\alpha^*)^2} \right]\\
    &\leq \sum_{i=1}^m \prob[ \abs{n \norm{\vx_i^\star} - (\alpha^*)^2} > \epsilon(\alpha^*)^2 ]\\
    & \leq m\frac{1}{C_{1}\epsilon} \e^{-C_{2}n\epsilon^{4}}
    \leq n\frac{1}{C_{1}\epsilon} \e^{-C_{2}n\epsilon^{4}},
\end{aligned}
\end{equation}
which completes the proof.

\subsection{Proof of the Main Vector Result, Lemma \ref{lem:concentration_column}} 
\label{sub:proof_of_concentration}

In the Appendix we establish Lemma \ref{lem:lasso_satisfies_equality}\footnote{We note that Lemma \ref{lem:lasso_satisfies_equality} is a simple restatement of \cite[Lemma 9.2]{Oymak:2013vm} with dependence on $\epsilon$ made explicit.} which lets us rewrite the optimization
\begin{equation}
    \label{eq:no_lasso}
    \vx^{*} = \arg\min_{\mA \vx = \ve_1} \norm{\vx}_p
\end{equation}
in an unconstrained form as (note that the data fidelity term is \emph{not} squared here):
\begin{equation}
    \label{eq:lasso}
    \tilde{\vx} = \arg\min_{\vx} \ \norm{\mA \vx - \ve_1} + \lambda \norm{\vx}_p.
\end{equation}
More precisely, since the $\ell^{p}$ norm is $L$-Lipschitz with respect to the $\ell^{2}$ norm (with $L \bydef n^{\max(1/p-1/2,0)}$), Lemma \ref{lem:lasso_satisfies_equality} tells us that if we choose $\lambda \leq \frac{\sqrt{n} - \sqrt{m}}{L} (1 - \epsilon)$, minimizers $\vx^{*}$ and $\tilde{\vx}$ of \eqref{eq:no_lasso} and \eqref{eq:lasso} coincide\footnote{An analogous result does not hold for the squared lasso (except for $\lambda = 0^+$).}  with probability at least $1 - 2 \e^{-\epsilon^2 (\sqrt{n} - \sqrt{m})^2/2}$. Rewriting the $\ell^2$ norm in a variational form, we
obtain
\begin{equation}
 \label{eq:frob_variational}
 \tilde{\vx} = \arg\min_{\vx} \max_{\vu:\norm{\vu} \leq 1} \vu^\T \mA \vx + \lambda \norm{\vx}_p - \vu^\T \ve_1.
\end{equation}
The expression \eqref{eq:frob_variational} is a sum of a bilinear term involving $\mA$ and a convex-concave function\footnote{Convex in the first argument, concave in the second one.} $\psi(\vx, \vu) = \lambda \norm{\vx}_p - \vu^\T \ve_{1}$. That is exactly the structure required by the convex Gaussian min-max theorem \cite[Theorem 3]{Thrampoulidis:2015vf} \cite[Theorem 6.1]{Thrampoulidis:2016vo}. For readers' convenience, we reproduce the version of CGMT we use in Appendix \ref{append:CGMT} as Theorem \ref{thm:cgmt}. The only non-conforming detail is that $\vx$ is not constrained to be in a compact set. To address this, instead of \eqref{eq:frob_variational}, we analyze the following bounded modification:
\begin{equation}
 \label{eq:po_bounded}
 \tilde{\vx}_{K} = \arg\min_{\vx : \norm{\vx} \leq K} \max_{\vu:\norm{\vu} \leq 1} \vu^\T \mA
 \vx + \lambda \norm{\vx}_p - \vu^\T \ve_1
\end{equation}
We will show in due time that $\wt{\vx}_{K} = \wt{\vx}$ with high probability.
%

By the CGMT, instead of analyzing the so-called principal optimization
\eqref{eq:po_bounded}, we can analyze an auxiliary optimization
\begin{equation}
 \label{eq:auxiliary}
 \hat{\vx}_{K} = \arg\min_{\vx : \norm{\vx} \leq K} \max_{\vu:\norm{\vu} \leq 1} \norm{\vx} \vg^\T \vu - \norm{\vu} \vh^\T \vx + \lambda \norm{\vx}_p - \vu^\T \ve_1,
\end{equation}
where $\norm{\, \cdot \,}$ are 2-norms and $\vg \in \R^m$ and $\vh \in \R^n$ are iid standard Gaussian random vectors. Part \ref{thmpart:cgmt_optval_concentrate} of the CGMT tells us that if the optimal value of the auxiliary optimization~\eqref{eq:auxiliary} concentrates (note that \eqref{eq:auxiliary} is a random optimization program), the optimal value of \eqref{eq:po_bounded} will concentrate around the same value. This lets us prove that if the norm $\norm{\hat{\vx}_{K}}$ of the optimizer of \eqref{eq:auxiliary} concentrates, the norm $\norm{\tilde{\vx}_{K}}$ of the optimizer of \eqref{eq:po_bounded} will also concentrate around the same value.\footnote{The referenced CGMT contains a similar statement, albeit we need a different derivation to get exponential concentration.}

We will now formally show how to use this property for our purpose by going through a series of steps to simplify \eqref{eq:auxiliary}. The analysis will be easier if the argument of the optimization is of order 1, so let $\vz$ be an appropriately scaled version of $\vx$, $\vz = \vx \sqrt{n}$ (accordingly $A = K \sqrt{n}$). Using the variational characterization of the $\ell^p$ norm we can rewrite \eqref{eq:auxiliary} as

\begin{equation}
 \label{eq:aux_dual}
\hat{\vz}_{A} =  \arg \min_{\vz : \norm{\vz} \leq A} \max_{\substack{\beta: 0 \leq \beta \leq 1\\ \vu : \norm{\vu} = \beta\\\vw : \norm{\vw}_{p^*} \leq 1}} \frac{1}{\sqrt{n}} \norm{\vz} \vg^\T \vu - \frac{1}{\sqrt{n}} \norm{\vu} \vh^\T \vz - \vu^\T \ve_1 + \frac{\lambda}{\sqrt{n}} \vw^\T \vz.
\end{equation}

Let $\hat{\vz}_{A}$ be a (random) optimizer of \eqref{eq:aux_dual}, and $\check{\vz}_{A}$ a (random) value of $\vz$ at the optimum of the following program with the same cost function but an altered order of minimization and maximization:
\begin{equation}
    \label{eq:aux_dual_swap}
    \max_{\substack{\beta : 0 \leq \beta \leq 1\\\vw:\norm{\vw}_{p^*} \leq 1}} \min_{\vz : \norm{\vz} \leq A}  \max_{\vu:\norm{\vu} = \beta} \frac{1}{\sqrt{n}} \norm{\vz} \vg^\T \vu - \frac{1}{\sqrt{n}} \norm{\vu} \vh^\T \vz 
- \vu^\T \ve_1 + \frac{\lambda}{\sqrt{n}} \vw^\T \vz.
\end{equation}
Even though after performing optimization over $\vu$ in \eqref{eq:aux_dual} we do not get a convex-concave cost function, we will see that we can indeed reorder maximizations as in \eqref{eq:aux_dual_swap} for our purposes. In the following we focus on \eqref{eq:aux_dual_swap} since it is easier to analyze. We compute as follows:

\begin{subequations}
\begin{align}
\label{eq:aux_dual_swap_first_line}
\eqref{eq:aux_dual_swap} 
=
&\max_{\substack{\beta : 0 \leq \beta \leq 1\\\vw : \norm{\vw}_{p^*} \leq 1}} \min_{\vz : \norm{\vz} \leq A}  \beta \bignorm{\frac{\norm{\vz} \vg}{\sqrt{n}} - \ve_1} - \frac{\beta}{\sqrt{n}} \vh^\T \vz + \frac{\lambda}{\sqrt{n}} \vw^\T \vz \\
=
&\max_{\substack{\beta : 0 \leq \beta \leq 1\\\vw : \norm{\vw}_{p^*} \leq 1}} \min_{\alpha : 0 \leq \alpha \leq A} \min_{\vz:\norm{\vz} = \alpha}  \beta \bignorm{\frac{\norm{\vz} \vg}{\sqrt{n}} - \ve_1} - \frac{\beta}{\sqrt{n}} \vh^\T \vz + \frac{\lambda}{\sqrt{n}} \vw^\T \vz \\
=
&\max_{\substack{\beta : 0 \leq \beta \leq 1\\\vw : \norm{\vw}_{p^*} \leq 1}} \min_{\alpha : 0 \leq \alpha \leq A} \beta \norm{\frac{\alpha \vg}{\sqrt{n}} - \ve_1} - \frac{\alpha}{\sqrt{n}} \norm{\beta \vh - \lambda \vw }.
\end{align}
\end{subequations}

The objective in the last line is convex in $\alpha$ and jointly
concave in $(\beta, \vw)$. Additionally, the constraint sets are all convex and
bounded, so by \cite[Corollary 3.3]{Sion:1958jm} we can swap the order of
$\min$ and $\max$ again. We thus continue writing:

\begin{subequations}
\begin{align}
= &\min_{\alpha : 0 \leq \alpha \leq A} \max_{\substack{\beta : 0 \leq \beta \leq 1\\\vw : \norm{\vw}_{p^*} \leq 1}}  \beta \norm{\frac{\alpha \vg}{\sqrt{n}} - \ve_1} - \frac{\alpha}{\sqrt{n}} \norm{\beta \vh - \lambda \vw } \\
= &\min_{\alpha : 0 \leq \alpha \leq A} \max_{\beta : 0 \leq \beta \leq 1} \left( \beta \norm{\frac{\alpha \vg}{\sqrt{n}} - \ve_1} -  \frac{\alpha}{\sqrt{n}} \min_{\vw : \norm{\vw}_{p^*} \leq 1} \norm{\beta \vh - \lambda \vw }\right) \\
= &\min_{\alpha : 0 \leq \alpha \leq A} \max_{\beta : 0 \leq \beta \leq 1}  \beta \norm{\frac{\alpha \vg}{\sqrt{n}} - \ve_1} - \frac{\alpha \beta}{\sqrt{n}} \dist(\beta\vh, \norm{\,\cdot\,}_{p^*} \leq \lambda )\\
\label{eq:opt_two_scalars}
= &\min_{\alpha : 0 \leq \alpha \leq A} \max_{\beta : 0 \leq \beta \leq 1} \phi(\alpha,\beta;\vg,\vh)
\end{align}
\end{subequations}
where
\begin{equation}\label{eq:DefPhi2Scalars}
\phi(\alpha,\beta;\vg,\vh) \bydef  \beta \norm{\frac{\alpha \vg}{\sqrt{n}} - \ve_1} - \frac{\alpha}{\sqrt{n}} \dist(\beta\vh, \norm{\,\cdot\,}_{p^*} \leq \lambda).
\end{equation}

We have achieved two important feats: 1) we simplified a high-dimensional vector optimization~\eqref{eq:lasso} into an optimization over two scalars~\eqref{eq:opt_two_scalars}, 2) one of these scalars, $\check{\alpha} = \norm{\check{\vz}_{A}}$ is almost giving us what we seek---the (scaled) Frobenius norm of $\mX$. To put all pieces together, there now remains to formally prove that the min-max switches and the concentation results we mentioned actually hold.

\paragraph{Combining the ingredients.}

By Lemma~\ref{lem:concentration2}, $\phi(\alpha,\beta; \vg,\vh)$ concentrates around some deterministic function $\kappa(\alpha,\beta) = \kappa_{p}(\alpha,\beta;n,\delta,\lambda)$ with $\delta = (m-1)/n$. By Lemma~\ref{le:ArgMinKappa}-Property 3, the minimizer of this function is given by
\begin{equation}
\argmin_{\alpha : 0 \leq \alpha \leq A} \max_{\beta: 0 \leq \beta \leq 1} \kappa(\alpha,\beta) 
= \alpha^{*} = \alpha^{*}_{p}(\delta;n) \bydef \sqrt{\frac{D_{p}(t^{*}_{p};n)}{\delta(\delta-D_{p}(t^{*}_{p};n))}}
\end{equation}
as soon as $A=K\sqrt{n}>\alpha^{*}$ and $\lambda \leq t^{*}$ (cf. Lemma~\ref{lem:Dp_det_prop}-Property 7 for the definition of $t^{*}=t^{*}_{p}(\delta,n)$). By Lemma~\ref{le:empiricalminmax}, for $0<\epsilon\leq \max(\alpha^{*},A-\alpha^{*})$, the minimizer 
\[
\alpha^{*}_{\phi} \bydef \argmin_{0 \leq \alpha \leq A} \max_{0 \leq \beta \leq 1}\phi(\alpha, \beta; \vg, \vh)
\]
stays $\epsilon$-close to $\alpha^*$ with high probability:
\[
 \prob[\abs{\alpha^{*}_{\phi}-\alpha^{*}} \geq \epsilon] \leq \zeta\big(n,\tfrac{\omega(\epsilon)}{2}\big),
\]
where 
\[
\zeta(n,\xi) = \zeta(n,\xi;A,\delta) \bydef \tfrac{c_{1}}{\xi} \e^{-c_{2} \xi^{2}n \cdot c(\delta,A)},
\qquad\qquad \text{with}\ c(\delta,A) \bydef \tfrac{\min(\delta,A^{-2})}{1+\delta A^{2}}
\]
and $c_{1},c_{2}$ universal constants, and
\[
\omega(\epsilon) = \omega_{p}(\epsilon;n,\delta,\lambda) \bydef \frac{\epsilon^2}{2} \frac{\lambda\delta/t^{*}}{(1 + \delta(\alpha^* + \epsilon)^2)^{3/2}}.
\]
As a consequence, as established in Lemma~\ref{lem:minmaxswap}, the scaled norm $\sqrt{n} \norm{ \tilde{\vx}_{K}}$ of the minimizer of the (bounded) principal optimization problem \eqref{eq:po_bounded} stays $\epsilon$-close to $\alpha^*$ with high probability
\[
        \prob \big[ \abs{\sqrt{n} \norm{\wt{\vx}_K} - \alpha^*} \geq \epsilon \big]  \leq 4 \zeta\big(n, \tfrac{\omega(\epsilon)}{2}\big).
\]
Similarly, for $0<\epsilon\leq \min(\alpha^{*},A-\alpha^{*})$ by invoking Lemma~\ref{lemma:bounded_equals_unbounded} we obtain that the scaled norm $\sqrt{n} \norm{ \tilde{\vx}}$ of the minimizer of the \emph{unbounded} optimization \eqref{eq:frob_variational} stays in the neighborhood of $\alpha^*$ with high probability
\[        
    \prob \big[ \abs{\sqrt{n} \norm{\wt{\vx}} - \alpha^*} \geq \epsilon \big]  \leq 4 \zeta\big(n, \tfrac{\omega(\epsilon)}{2}\big)
\]
and in fact \(    \prob[\wt{\vx} \neq \wt{\vx}_{K}] \leq 4\zeta\big(n,\tfrac{\omega(\epsilon)}{2}\big) \).

Since the $\ell^{p}$ norm is $L$-Lipschitz with respect to the Euclidean metric in $\R^{n}$, with $L = n^{\max(1/p-1/2,0)}$, invoking Lemma \ref{lem:lasso_satisfies_equality} yields that, for any  
\[
\lambda < \lambda_{\max}(n,m,t^{*}) \bydef
\min \set{\tfrac{\sqrt{n} - \sqrt{m}}{2 L}, t^{*}}. 
\] 
the minimizer $\vx^{*}$ of the equality-constrained optimization \eqref{eq:no_lasso} coincides with the minimizer $\tilde{\vx}$ of the lasso formulation \eqref{eq:lasso}-\eqref{eq:frob_variational} except with probability at most $\e^{-n(1 - \sqrt{\delta+1/n})^2/8}$.

Overall this yields, for $\lambda < \lambda_{\max}$ and $\epsilon \leq \min(\alpha^{*},A-\alpha^{*})$:
\begin{equation}
\label{eq:CombiningIngredientsMainProbBound}
 \prob \big[\abs{\sqrt{n}\|\vx^{*}\|-\alpha^{*}} \geq \epsilon\big]  \leq \e^{-n(1 - \sqrt{\delta+1/n})^2/8} + 4\zeta\big(n,\tfrac{\omega(\epsilon)}{2}\big).
\end{equation}
The infimum over admissible values of $\lambda$ is obtained by taking its value when $\lambda = \lambda_{\max}$. 

\paragraph{Expliciting the bound~\eqref{eq:CombiningIngredientsMainProbBound}.}

From now on we choose $A \bydef 2\alpha^{*}$ and, for $0 < \epsilon' \leq 1$, we consider $\epsilon \bydef \epsilon' \alpha^{*}$ (which satisfies $0 < \epsilon  \leq \min(\alpha^{*},A-\alpha^{*}) = \alpha^{*}$). 
We use $\lesssim_{\delta}$ and $\gtrsim_{\delta}$ to denote inequalities up to a constant that may depend on $\delta$, but not on $n$ or $\epsilon'$, provided $n \geq n(\delta)$. We specify $n(\delta)$ where appropriate.

By Lemma~\ref{le:genericboundTp}-item~\ref{it:TMax}, for any $1 \leq p \leq 2$ and any $n \geq 1$
\begin{eqnarray*}
\lambda_{\max} = t^{*}\ \min\left(\tfrac{\sqrt{n}-\sqrt{m}}{n^{1/p-1/2}}\tfrac{1}{2t^{*}},1\right)
&\gtrsim_{\delta}& t^{*}\ \min\left(\tfrac{\sqrt{n}(1-\sqrt{m/n})}{n^{1/p-1/2}}\tfrac{1}{n^{1-1/p}},1\right)
= t^{*} \min\left(1-\sqrt{m/n},1\right).
\end{eqnarray*}
For $n \geq \tfrac{2}{1-\delta}$ we have 
\begin{equation}\label{eq:lipconcentTmp}
1-\sqrt{m/n} = 1-\sqrt{\delta+\tfrac{1}{n}} \geq 1-\sqrt{(1+\delta)/2} \gtrsim_{\delta} 1
\end{equation}
hence $\lambda_{\max} \gtrsim_{\delta} t^{*}$.

With the shorthands $D(t) = D_{p}(t;n)$ and $D = D_{p}(t^{*};n)$, we have
\[
1+\delta (\alpha^{*}+\epsilon)^{2} \leq 1+\delta A^{2} \leq 4(1+\delta (\alpha^{*})^{2}) = 4(1+\tfrac{D}{\delta-D})= \frac{4\delta}{\delta-D} \lesssim_{\delta} (\delta-D)^{-1}.
\]
hence with $\lambda = \lambda_{\max}$ we get for $n \geq 2/(1-\delta)$ 
\begin{eqnarray*}
\omega(\epsilon) 
& = &
\frac{\epsilon'^{2}}{2} \frac{\lambda_{\max}}{t^{*}} \frac{\delta (\alpha^{*})^{2}}{(1+\delta(\alpha^{*}+\epsilon)^{2})^{3/2}}
\gtrsim_{\delta} 
\epsilon'^{2} \frac{ \tfrac{D}{\delta-D}}{(\tfrac{4\delta}{\delta-D})^{3/2}}
\gtrsim_{\delta}
\epsilon'^{2}  D\sqrt{\delta-D} 
\end{eqnarray*}
By Lemma~\ref{le:genericboundTp}-item~\ref{it:lowerBoundDp} we have $D = D_{p}(t^{*};n) \geq (\delta/C)^{2}$ for a universal constant $C$ independent of $n$ or $p$, hence for $n \geq 1$, $D \gtrsim_{\delta} 1$ and for $n \geq 2/(1-\delta)$,
\[
\omega(\epsilon) \gtrsim_{\delta} \sqrt{-t^{*}D'_{p}(t^{*};n)} \epsilon'^{2}
\]
Moreover since $\min(\delta,A^{-2}) \geq \tfrac{1}{4}\min(\delta,(\alpha^{*})^{-2}) = \tfrac{1}{4}\min(\delta,\delta(\delta-D)/D)$ we also get
\begin{eqnarray*}
c(\delta,A) 
&= &
\tfrac{\min(\delta,A^{-2})}{1+\delta A^{2}} 
\geq
\tfrac{\min(\delta, \delta(\delta-D)/D)}{\delta/(\delta-D)}
=
\min(\delta-D, \tfrac{(\delta-D)^{2}}{D}) \geq (\delta-D)^{2} = \tfrac{1}{4} [-t^{*}D'_{p}(t^{*};n)]^{2}.
\end{eqnarray*}
Since $-t^{*} D'_{p}(t^{*};n) \geq \gamma(\delta) > 0$ for any $n \geq N(\delta)$ (recall that $t^{*}= t^{*}_{p}(\delta;n)$), we have
\begin{equation}\label{eq:MainAssumption}
-t^{*}D'_{p}(t^{*}) \gtrsim_{\delta} 1,
\end{equation}
and we obtain for $n \geq \max(2/(1-\delta),N)$: $\omega(\epsilon) \gtrsim_{\delta} \epsilon'^{2}$ and $c(\delta,A) \gtrsim_{\delta} 1$. Combining the above yields, for $0<\epsilon'\leq 1$, $n \geq \max(2/(1-\delta),N(\delta))$:
\begin{eqnarray*}
    \prob[\abs{\norm{\vx^*} - \tfrac{\alpha^*}{\sqrt{n}}} \geq \tfrac{\epsilon'\alpha^{*}}{\sqrt{n}}] 
&    \leq &   \e^{-C_{1} n}
    + 4\zeta\big(n, C_{2} \epsilon'^{2}\big)\\
 &\leq &  
+ \frac{1}{K_{1}(\epsilon')^{2}} \e^{-K_{2}(\epsilon')^{4} n} 
\end{eqnarray*}
with $K_{i},C_{i} \gtrsim_{\delta} 1$. 


\section{Conclusion} 
\label{sec:conclusion}

In this paper (Part II of the ``Beyond Moore-Penrose'' mini-series) we looked at generalized matrix inverses which minimize entrywise $\ell^p$ norms, with a particular emphasis on the sparse pseudoinverse ($p = 1$). Our central result is Theorem~\ref{thm:frob_of_l1} together with Corollaries \ref{cor:P2} and \ref{cor:P1} which discuss numerical stability of a class of generalized inverses as measured by their Frobenius norm. This allows us to quantify the MSE hit incurred by using the sparse pseudoinverse instead of the MPP and to show that, in fact, this hit is controlled. 

We highlight three main results about the sparse pseudoinverse $\spinv(\mA) = \ginv{1}{\mA}$ of a \emph{generic} matrix $\mA \in \R^{m \times n}$: 1) it is unique, 2) it has precisely $m$ zeros per column, and 2) its Frobenius norm is characterized by Corollary \ref{cor:P1}. For a large range of $m / n$ (with $m < n$) the Frobenius norm of the sparse pseudoinverse is relatively close to the Frobenius norm of the MPP which is the smallest possible among all generalized inverses (Figure \ref{fig:frobinv1}). This does not hold for certain ad hoc strategies that yield generalized inverses with the same non-zero count (Figure \ref{fig:minor}). More generally, we gave finite-size concentration bounds for the square of the Frobenius norm of all $\ginv{p}{\,\cdot\,}$, $1 \leq p \leq 2$, complementing a known result for $p = 2$.

Theorem~3.1 of Part I and Corollary~\ref{cor:P1} together establish a form of regularity of the sparse pseudoinverse, namely that $\tfrac{n}{m} \E \left[ \spinv(\mA) \mA \vx\right] = \vx$ and that $\norm{\spinv(\mA) \mA \vx}$ is well-behaved (at least for Gaussian random matrices)---two essential properties exhibited by the MPP. 

As a useful side product, along the way we presented a number of new results related to matrix analysis, sparse representations, and convex programming. For example, we proved that basis pursuit generically has a unique minimizer---a folklore fact mentioned repeatedly in the literature, but for which we could find no proof. The same goes for the fact that $\ell_p$ minimization for $0 \leq p \leq 1$ has a sparse solution---to the best of our knowledge, yet another piece of thus far unproven folklore.

The most important future work is to develop extensions of Theorem \ref{thm:frob_of_l1} and the two corollaries to $p > 2$.



\section{Acknowledgments}

The authors would like to thank Mihailo Kolund\v{z}ija, Miki Elad, Jakob Lemvig, and Martin Vetterli for the discussions and input in preparing this manuscript. A special thanks goes to Christos Thrampoulidis for his help in understanding and applying the Gaussian min-max theorem, and to Simon Foucart for discussions leading to the proof of Lemma~\ref{lem:zeros_in_l1} for $p<1$.

This work was supported in part by the European Research Council, PLEASE project (ERC-StG-2011-277906).


\section*{Appendices}
\appendix

\section{Results about Gaussian processes}

\subsection{Concentration of measure} 

\begin{lemma}
\label{lem:concentration_results}
Let $\vh$ be a standard Gaussian random vector of length $n$, $\vh \sim \mathcal{N}(\vec{0}, \mI_n)$, and $f : \R^n \to \R$ a 1-Lipschitz function. Then the following hold:

\begin{enumerate}[label=(\alph*)]
    \item \label{lemitem:gausnormsingle} For any $0 < \epsilon < 1$, $\prob[\norm{\vh}^2 \leq n(1 - \epsilon)] \leq \e^{-\epsilon^2 / 4}$; $\prob[\norm{\vh}^2 \geq n/(1 - \epsilon)] \leq \e^{-\epsilon^2 / 4}$;
    \item \label{lemitem:gausnormabs} For any $0 < \epsilon < 1$, $\prob\left\{\norm{\vh}^2 \notin [(1 - \epsilon)n, n/(1 - \epsilon)]\right\} \leq 2\e^{-\tfrac{\epsilon^2 n}{4}}$;
    \item \label{lemitem:norm2abs} For any $\epsilon > 0$,
        $\prob\left[\abs{\norm{\vh}^2 - n } \geq \sqrt{\epsilon n}\right] \leq         
        \begin{cases}
            2 \e^{-\frac{\epsilon}{8}} & \text{for $0 \leq \epsilon \leq n$}, \\
            2 \e^{-\frac{\sqrt{\epsilon n}}{8}} &\text{for $\epsilon > n$}.
        \end{cases}$
    \item \label{lemitem:lipgaussingle} For any $u > 0$, $\prob\left[f(\vh) - \E f(\vh) \geq t\right] \leq \e^{-u^2/2}$; $\prob\left[f(\vh) - \E f(\vh) \leq -u\right] \leq \e^{-u^2/2}$;
    \item \label{lemitem:lipgausabs} For any $u > 0$, $\prob[\abs{f(\vh) - \E f(\vh)} \geq u] \leq 2\e^{-u^2/2}$;
    \item \label{lemitem:varoflip} $\var[f(\vh)] \leq 1$
\end{enumerate}
\end{lemma}

\begin{proof}[Proofs and references]
We mostly refer to texts where the proofs can be found and show how to modify the standard forms of the bounds address our needs.
\begin{enumerate}[label=(\alph*)]
    \item See \cite[Corollary 2.3]{barvinok2005math}.
    \item Immediate consequence of $\ref{lemitem:gausnormsingle}$ by a union bound.
    \item We use a Bernstein-type inequality for sub-exponential random variables. Namely, $X_i = h_i^2$ is $\chi_1^2$ which is subexponential with parameters $\nu = 2, b = 4$ \cite[Example 2.4]{wainwright}, i.e. $\E[\e^{\lambda(X_i - \mu)}] \leq \e^{\frac{\nu^2 \lambda^2}{2}}$ for all $\abs{\lambda} \leq \frac{1}{b}$, and $\mu = \E[X_i] = 1$. Applying \cite[Proposition 2.2]{wainwright}
    then yields
    \[
        \prob\left[\abs{\norm{\vh}^2 - n} \geq \sqrt{\epsilon n} \right] \leq 2
        \begin{cases}
            \e^{-\frac{\epsilon}{8}} & \text{for $0 \leq \epsilon \leq n$}, \\
            \e^{-\frac{\sqrt{\epsilon n}}{8}} &\text{for $\epsilon > n$}.
        \end{cases}
    \]
    \item \cite[Eq. (1.22)]{Ledoux:1999ip}.
    \item Union bound applied to \ref{lemitem:lipgaussingle}.
    \item A consequence of Poincar\'e inequality for Gaussian measures \cite[Eq. (2.16)]{Ledoux:1999ip}: $\var[f(\vh)] \leq \E [\norm{\nabla f(\vh)}^2]$ for 1-Lipschitz $f$ for which $\norm{\nabla f(\vh)} \leq 1$.
\end{enumerate}
\end{proof}

\subsection{Complementary error function and related functions} 

We will be using the complementary error function and will need some basic properties which can be found e.g. in \cite{Chiani:2003eu} or references therein. 
\begin{lemma}
\label{le:erfc}
The complementary error function is defined for any $z \in \mathbb{R}$ as
\begin{equation}
\label{eq:deferfc}
\erfc(z) \bydef \frac{2}{\sqrt{\pi}} \int_{z}^{\infty} e^{-t^{2}} \! dt.
\end{equation}
For any $z>0$ we have
\begin{equation}
\label{eq:upperbounderfc}
\erfc(z) \leq \exp(-z^{2}).
\end{equation}
Moreover, if $g \sim \mathcal{N}(0,1)$ is a standard centered normal variable then any $x>0$
\begin{equation}
\prob \left\{ |g| > x\right\} = \erfc(x/\sqrt{2}).
\end{equation}
\end{lemma}

\begin{lemma}
\label{lem:theta_explicit}
Let $\theta(t) \bydef \E (\abs{h} - t)_+^2$, where $h \sim \mathcal{N}(0, 1)$ and $t \geq 0$. Then 
\begin{enumerate}
    \item $\theta(t) = \left(t^2+1\right) \erfc\left(\frac{t}{\sqrt{2}}\right)-\sqrt{\frac{2}{\pi }} \e^{-\frac{t^2}{2}} t $,
    \item $\theta(t) - (t/2) \theta'(t) = \erfc(t / \sqrt{2})$.
\end{enumerate}
\end{lemma}

\begin{proof}
$\displaystyle\begin{aligned}[t]
\theta(t) 
= \E (\abs{h} - t)_+^2 
= \int_{-\infty}^\infty (\abs{\varphi} - t)_+^2 p_h(\varphi) \ \di \varphi 
= 2 \int_t^\infty (\varphi^2 - 2 t \varphi + t^2) p_h(\varphi) \ \di \varphi.
\end{aligned}$\\
Now we compute as follows:
\begin{itemize}
    \item 
    
    $\displaystyle\begin{aligned}[t]
    \int_t^\infty \varphi^2 p_h(\varphi) \ \di \varphi 
    \ = \ \tfrac{1}{\sqrt{2\pi}} \int_t^\infty \varphi^2 \e^{-\frac{\varphi^2}{2}} \ \di \varphi 
    \ = \ & \tfrac{1}{\sqrt{2\pi}}  \int_t^\infty (-\varphi) \frac{\di}{\di \varphi} \left\{ \e^{-\frac{\varphi^2}{2}} \right\} \ \di \varphi \\
    \ = \ & \tfrac{1}{\sqrt{2\pi}} \left( \left[-\varphi \e^{-\frac{\varphi^2}{2}} \right]_t^{\infty} + \int_t^\infty \e^{-\frac{\varphi^2}{2}} \ \di \varphi \right) \\
    \ \stackrel{(a)}{=} \ & \tfrac{1}{2} \erfc\left(\tfrac{t}{\sqrt{2}}\right)+\tfrac{\e^{-\frac{t^2}{2}} t}{\sqrt{2 \pi }},
    \end{aligned}$\\
    where in $(a)$ we used Lemma \ref{le:erfc},
    \item
    $\displaystyle\begin{aligned}[t]
        \int_t^\infty \varphi p_h(\varphi) \ \di \varphi = \tfrac{1}{\sqrt{2\pi}} \int_t^\infty \frac{\di}{\di \varphi} \left\{ -\e^{-\frac{\varphi^2}{2}} \right\} \ \di \varphi = \tfrac{1}{\sqrt{2\pi}} \left[-\e^{-\frac{\varphi^2}{2}} \right]_t^{\infty} = \tfrac{\e^{-\frac{t^2}{2}}}{\sqrt{2 \pi }},
    \end{aligned}$
    \item 
    $\displaystyle\begin{aligned}[t]
    \int_t^\infty p_h(\varphi) \ \di \varphi = \tfrac{1}{2} \erfc \left( \tfrac{t}{\sqrt{2}} \right)
    \end{aligned}$ again by Lemma \ref{le:erfc}.
\end{itemize}
Putting together the above we get the desired expression. The second statement follows from the fact that $\frac{\di}{\di t} \erfc \left( \frac{t}{\sqrt{2}} \right) = -\sqrt{\frac{2}{\pi }}\e^{-\frac{t^2}{2}}$.
\end{proof}

\subsection{Convex Gaussian Min-Max Theorem (CGMT) \cite[Theorem 6.1]{Thrampoulidis:2016vo}}

\label{append:CGMT}

Let the principal optimization (PO) and auxiliary optimization (AO) be defined as 

\begin{align}
    \Phi(\mG) & \ \bydef \ \min_{\vv \in \setS_\vv} \max_{\vu \in \setS_\vu} \ \vu^\T \mG \vw + \psi(\vv,\vu)     \label{eq:po}
\\
    \phi(\vg, \vh) & \ \bydef \ \min_{\vv \in \setS_\vv} \max_{\vu \in \setS_\vu} \ \norm{\vv}_2 \vg^\T \vu + \norm{\vu}_2 \vh^\T \vv + \psi(\vv, \vu),     \label{eq:ao}
\end{align}
with $\mG \in \R^{m \times n}$, $\vg \in \R^m$, $\vh \in \R^n$, $\setS_\vv \subset
\R^n$, $\setS_\vu \subset \R^m$ and $\psi : \R^n \times \R^m \to \R$. Let further
$\vv_\Phi \bydef \vv_\Phi(\mG)$ denote any optimal minimizer of (PO) and $\vv_\phi \bydef \vv_\phi(\vg,\vh)$ any optimal minimizer of (AO). Then the following holds:

\begin{theorem}
    \label{thm:cgmt}

    In \eqref{eq:po} and \eqref{eq:ao}, let $\setS_\vv$ and $S_\vu$ be compact and $\psi$
    continuous on $\setS_\vv \times \setS_\vu$. Let also $\mG$, $\vg$, $\vh$ have
    entries that are iid standard normal. The following hold:

    \begin{enumerate}
        \item \label{thmpart:1} For all $c \in \R$
        \begin{equation*}
            \prob[  \Phi(\mG) < c ] \leq 2\prob[ \phi(\vg, \vh) \leq c ]
        \end{equation*}
        
        \item \label{thmpart:cgmt_optval_concentrate} Assume further that $\psi(\vv,\vu)$ is convex-concave on $\setS_\vv \times \setS_\vu$ where $\setS_{\vv}$ and $\setS_{\vu}$ are convex. Then for all $c \in \R$
        \begin{equation*}
            \prob[ \Phi(\mG) > c] \leq 2 \prob[\phi(\vg, \vh) \geq c].
        \end{equation*}
        In particular, for all $\mu \in \R$ and $t > 0$,
        \begin{equation*}
            \prob[\abs{\Phi(\mG) - \mu} > t] \leq 2\prob[\abs{\phi(\vg, \vh) - \mu} \geq t].
        \end{equation*}
    \end{enumerate}
\end{theorem}


\section{Lemmata for Section \ref{sub:stability}}
\label{appendix:stability}

\begin{lemma}[Rewriting of {\cite[Lemma 9.2]{Oymak:2013vm}} with dependence on $\epsilon$ made explicit]
    \label{lem:lasso_satisfies_equality}
    Let $\mA \in \R^{m \times n}$ be a random matrix with iid standard normal
    entries, and $m < n$. Let further $\vy \in \R^m$. Consider the solution of an
    $\ell^2$-lasso with a regularizer $f$ which is $L$-Lipschitz with respect to the $\ell^{2}$-norm:

    \begin{equation*}
        \vx^\star \bydef \argmin_{\vx \in \R^n} \norm{\vy - \mA \vx} + \lambda f(\vx).
    \end{equation*}
    Then for any $0 \leq \epsilon < 1$ and $0 < \lambda < \frac{\sqrt{n} - \sqrt{m}}{L}(1 - \epsilon)$ the
    solution $\vx^\star$ satisfies $\vy = \mA \vx^\star$ with
    probability at least $1 - \e^{-\epsilon^2 (\sqrt{n} - \sqrt{m})^2 / 2}$. In other words,
    $\ell^2$-lasso (let us emphasize that the data fidelity term is \emph{not} squared here) gives the same result as the equality-constrained
    $f$-minimization.
\end{lemma}

\begin{proof}
    Using \cite[Corollary 5.35]{Vershynin:2009fv}\footnote{We actually use a one-sided variant of \cite[Corollary 5.35]{Vershynin:2009fv} which can be obtained by combining Lemma \ref{lem:concentration_results}\ref{lemitem:lipgaussingle} with the estimate of the expectation of $\sigma_{\text{min}}$, \cite[Theorem 5.32]{Vershynin:2009fv}.}, we have for every $t \geq 0$
    \[
        \prob[\sigma_{\text{min}}(\mA^\T) \leq \sqrt{n} - \sqrt{m} - t] \leq \e^{-t^2 / 2}.
    \]
    Dividing by $\sqrt{n} - \sqrt{m}$, we can rewrite this as

    \[
        \prob[\sigma_{\text{min}}(\mA^\T) / (\sqrt{n} - \sqrt{m}) \leq 1 - t / (\sqrt{n} - \sqrt{m})] \leq \e^{-t^2 / 2},
    \]
    or equivalently, by setting $t = \epsilon(\sqrt{n} - \sqrt{m})$ for any $\epsilon \geq 0$,

    \begin{equation}{}
        \label{eq:concentration_of_sigma_min}
        \prob[\sigma_{\text{min}}(\mA^\T) / (\sqrt{n} - \sqrt{m}) \leq 1 - \epsilon] \leq \e^{-\epsilon^2(\sqrt{n} - \sqrt{m})^{2}/ 2}.
    \end{equation}

    Let $\vp \bydef \vy - \mA \vx^\star$. The goal is to show that $\vp =
    \vzero$. To this end, let $\vw \bydef \mA^\dag \vp$, where $\mA^\dag =
    \mA^\T (\mA \mA^\T)^{-1}$ denotes the MPP (note that $\mA \mA^\T$ is
    almost surely invertible). Because
    \[
        \norm{\vw}^2 = \vp^\T (\mA \mA^\T)^{-1} \vp \leq \frac{\norm{\vp}^2}{ \sigma_{\text{min}}^2(\mA^\T)}
    \]
    we have for $0<\epsilon<1$, by the concentration of $\sigma_{\text{min}}(\mA^\T)$
    \eqref{eq:concentration_of_sigma_min} that with probability at least $1 -
    \e^{-\epsilon^2 (\sqrt{n} - \sqrt{m})^{2} / 2}$,

    \[
        \norm{\vw} \leq \frac{\norm{\vp}}{(\sqrt{n} - \sqrt{m})(1 - \epsilon)}.
    \]
    Let $\vx^\circ \bydef \vx^\star + \vw$ so that $\vy - \mA \vx^\circ = \vzero$.
    Using this together with the optimality of $\vx^\star$ gives

    \begin{equation}
        \label{eq:lasso_is_exact_nonpositive_diff}
        \big[ \norm{\vy - \mA \vx^\star} + \lambda f(\vx^\star) \big] - \big[ \norm{\vy - \mA \vx^\circ} + \lambda f(\vx^\circ) \big] \leq 0.
    \end{equation} 
    On the other hand,

    \begin{align*}
        \big[ \norm{\vy - \mA \vx^\star} + \lambda f(\vx^\star) \big] - \big[ \norm{\vy - \mA \vx^\circ} + \lambda f(\vx^\circ) \big] 
        = \ &\norm{\vp} + \lambda f(\vx^\star) - \lambda f(\vx^\circ) \\
        \geq \ &\norm{\vp} - \lambda \abs{f(\vx^\star) - \lambda f(\vx^\circ)} \\
        \geq \ &\norm{\vp} - \lambda L \norm{\vx^\star - \vx^\circ} \\
        = \ &\norm{\vp} - \lambda L \norm{\vw} \\
        \geq \ &\norm{\vp} \left( 1 - \frac{\lambda L}{(\sqrt{n} - \sqrt{m})(1 - \epsilon))} \right),
    \end{align*}
    where, by \eqref{eq:lasso_is_exact_nonpositive_diff}, the last expression must be non-positive. But if we choose

    \[
        \lambda < \frac{\sqrt{n} - \sqrt{m}}{L} (1 - \epsilon),
    \]
    the only way to make it non-positive is that $\norm{\vp} = 0$.
\end{proof}

\begin{lemma}
    \label{lem:concentration2}
    Consider $\vg \sim \mathcal{N}(\vzero, \mI_m)$, $\vh \sim \mathcal{N}(\vzero, \mI_n)$ and define
    \begin{eqnarray*}
    \Delta_p(\beta; \vh,\lambda) 
    &\bydef&   
    \tfrac{1}{\sqrt{n}}\  \dist(\beta \vh,  \norm{\,\cdot\,}_{p^*} \leq \lambda)\\
    \Delta_p(\beta;n,\lambda) 
    &\bydef&   \E[\Delta_p(\beta;\vh,\lambda)]
     \end{eqnarray*}
  There exist universal constants $c_1, c_2 > 0$ such that for any $0 < \epsilon < 2$, any integers $m,n$, any $A>0$ we have,  with $\delta \bydef (m-1)/n$ and $\phi(\alpha,\beta;\vg,\vh)$ defined as in~\eqref{eq:DefPhi2Scalars}:
    \begin{equation}
        \prob \big\{ \exists 0<\alpha \leq A, 0<\beta<1,\ \abs{\phi(\alpha, \beta ; \vg, \vh) - \kappa(\alpha, \beta)} \geq \epsilon \big\} \leq \tfrac{c_1}{\epsilon} \e^{-c_2 \epsilon^2 n \cdot \tfrac{\min(\delta,1/A^{2})}{1+\delta A^{2}}} \bydef \zeta(n, \epsilon; A,\delta) \label{eq:DefZeta}
    \end{equation}
    where
    \begin{equation}\label{eq:DefKappa}
       \kappa(\alpha, \beta) = \kappa_{p}(\alpha,\beta; n,\delta,\lambda)
    \bydef \beta \sqrt{\delta \alpha^2 + 1} - \alpha \Delta_p(\beta;n,\lambda).
\end{equation}
\end{lemma}

\begin{proof}
    We first look at the term $\norm{\alpha \vg / \sqrt{n} - \ve_1}$. 
    We have
    \[
        \norm{\alpha \vg / \sqrt{n} - \ve_1}^2 = \norm{\alpha \wt{\vg} / \sqrt{n}}^2 + (\alpha g_0 / \sqrt{n} - 1)^2,
    \]
    where we partitioned $\vg$ as $[g_0, \ \wt{\vg}^\T]^\T$, $\wt{\vg} \in \mathbb{R}^{m-1}$. 
    By Lemma \ref{lem:concentration_results}\ref{lemitem:gausnormabs}, for a standard Gaussian random vector $\vx \sim \mathcal{N}(\vzero,\mI_{n})$ it holds that for any $0<\epsilon<1$
    \begin{equation}\label{eq:stdGaussianConcentration}
        \prob\left\{\norm{x}^2 \notin [(1 - \epsilon)n, n/(1 - \epsilon)]\right\} \leq 2\e^{-\tfrac{\epsilon^2 n}{4}}.
    \end{equation}
    Denoting $\delta \bydef \frac{m-1}{n}$, it follows that
    \[
        \prob \left[  \exists 0<\alpha \leq A,\  \  \norm{\alpha \wt{\vg} / \sqrt{n}}^2 \notin [(1 - \epsilon) \delta \alpha^2, \delta \alpha^2/(1 - \epsilon)]\right] \leq 2 \e^{-\tfrac{\epsilon^2 (m-1)}{4}} = 2 \e^{-\tfrac{\epsilon^{2}n\delta}{4}}.
    \]
    Next, we show that the term $(\alpha g_0 / \sqrt{n} - 1)^2$ cannot deviate much from 1: setting $\epsilon' = \epsilon/2$,
    we have $\sqrt{1-\epsilon} < 1-\epsilon'$ and $1+\epsilon' < 1/\sqrt{1-\epsilon}$ hence
    \begin{align*}
        \prob \left\{ \forall 0<\alpha \leq A,\  \left(\tfrac{\alpha g_0}{\sqrt{n}} - 1\right)^2 \in [1 - \epsilon, 1/(1 - \epsilon)]\right\} 
        \geq \
        &\prob \left\{\forall 0<\alpha \leq A,\  \abs{\tfrac{\alpha g_0}{\sqrt{n}} - 1} \in [1 - \epsilon', 1 + \epsilon']\right\} \\
        \geq \ 
        &
        \prob \left\{\forall 0<\alpha \leq A,\ \alpha g_0 / \sqrt{n} \in [- \epsilon',\epsilon']\right\}\\
        = \ &  \prob \left\{A g_0 / \sqrt{n} \ \in \ [- \epsilon',\epsilon']\right\}
    \end{align*}
   so that using Lemma~\ref{le:erfc} we have, with $\erfc$ the complementary error function,    
   \begin{align*}
        \prob \left\{ \exists 0<\alpha \leq A,\ \ \left(\tfrac{\alpha g_0}{\sqrt{n}} - 1\right)^2 \notin [1 - \epsilon, 1/(1 - \epsilon)]\right\}
 &       \leq 
        \prob( |g_0| > \tfrac{\epsilon'\sqrt{n}}{A}) =
    \erfc \left[ \tfrac{\epsilon' \sqrt{n}}{\sqrt{2} A} \right] \\
&        \leq \e^{-\tfrac{\epsilon'^2 n}{2 A^2}}
        = \e^{-\tfrac{\epsilon^2 n}{8 A^2}}.
    \end{align*}
    Combining the above, we get that: for any $0<\epsilon<1$, setting $\epsilon' = 1-(1-\epsilon)^{2} = \epsilon(2-\epsilon) \geq \epsilon$,
    \begin{align}
        \prob & \left\{  \exists 0<\alpha \leq A,\  \ \norm{\alpha \vg / \sqrt{n} - \ve_1} \notin [(1-\epsilon)\sqrt{\delta \alpha^2 + 1}, \sqrt{\delta \alpha^2 + 1}/(1-\epsilon)] \right\}\notag\\
        & = \prob \left\{  \exists 0<\alpha \leq A,\  \ \norm{\alpha \vg / \sqrt{n} - \ve_1}^2 \notin [(1-\epsilon')(\delta \alpha^2 + 1), (\delta \alpha^2 + 1)/(1-\epsilon')] \right\}\notag\\
        & \leq \ 2\e^{-\tfrac{\epsilon'^2 n \delta}{4}} + \e^{-\tfrac{\epsilon'^2 n}{8A^2}} 
         \leq \ c'_1 \e^{-c'_2 \epsilon'^2 n \cdot \min(\delta,1/A^{2})} \leq c'_{1} \e^{-c'_{2}\epsilon^{2}n \cdot \min(\delta,1/A^{2})},\label{eq:concentrationlemmaterm1}
    \end{align}
    where the constants $c'_{1},c'_{2}$ are universal.  
    
    For the second term in $\phi(\alpha, \beta; \vg, \vh)$, we use Gaussian
    concentration of Lipschitz functions. It is known that (Euclidean) distance to a convex set, $\dist(\vx, \cal C)$, is 1-Lipschitz in $\vx$ (with respect to the Euclidean metric) so by Lemma \ref{lem:concentration_results}\ref{lemitem:lipgausabs} we get for $0<\beta \leq 1$ (we omit the dependency in $\lambda$ for brevity) and any $t \geq 0$
    \[
        \prob \left[ \tfrac{\sqrt{n}}{\beta} \abs{ \Delta_p(\beta; \vh) - \Delta_p(\beta)} > t \right]
        \leq
        2 \e^{- \tfrac{t^2}{2}}, 
    \]
    Setting $t=\tau \sqrt{n}/\beta$ we obtain for any $\tau \geq 0$
        \begin{equation}
        \prob \left[ \abs{ \Delta_p(\beta; \vh) - \Delta_p(\beta)} > \tau \right]
        \leq 2 \e^{- \tfrac{\tau^2 n}{2\beta^{2}}} \leq 2 \e^{- \tfrac{\tau^2 n}{2}},\label{eq:BoundFh}
    \end{equation}
    This obviously extends to $\beta = 0$ since $\Delta_{p}(0;\vh) = \Delta_{p}(0) = 0$.

    Next we want to bound 
    \[ 
        \prob \left[ \exists \beta \in [0, 1], \ \abs{ \Delta_p(\beta; \vh) - \Delta_p(\beta)} > \tau \right]
        =
        \prob \left[ \sup_{\beta \in [0, 1]} \abs{ \Delta_p(\beta; \vh) - \Delta_p(\beta)} > \tau \right].
    \] 

    By Lemma \ref{lem:sqrtn_lipschitz}, 
    the function $\beta \mapsto f_{\vh}(\beta) \bydef \abs{
    \Delta_p(\beta; \vh) - \Delta_p(\beta)}$ is $L_{\vh}$-Lipschitz in $\beta$ with $L_{\vh} \bydef \max \set{\norm{\vh}/\sqrt{n}, 1}$. Hence $f_{\vh}$ is continuous, and its supremum on the closed interval $[0,\ 1]$ is indeed a maximum reached at some maximizer $\beta_{\vh}$. 
    
    Let $b \leq 1/2$ and $Y_{b} = \set{b\tau/2,3b\tau/2,\ldots,(n-1/2)b\tau}$ be a uniform sampling of $[0,
    1]$ with spacing $b \tau$, with the last segment possibly being
    shorter. For a given $\vh$, there exists $y_\vh \in Y_{b}$ such that $\abs{\beta_{\vh} - y_{\vh}} \leq b \tau$. 
    For this $y_{\vh}$ we write
    \[
        f_{\vh}(\beta_{\vh}) - f_{\vh}(y_{\vh}) \leq \abs{ f_{\vh}(\beta_{\vh}) - f_{\vh}(y_{\vh}) } 
        \leq L_{\vh} \abs{\beta_{\vh} - y_{\vh}} \leq L_{\vh} b \tau
    \]
    so that 
    \begin{eqnarray*}
          \prob[ \sup_{\beta \in [0, 1]} f_{\vh}(\beta) > \tau] 
          &=&      \prob[ f_{\vh}(\beta_\vh) > \tau ] 
          \ \leq\  \prob[ f_{\vh}(y_\vh) + L_{\vh} b\tau > \tau]
          \ \leq\  \prob\left[ f_{\vh}(y_\vh) > \tau/2 \bigvee L_{\vh} b\tau > \tau/2\right]\\
          &\leq& \prob[ f_{\vh}(y_\vh) > \tau/2] + \prob[\max \set{\norm{\vh}/\sqrt{n},1} b > 1/2]
      \end{eqnarray*}
      As we do not know a priori to which $y \in Y_{b}$ the maximizer $\beta_{\vh}$ will be close, we continue with a union bound, and we further use that $b \leq 1/2$ to obtain by~\eqref{eq:BoundFh} and Lemma \ref{lem:concentration_results}\ref{lemitem:gausnormsingle}
          \begin{eqnarray*}
            \prob[ \sup_{\beta \in [0, 1]} f_{\vh}(\beta) > \tau] 
      & \leq & \prob\Big[ \bigvee_{y \in Y_{b}} f_{\vh}(y) > \tau/2\Big] + \prob\Big[\norm{\vh}^{2} > \tfrac{n}{4b^{2}}\Big]
         \leq   \sharp Y_{b} \  2 \e^{-\tfrac{\tau^2 n}{8}} + \e^{-\tfrac{(1-4b^{2})^{2}n}{4}}\\
        & \leq & (1+1/(b\tau)) 2 \e^{-\tfrac{\tau^2 n}{8}} + \e^{-\tfrac{(1-4b^{2})^{2}n}{4}}.
    \end{eqnarray*}
    We conclude the argument for the second term by choosing
    $b \bydef \tfrac{1}{2}\sqrt{\tfrac{2 - \sqrt{2}}{2}} \approx 0.27$, so that: for any $0<\tau<1$, we have $\tfrac{(1-4b^{2})^2}{4} = \tfrac{1}{8} \geq \tfrac{\tau^2}{8}$  hence
    \begin{equation}
       \label{eq:concentrationlemmaterm2}
      \prob[ \sup_{\beta \in [0, 1]} f_{\vh}(\beta) > \tau] 
      \leq
      (3+2/(b\tau)) \e^{-\tfrac{\tau^2 n}{8}}
       <  \tfrac{11}{\tau} \e^{-\tfrac{\tau^2 n}{8}}, ~ 0 < \tau < 1.
    \end{equation}

    To conclude we combine concentration bounds for both terms.
    First, we observe that
    \[
    \sup_{\substack{0 \leq \alpha \leq A\\ 0 \leq \beta \leq 1}} \abs{\phi(\alpha,\beta;\vg,\vh)-\kappa(\alpha,\beta)}
    \leq 
    \sup_{0 \leq \alpha \leq A} \abs{\norm{\frac{\alpha \vg}{\sqrt{n}}-\ve_{1}}-\sqrt{1+\delta\alpha^{2}}}
    +
    A \sup_{0 \leq \beta \leq 1} \abs{\Delta_{p}(\beta;\vh)-\Delta_{p}(\beta)}
    \]
    hence by a union bound we just need to control the probability that each term exceeds $\epsilon/2$.
    Since we assume that $0<\epsilon<2$, we can use the multiplicative control~\eqref{eq:concentrationlemmaterm1} of the first term and convert it into the desired additive bound as follows:
    \begin{equation*}
    \begin{aligned}
      \prob \left\{  \exists \alpha \in (0,A], \ \abs{\norm{\tfrac{\alpha \vg}{\sqrt{n}} - \ve_1} - \sqrt{\delta \alpha^2 + 1}}  > \tfrac{\epsilon}{2} \right\}
      & =   \prob \left\{  \exists \alpha \in (0,A],\  \ \abs{\tfrac{\norm{\tfrac{\alpha \vg}{\sqrt{n}} - \ve_1}}{\sqrt{\delta \alpha^2 + 1}}-1} >\tfrac{\epsilon}{2\sqrt{\delta \alpha^2 + 1}}] \right\}\\
      & \leq   \prob \left\{  \exists \alpha \in (0,A],\  \ \abs{\tfrac{\norm{\tfrac{\alpha \vg}{\sqrt{n}} - \ve_1}}{\sqrt{\delta \alpha^2 + 1}}-1} >
      \tfrac{\epsilon}{2\sqrt{\delta A^2 + 1}} \right\}\\
      & \leq   \prob \left\{ \exists \alpha \in (0,A],\  \ \tfrac{\norm{\alpha \vg / \sqrt{n} - \ve_1}}{\sqrt{\delta \alpha^2 + 1}} \notin   [1-\epsilon',\tfrac{1}{1-\epsilon'}] \right\}\\
      &\leq c'_1 \e^{-c'_2 \epsilon'^2 n \cdot \min(\delta,1/A^{2})}.
    \end{aligned}
    \end{equation*}
    provided that $[1-\epsilon',1/(1-\epsilon')] \subset [1-\tfrac{\epsilon}{2\sqrt{\delta A^2 + 1}},1+\tfrac{\epsilon}{2\sqrt{\delta A^2 + 1}}]$. This is achieved with $\epsilon' = \epsilon/(1+\sqrt{\delta A^{2}+1})$. Combining the resulting probability bound with the bound on $\prob[ \sup_{\beta \in [0, 1]} f_{\vh}(\beta) > \tfrac{\epsilon}{2A}]$ resulting from~\eqref{eq:concentrationlemmaterm2} yields the result.
\end{proof}

\begin{corollary}\label{cor:empiricalmaxbeta}
    Consider $A>0$ and define for any set $S \subseteq [0, A]$:
    \begin{align*}
        \phi(\alpha ; \vg, \vh) & \bydef  \sup_{0 \leq \beta \leq 1} \phi(\alpha, \beta ; \vg, \vh), & 
        \kappa(\alpha) & \bydef  \sup_{0 \leq \beta \leq 1} \kappa(\alpha, \beta), \\
        \phi_S(\vg, \vh) & \bydef  \inf_{\alpha \in S} \phi(\alpha ; \vg, \vh), &
        \kappa_S & \bydef  \inf_{\alpha \in S} \kappa(\alpha),
    \end{align*}
    With $\vg \sim \mathcal{N}(\vzero, \mI_m)$, $\vh \sim \mathcal{N}(\vzero, \mI_n)$, $\delta = (m-1)/n$, we have for $0<\epsilon<2$
    \begin{equation}
        \label{eq:concentration_of_min_alpha}
        \prob \left[ \sup_{S \subset [0,A]} \abs{\phi_S(\vg, \vh) - \kappa_S} \geq \epsilon \ \right] \leq \zeta(n, \epsilon).
    \end{equation}
    with $\zeta(n,\epsilon)=\zeta(n,\epsilon;A,\delta)$ defined in~\eqref{eq:DefZeta}. In particular, 
    \begin{equation}
        \label{eq:concentration_of_max_beta}
        \prob \left[ \sup_{0 \leq \alpha \leq A}\abs{\phi(\alpha) - \kappa(\alpha)} \geq \epsilon \right] \leq \zeta(n, \epsilon).
    \end{equation}
\end{corollary}
\begin{proof}
To lighten notation we suppress the dependence of the stochastic function $\phi$ on random vectors $\vg$ and $\vh$. By Lemma~\ref{lem:concentration2} we have with probability at least $1-\zeta(n,\epsilon)$: for all $0 \leq \alpha \leq A$ and $0 \leq \beta \leq 1$, $\abs{\phi(\alpha,\beta)-\kappa(\alpha,\beta)} \leq \epsilon$. When this holds we have for any $S \subset [0,A]$:
\begin{equation}
    \phi_{S} = \inf_{\alpha \in S} \phi(\alpha) \leq \inf_{\alpha \in S} [\kappa(\alpha) + \epsilon] = \kappa_{S}+\epsilon
\end{equation}
and
\begin{equation}
    \phi_{S} = \inf_{\alpha \in S} \phi(\alpha) \geq \inf_{\alpha \in S} [\kappa(\alpha)-\epsilon] = \kappa_{S}-\epsilon.
\end{equation}
\end{proof}
We will shortly characterize $\kappa(\alpha,\beta)$ and $\kappa(\alpha)$. To that end, we will use some properties of the following quantity
\begin{equation}
    \label{eq:DefDp}
    D_{p}(t;n) \bydef \left(\E[    \tfrac{1}{\sqrt{n}}\  \dist(\vh,  \norm{\,\cdot\,}_{p^*} \leq t)]\right)^{2}.
\end{equation}
    
\begin{lemma}[Deterministic properties of $D_p$]
    \label{lem:Dp_det_prop}
    Define 
    \(
      \mathcal{C}_{t} = \mathcal{C}_{t,p} \bydef \set{\vx \in \R^{n}: \norm{\vx}_{p^{*}} \leq t}
    \)
    and $D_p(t;n)$ as in~\eqref{eq:DefDp}. Using $D_{p}(t)$ as a shortand, the following hold:
    \begin{enumerate}
        \item The sets $\mathcal{C}_{t}$ are convex and nested with $\mathcal{C}_{t} \subsetneq \mathcal{C}_{t'}$ for $t < t'$.
        \item For any vector $\vh$, the function $t \mapsto \dist(\vh,\mathcal{C}_{t})$ is non-inscreasing and convex.
        \item $D_{p}(t)$ is a (strictly) decreasing convex function of $t$,
        \item $\lim_{t \to \infty} D_{p}(t) = 0$,
        \item $\frac{n}{n+1} \leq D_{p}(0) \leq 1$, 
        \item The function $t\mapsto D_{p}(t)$ is infinitely differentiable. 
        \item Let $g(t) = g(t;n) \bydef D_p(t) -\frac{t}{2}D'_{p}(t)$. For any $0 < \delta < D_{p}(0)$ there is a unique 
        \[t^{*} = t^{*}_{p}(\delta;n) \in (0,\infty)\] such that $g(t) > \delta$ for $t < t^*$ and $g(t) < \delta$ for $t > t^*$. It holds that $D_{p}(t^{*}) < g(t^{*}) = \delta$.
    \end{enumerate}
\end{lemma}

\begin{proof}
\begin{itemize}
  \item Property 1 is obvious. 
  \item Property 2: we recall that $\vx \mapsto \dist(\vx, {\cal C})$ is convex for convex $\cal C$ \cite[Example 3.16]{Boyd:2004uz}.
  Next, $(\vx, t) \mapsto t \dist(\vx / t, {\cal C})$ is convex in both
  arguments because it is the perspective of $\vx \mapsto \dist(\vx, {\cal
  C})$ \cite[Chapter 2]{Boyd:2004uz}. Applying this to ${\cal C} = {\cal C}_{1}$ and observing that $\dist(\vh,{\cal C}_{t}) = t\dist(\vh/t,{\cal C}_{1})$ we obtain that $\dist(\vh,{\cal C}_{t})$ is convex in $t$. The fact that it is non-increasing follows from Property~1.
  \item Property 3: 
  since expectation of convex functions is convex, and the pointwise square of a non-negative convex function is convex, we conclude that $D_p(t)$ is convex as claimed. The fact that it is (strictly) decreasing is obvious.
  \item Property 4: for any $\vy \in \R^{n}$ and $p\geq 1$, $\norm{\vy}_{p^{*}} \leq \norm{\vy}_{1} \leq \sqrt{n} \norm{\vy}$. Hence, for any given $t > 0$, $\dist(\vy, \mathcal{C}_{t}) = 0$ as soon as $\norm{\vy} \leq t / \sqrt{n}$. Further, for all $p \geq 1$ we can write
  \begin{align*}
    \E[  \dist(\vh, \mathcal{C}_{t}) ]
    &= \int_{\R^n} \dist(\vy, \mathcal{C}_{t}) \ p_\vh(\vy) \, \di \vy \\
    &= \int_{\vb \in \mathbb{S}^{n-1}} \int_{r = 0}^{\infty} \dist(r\vb, \mathcal{C}_{t}) \ p_\vh(r\vb) \ \mu(\di \vb) \ r^{n-1} \, \di r \\
    &= \int_{\vb \in \mathbb{S}^{n-1}} \int_{r = t/\sqrt{n}}^{\infty} \dist(r\vb, \mathcal{C}_{t}) \ p_\vh(r\vb) \ \mu(\di \vb) \ r^{n-1} \, \di r \\
    &\leq \int_{\vb \in \mathbb{S}^{n-1}} \int_{r = t/\sqrt{n}}^{\infty} r \ p_\vh(r\vb) \ \mu(\di \vb) \ r^{n-1} \, \di r \\
    &= \int_{\vb \in \mathbb{S}^{n-1}} \mu(\di \vb) \int_{r = t/\sqrt{n}}^{\infty} r^n \ Z_n \e^{-r^2 / 2} \ \di r \\
    &= Z_n \mu(\mathbb{S}^{n-1}) \int_{r = t/\sqrt{n}}^{\infty} r^n \e^{-r^2 / 2} \ \di r \to 0 \text{ as } t \to \infty,
  \end{align*}
  where $Z_n = (2\pi)^{-n/2}$ is the normalization term for the $n$-variate iid Gaussian distribution, and $\mu(\mathbb{S}^{n-1}) = 2\pi^{n/2} / \Gamma \left( \frac{n}{2} \right)$ is the surface area of the unit-radius $(n-1)$-sphere embedded in $\R^n$ with $\Gamma \left( \cdot \right)$ being the gamma function.

  \item Property 5: by Jensen's inequality and Property 3, we obtain the upper bound
  \[
    D_{p}(0) = \left( \E \tfrac{1}{\sqrt{n}} \dist(\vh,\mathcal{C}_{0})\right)^{2} 
    \leq \E \left(\tfrac{1}{\sqrt{n}} \dist(\vh,\mathcal{C}_{0})\right)^{2} = \E \tfrac{\norm{\vh}^{2}}{n} = 1.
  \]
  To get the lower bound, we again note that $n D_p(0) = \left[ \E \norm{\vh} \right]^2$, which can be computed by integration in polar coordinates. Working as in the proof of Property 4 we get that

  \[
      \E[  \dist(\vh, \mathcal{C}_{0}) ] = \E[ \norm{\vh} ]
      = \int_{\vb \in \mathbb{S}^{n-1}} \int_{r = 0}^{\infty} r \ p_\vh(r\vb) \ \mu(\di \vb) \ r^{n-1} \, \di r \\
  \]
  since $r = \norm{\vh} = \dist(\vh, \mathcal{C}_0)$. It follows that
  \begin{eqnarray*}
      \E[ \norm{\vh} ] 
      & = & Z_n \mu(\mathbb{S}^{n-1}) \int_{r = 0}^{\infty} r^n \e^{-r^2/2} \, \di r \\
      & \stackrel{(a)}{=} & (2 \pi)^{-n/2} \frac{ 2 \pi^{n / 2}}{\Gamma \left( \frac{n}{2} \right)} \int_{s = 0}^{\infty} 2^{\frac{n-1}{2}}s^{\frac{n-1}{2}} \e^{-s} \, \di r \\
      & \stackrel{(b)}{=} & {\sqrt{2} \, \Gamma \left( \frac{n+1}{2} \right) } \bigg/ {\Gamma \left( \frac{n}{2} \right) } 
  \end{eqnarray*}
  where in $(a)$ we used the substitution $u = r^2/2$, and in $(b)$ we invoked the definition of the gamma function, $\Gamma(z) = \int_0^\infty x^{z - 1} \e^{-x} \di x$. We now use the inequality of Wendel, \cite[Eq. (7)]{Wendel:1948fv}:
  \[
      {\Gamma(x + a)} \big/ {\Gamma(x)} \geq x(x + a)^{a-1}
  \]
  to conclude that 
  \begin{equation}\label{eq:lowerBoundExpNormH}
  \E [ \norm{\vh} ] \geq n (n + 1)^{-1/2}
  \end{equation}
   and $D_p(0) \geq n / (n + 1)$.
  \item Property 6: 
  starting similarly as in the proof of Property 4 and using $\dist(\vh, \mathcal{C}_{t}) = t \dist(\vh / t, \mathcal{C}_1)$ and a change of variable $r = t \rho$ we obtain the following expression for $D_p(t)$:
  \begin{align}
      q(t) \bydef \sqrt{n D_p(t)} &= t^{n+1} \int_{\vb \in \mathbb{S}^{n-1}} \int_{\rho = 0}^{\infty} \dist(\rho \vb, {\cal C}_1) p_\vh(\rho t \vb)  \, \mu(\di \vb) \rho^{n-1} \, \di \rho \\
      &= Z_n t^{n+1} \int_{\vb \in \mathbb{S}^{n-1}} \int_{\rho = 0}^{\infty} \dist(\rho \vb, {\cal C}_1) \e^{-\rho^2 t^2 / 2}  \, \mu(\di \vb) \rho^{n-1} \, \di \rho \label{eq:Dp_for_inf_diff}
  \end{align}
  so that $D_p(t)$ is infinitely differentiable (by the repeated application of the dominated convergence theorem). 
  \item Property 7: in particular, since 
  \[
      \abs{\parder{}{t} \dist(\rho \vb, {\cal C}_1) \e^{-\rho^2 t^2 / 2} \rho^{n-1}} 
      = \dist(\rho \vb, {\cal C}_1) t \rho^{n+1}  \e^{-\rho^2 t^2 / 2}
      \leq  \rho^{n+2} t  \e^{-\rho^2 t^2 / 2}
  \]
  where the rightmost expression is integrable for every $t > 0$ and $n \in \N$, the dominated convergence theorem allows us to differentiate under the integral sign in ~\eqref{eq:Dp_for_inf_diff} to get 
  \begin{align}
      -\tfrac{t}{2} D_p'(t) &= -\tfrac{t}{n} q(t) q'(t) \notag\\
      &= -\tfrac{1}{n} q(t) \left( (n+1) q(t) - Z_n t^{n+3} \int_{\vb \in \mathbb{S}^{n-1}} \int_{\rho = 0}^{\infty} \dist(\rho \vb, {\cal C}_1) \rho^{n+1} \e^{-\rho^2 t^2 / 2}  \, \mu(\di \vb) \, \di \rho \right).\label{eq:CpDevelopedExpression}
  \end{align}
  All terms can be seen to vanish as $t \to \infty$ by arguments analogous to those in the end of the proof of Property 5, hence $\lim_{t \to \infty} [-\tfrac{t}{2} D_p'(t)] = 0$.
   
  Since $D_p(t)$ is strictly decreasing we have $D_p'(t) < 0$. Since it is convex, we have $D_p''(t) \geq 0$. Thus
  \[
      g'(t) = D_p'(t) - \tfrac{1}{2} D_p'(t) - \tfrac{t}{2} D_p''(t) = \tfrac{1}{2} D_p'(t) - \tfrac{t}{2} D_p''(t) < 0
  \]
  for $t > 0$, meaning that $g(t)$ is strictly decreasing. Since $\lim_{t \to \infty} [-\tfrac{t}{2} D_p'(t)] = 0$ and $\lim_{t \to \infty} D_p(t) = 0$, we have $\lim_{t \to \infty} g(t) = 0$, it follows that for $0<t<\lim_{t \to 0} g(t)$, there is a unique $t^*(\delta)$, such that $0 < t^*(\delta) < \infty$ and $g(t) > \delta$ for $t < t^*$ and $g(t) < \delta$ for $t > t^*$. We conclude by observing that since $g(t) \geq D_{p}(t)$ we have $\lim_{t \to 0} g(t) \geq D_{p}(0)$.
\end{itemize}
\end{proof}

\begin{lemma}
\label{le:genericboundTp}
  Denote
  \begin{equation}\label{eq:DefTheta}
    \theta(t) \bydef \E (\abs{h}-t)_{+}^{2}
  \end{equation}
  with $h$ a standard centered normal variable and $(\cdot)_{+} = \max(.,0)$ the positive part. This is a strictly decreasing function of $t$ with $\theta(0)=1$ and $\lim_{t \to \infty}\theta(t) = 0$. The following holds for all $1 \leq p \le \infty$, $0<\delta<1$:
  \begin{enumerate}
    \item\label{it:TMax} For any $n \geq 1$
  \begin{equation}
    \label{eq:TMax}
    t^{*}_{p}(\delta;n) \leq  2\ \theta^{-1}(\delta)\  n^{1-1/p}.
  \end{equation}
  \item\label{it:TMin} For any $n \geq \tfrac{2}{1-\delta}$
  \begin{equation}\label{eq:TMin}
    t^{*}_{p}(\delta;n) \geq \theta^{-1}(\tfrac{1+\delta}{2}) > 0
  \end{equation}
  \item\label{it:TMin2} For any $n \geq  \max\left(\tfrac{2}{1-\delta},\tfrac{1}{\delta}\right)$
  \begin{equation}\label{eq:TMin2}
    t^{*}_{p}(\delta;n) \geq \max\left(\theta^{-1}(2\delta),\theta^{-1}(\tfrac{1+\delta}{2})\right) > 0
  \end{equation}
   \item\label{it:lowerBoundDp} There is a universal constant $C$ independent of $\delta$, $p$ and $n$ such that for all $n \geq 1$
  \begin{equation}
    \label{eq:lowerBoundDp}  
    D_{p}(t^{*}_{p};n)  \geq \left(\frac{\delta}{C}\right)^{2}\\
  \end{equation}
  where we use the shorthand $t^{*}_{p} = t^{*}_{p}(\delta;n)$.
  \end{enumerate}
\end{lemma}
\begin{lemma}\label{le:boundp2}
With the notations of Theorem~\ref{thm:frob_of_l1}, for $p=2$,
\begin{eqnarray}
\lim_{n \to \infty} -t^{*}_{2} D'_{2}(t^{*}_{2};n) & = & (1-\delta)\delta > 0\label{eq:UnivLowBoundDprime2}\\
\lim_{n \to \infty}\alpha^{*}_{2}(\delta;n) & = & \tfrac{1}{\sqrt{1-\delta}}.\label{eq:AlphaD2}
\end{eqnarray}
\end{lemma}

\begin{lemma}\label{le:boundp1}
 With the notations of Theorem~\ref{thm:frob_of_l1}, for $p=1$,
\begin{eqnarray}
&& \lim_{n \to \infty}  -t^{*}_{1} D'_{1}(t^{*}_{1};n)  >  0\label{eq:UnivLowBoundDprime1}\\
\alpha^{*}_{1}(\delta) & \bydef & \lim_{n \to \infty} \alpha^{*}_{1}(\delta;n) 
= \sqrt{\frac{1}{\sqrt{\tfrac{2}{\pi}} \e^{-\tfrac{(t_1^*)^2}{2}} t^{*}_{1} - \delta(t_1^*)^2} -\tfrac{1}{\delta}}.
\label{eq:AlphaD1}
\end{eqnarray}
where $t^*_{1}(\delta) \bydef \sqrt{2} \erfc^{-1}(\delta)$. 
\end{lemma}

\begin{proof}[Proof of Lemmas~\ref{le:genericboundTp}-\ref{le:boundp2}-\ref{le:boundp1}]
~
\begin{itemize}
  \item {\bf Step 1.} We prove that for all $p,n,t$
  \begin{equation}
  \label{eq:D1Dp}
    D_p(t;n)  \geq D_1(t;n) \geq D_p\left(t n^{1-1/p};n\right).
  \end{equation}
  The inequalities $\norm{\cdot}_{p^{*}} \geq \norm{\cdot}_\infty 
  \geq n^{-1/p^{*}} \norm{\cdot}_{p^{*}} = n^{-(1-1/p)} \norm{\cdot}_{p^{*}}$ imply the inclusions $\mathcal{C}_{t,p} \subset \mathcal{C}_{t,1} \subset \mathcal{C}_{t n^{1-1/p},p}$. It follows that $d(\cdot,\mathcal{C}_{t,p}) \geq d(\cdot,\mathcal{C}_{t,1}) \geq d(\cdot,\mathcal{C}_{t n^{1-1/p},p})$, which yields~\eqref{eq:D1Dp}.
  \item {\bf Step 2.} We establish that for any $p,n,\delta$, with the shorthand $t^{*}_{p} = t^{*}_{p}(\delta;n)$,
  \begin{eqnarray}
    \label{eq:lowerBoundDp1}
    D_{p}(t^{*}_{p}/2;n) &\geq& \delta\\
    \label{eq:upperBoundDp}  
    D_{p}(t^{*}_{p};n)  & \leq & \delta
  \end{eqnarray}
       
  With the additional shorthand $D_{p}(t) = D_{p}(t;n)$, since $D_{p}(t)$ is convex, we have for all $t,h$
  \begin{equation}
  \label{eq:DpConvex}
  D_{p}(t+h) \geq D_{p}(t) + h D'_{p}(t).
  \end{equation}
  Applying it to $t=t^{*}_{p}$ and $h=-t^{*}_{p}/2$ and using the definition of $t^{*}_{p}$, we get 
  \[
  D_{p}(t^{*}_{p}/2) \geq D_{p}(t^{*}_{p})-(t^{*}_{p}/2)D'_{p}(t^{*}_{p}) = \delta
  \] 
  i.e.~\eqref{eq:lowerBoundDp1} holds.  Since $D_{p}(t)$ is non-increasing, we have $D'_{p}(t) \leq  0$ and $D_{p}(t) -\tfrac{t}{2} D_p'(t) \geq D_{p}(t)$ for any $t$. Applying to $t=t^{*}_{p}$ this establishes~\eqref{eq:upperBoundDp} by definition of $t^{*}_{p}$. 
  
  \item {\bf Step 3.} We show that for any $p,n,t$
  \begin{equation}\label{eq:UpperBoundDpBarDp}
    D_p(t;n) \leq \bar{D}_{p}(t;n) \bydef \tfrac{1}{n} \E~\dist^2(\vh, \norm{\,\cdot\,}_{p^*} \leq t)
  \end{equation}
  and that for $p=1$
  \begin{equation} \label{eq:barD1theta}
   \bar{D}_{1}(t;n) = \theta(t)
  \end{equation}
  is independent of $n$.
  By the concavity of square root and Jensen's inequality we can write
  \[
    n D_{p}(t;n) =  \left(\E~\dist(\vh, \norm{\,\cdot\,}_{p^*} \leq t)\right)^{2} \leq \E~\dist^2(\vh, \norm{\,\cdot\,}_{p^*} \leq t) = n \bar{D}_{p}(t;n).
  \]
  This establishes~\eqref{eq:UpperBoundDpBarDp}.
    
  For $p = 1$, since $p^* = \infty$ we can compute
  \begin{align*}
    n \bar{D}_{1}(t) \bydef \E~\dist^2(\vh, \norm{\,\cdot\,}_{\infty} \leq t) 
    &= \E~\norm{\vh - \proj_{\norm{\,\cdot\,}_{\infty} \leq t} \vh}_{2}^2 = \E~\sum_{i = 1}^n (h_i - (\proj_{\norm{\,\cdot\,}_{\infty} \leq t} \vh)_i)^2 \\
    & = \sum_{i = 1}^n \E  (\abs{h_{i}} - t)_{+}^2 = n \theta(t).
  \end{align*}
    
  \item {\bf Step 4.} Combining the previous steps we get for any $p,\delta,n$:
  \begin{eqnarray*}
    \theta(t^{*}_{p} n^{1/p-1}/2) & \stackrel{\eqref{eq:barD1theta}}{=} & \bar{D}_{1}(t^{*}_{p} n^{1/p-1}/2)
     \stackrel{\eqref{eq:UpperBoundDpBarDp}}{\geq} D_{1}(t^{*}_{p} n^{1/p-1}/2) \stackrel{\eqref{eq:D1Dp}}{\geq} D_{p}(t^{*}_{p}/2) \stackrel{\eqref{eq:lowerBoundDp1}}{\geq} \delta
  \end{eqnarray*}
  This yields~\eqref{eq:TMax}.
  \item {\bf Step 5.} We show that for any $p,n,t$
  \begin{equation}\label{eq:LowerBoundDpBarDp}
    D_p(t;n) \geq \bar{D}_{p}(t;n) -\tfrac{1}{n}.
  \end{equation}
  Indeed, the function $f: \vh \mapsto f(\vh) \bydef \dist(\vh, \norm{\,\cdot\,}_{p^{*}} \leq t)$ is $1$-Lipschitz, and $\vh$ is a standard normal Gaussian variable, hence we can apply Lemma \ref{lem:concentration_results}\ref{lemitem:varoflip}.

  \item {\bf Step 6.} Combining with the previous steps yields for any $p,\delta$, and $n>\tfrac{1}{1-\delta}$:
  \begin{eqnarray*}
  \theta(t^{*}_{p}) & \stackrel{\eqref{eq:barD1theta}}{=} & \bar{D}_{1}(t^{*}_{p})
    \stackrel{\eqref{eq:LowerBoundDpBarDp}}{\leq} D_{1}(t^{*}_{p})+\tfrac{1}{n} 
    \stackrel{\eqref{eq:D1Dp}}{\leq} D_{p}(t^{*}_{p})+\tfrac{1}{n} \stackrel{\eqref{eq:upperBoundDp}}{\leq} \delta+\tfrac{1}{n} < 1
  \end{eqnarray*}
  For $n \geq\tfrac{2}{1-\delta}$, we have $\delta+\tfrac{1}{n} \leq \tfrac{1+\delta}{2} < 1$ which yields~\eqref{eq:TMin}. For $n \geq \max(\tfrac{2}{1-\delta},1/\delta)$, we have $\delta+\tfrac{1}{n} \leq \min(\tfrac{1+\delta}{2},2\delta) < 1$ which yields~\eqref{eq:TMin2}.
      
  \item {\bf Step 7.} To establish~\eqref{eq:lowerBoundDp} we  start with the expression~\eqref{eq:CpDevelopedExpression}, with the shorthand $q(t) \bydef \sqrt{n D_{p}(t)}$:
  \begin{equation}
    \label{eq:lbnd_Cp}
  \begin{aligned}
    -\tfrac{t}{2} D_p'(t)
    &= -\tfrac{1}{n} q(t) \left( (n+1) q(t) - Z_n t^{n+3} \int_{\vb \in \mathbb{S}^{n-1}} \int_{\rho = 0}^{\infty} \dist(\rho \vb, {\cal C}_1) \rho^{n+1} \e^{-\rho^2 t^2 / 2}  \, \mu(\di \vb) \, \di \rho \right) \\
    -\tfrac{t}{2} D_p'(t) + \tfrac{n+1}{n} q^{2}(t)    & = \frac{q(t)}{n} Z_n t^{n+3} \int_{\vb \in \mathbb{S}^{n-1}} \int_{\rho = 0}^{\infty} \dist(\rho \vb, {\cal C}_1) \rho^{n+1} \e^{-\rho^2 t^2 / 2}  \, \mu(\di \vb) \, \di \rho \\
     & \stackrel{r = t\rho}{=} \frac{q(t)}{n} \int_{\vb \in \mathbb{S}^{n-1}} \int_{\rho = 0}^{\infty} t \dist(r\vb/t, {\cal C}_1) r^{n+1} p_{\vh}(r \vb) \mu(\di \vb) \, \di r \\
     & \stackrel{\vy=r\vb}{=}  \frac{q(t)}{n} \int_{\R^n} t \dist(\vy / t, \mathcal{C}_1) \norm{\vy}^2 p_\vh(\vy) \di \vy \\
     & =  \frac{q(t)}{n} \E [\dist(\vh, \mathcal{C}_t) \norm{\vh}^2].
  \end{aligned}
  \end{equation}
  Since $\tfrac{n+1}{n} q^2(t) = (n+1)D_p(t)$ we can rewrite \eqref{eq:lbnd_Cp} using the definition of $D_p(t)$ as
  \begin{equation}
    \label{eq:Dp_plus_Cp}
    D_p(t) -\tfrac{t}{2} D_p'(t)  = \frac{q(t)}{n} \E[\dist(\vh, {\cal C}_t) (\norm{\vh}^2 - n)].
  \end{equation}
  Observing further that $\E [\norm{\vh}^2 - n] = 0$ and $\E [\dist(\vh, {\cal C}_t)] = q(t)$, the following holds:
  \begin{eqnarray*}
    D_p(t) -\tfrac{t}{2} D_p'(t) 
    & = & \frac{q(t)}{n}\ \E\left[(\dist(\vh, {\cal C}_t) - q(t)) (\norm{\vh}^2 - n)\right] \\
    & = & \sqrt{D_{p}(t)}\ \E\left[\left(\dist(\vh, {\cal C}_t) - q(t)\right) \frac{\norm{\vh}^2 - n}{\sqrt{n}} \right] \\
    & \leq & \sqrt{D_{p}(t)}\ \E\left[\bigg|\dist(\vh, {\cal C}_t) - q(t)\bigg| \abs{\frac{\norm{\vh}^2 - n}{\sqrt{n}}} \right] \\
    & = & \sqrt{D_{p}(t)}\  \int_0^{\infty} \prob \left[ \bigg|\dist(\vh, {\cal C}_t) - q(t)\bigg| \abs{\frac{\norm{\vh}^2 - n}{\sqrt{n}}} \geq \epsilon \right] \di \epsilon \\
  \end{eqnarray*}
  The integrand can be controlled by a union bound as
  \[
    \prob \left[ \bigg|\dist(\vh, {\cal C}_t) - q(t)\bigg| \abs{\frac{\norm{\vh}^2 - n}{\sqrt{n}}} \geq \epsilon \right]
    \leq 
    \prob \left[ \bigg|\dist(\vh, {\cal C}_t) - q(t)\bigg| \geq \sqrt{\epsilon} \right]
    +
    \prob \left[ \abs{\frac{\norm{\vh}^2 - n}{\sqrt{n}}} \geq \sqrt{\epsilon} \right] 
  \]
  which together with Lemma \ref{lem:concentration_results}\ref{lemitem:lipgausabs} for the first term and Lemma \ref{lem:concentration_results}\ref{lemitem:norm2abs} for the second term yields $D_p(t) -\tfrac{t}{2} D_p'(t) \leq C_{n} \sqrt{D_{p}(t)}$ with
  \begin{eqnarray*}
    C_{n} & \bydef &  \int_{0}^\infty 2\left(\e^{-\epsilon / 2} +
    \begin{cases}
      \e^{-\frac{\epsilon}{8}} & \text{for $0 \leq \epsilon \leq n$}, \\
      \e^{-\frac{\sqrt{\epsilon n}}{8}} &\text{for $\epsilon > n$}
    \end{cases}
    \right) \di \epsilon
        \\
        &\leq& \int_{0}^\infty 2\left(\e^{-\epsilon / 2} + \e^{-\frac{\epsilon}{8}} + \e^{-\frac{\sqrt{\epsilon n}}{8}}  
        \right) \di \epsilon
    \\
        &\leq& \int_{0}^\infty 2\left(\e^{-\epsilon / 2} + \e^{-\frac{\epsilon}{8}} + \e^{-\frac{\sqrt{\epsilon}}{8}}  
        \right) \di \epsilon 
        \bydef C < \infty.
  \end{eqnarray*}
\item {\bf Step 8. (Proof of Lemma~\ref{le:boundp2})}. Since $\dist(\vh,\norm{\cdot}_{p^{*}} \leq t) = (\norm{\vh}_{2}-t)_{+}$, we have $D_{2}(t;n) = \tocheck{\tfrac{1}{n}} \left(\E (\norm{\vh}_2-t)_+\right)^{2}$. With $d(t) \bydef \sqrt{D_{2}(t;n)} = \tfrac{1}{\sqrt{n}} \E (\norm{\vh}_2-t)_+$ we have
\begin{eqnarray*}
d'(t) &=& -\tfrac{1}{\sqrt{n}} \E [ \ind{\norm{\vh}_2>t} ] = -\tfrac{1}{\sqrt{n}} \prob(\norm{\vh}_2>t)
\end{eqnarray*}
and $-(t/2)D'_{2}(t) = -t d(t) d'(t)$. With a change of variables $\tau = t/\sqrt{n}$, define
\begin{eqnarray*}
F_{1}(\tau;n) & \bydef & D_{2}(\tau\sqrt{n};n)  =   \left\{\E \left(\tfrac{\norm{\vh}_2}{\sqrt{n}} -\tau\right)_+\right\}^{2} \\
F_{2}(\tau;n) & \bydef & -\tfrac{\tau\sqrt{n}}{2}D'_{2}(\tau\sqrt{n};n) =  \E \left(\tfrac{\norm{\vh}_2}{\sqrt{n}} -\tau\right)_+ \cdot \tau \prob\left(\tfrac{\norm{\vh}_2}{\sqrt{n}}>\tau\right)\\
F(\tau;n) & \bydef & F_{1}(\tau;n)+F_{2}(\tau;n)
\end{eqnarray*}
Since $\norm{\vh}_{2}$ concentrates around $\sqrt{n}$ for large $n$, as we show below for any $0<\tau<1$
\begin{eqnarray}
\lim_{n \to \infty} F_{1}(\tau;n) \bydef F_1(\tau) &=& (1-\tau)^{2}\label{eq:limF1}\\
\lim_{n \to \infty} F_{2}(\tau;n) \bydef F_2(\tau) &=& (1-\tau)  \tau\label{eq:limF2}
\end{eqnarray}
It follows that for $\delta \in (0,1)$ we have $\lim_{n \to \infty} F(1-\delta;n)=\delta$ and 
\begin{eqnarray}
\lim_{n \to \infty} \tfrac{t^{*}_{2}(\delta;n)}{\sqrt{n}} & = & 1-\delta\notag\\ 
\lim_{n \to \infty} -(t^{*}_{2}/2) D'_{2}(t^{*}_{2};n) & = & F_{2}(\delta) = (1-\delta)\delta\notag\\
\lim_{n \to \infty }D(t^{*}_{2};n) & = & F_{1}(\delta) = \delta^{2}\notag\\
\lim_{n \to \infty}\alpha^{*}_{2}(\delta;n) & = & \sqrt{\tfrac{\delta^{2}}{\delta(\delta-\delta^{2})}} = 1/\sqrt{1-\delta}.\notag
\end{eqnarray}
which establishes~\eqref{eq:UnivLowBoundDprime2}-\eqref{eq:AlphaD2}. 

To turn these estimates into mathematics we compute for $0<\tau<1$, and $\epsilon>0$
\begin{eqnarray*}
\E \left(\tfrac{\norm{\vh}_2}{\sqrt{n}}-\tau\right)_+ 
&=& 
\E \left\{ \left(\tfrac{\norm{\vh}_2}{\sqrt{n}}-\tau\right)_+ 
\ \Big|\ 
\abs{\tfrac{\norm{\vh}_2}{\sqrt{n}} - 1} > \epsilon \right\} 
\prob\left( \abs{\tfrac{\norm{\vh}_2}{\sqrt{n}} - 1} > \epsilon\right)\\ 
&& + \E \left\{ (\tfrac{\norm{\vh}_2}{\sqrt{n}} -\tau)_+ 
\ \Big|\ 
\abs{\tfrac{\norm{\vh}_2}{\sqrt{n}} -1} \leq \epsilon \right\} 
\prob\left( \abs{\tfrac{\norm{\vh}_2}{\sqrt{n}} -1} \leq \epsilon\right)\\
\prob\left(\tfrac{\norm{\vh}_2}{\sqrt{n}}>\tau\right)
&=& 
\prob\left(\tfrac{\norm{\vh}_2}{\sqrt{n}}>\tau 
\ \Big|\ 
\abs{\tfrac{\norm{\vh}_2}{\sqrt{n}} - 1} > \epsilon \right) 
\prob\left( \abs{\tfrac{\norm{\vh}_2}{\sqrt{n}} - 1} > \epsilon\right)\\ 
&& + \prob \left( \tfrac{\norm{\vh}_2}{\sqrt{n}} >\tau
\ \Big|\ 
\abs{\tfrac{\norm{\vh}_2}{\sqrt{n}} -1} \leq \epsilon \right) 
\prob\left( \abs{\tfrac{\norm{\vh}_2}{\sqrt{n}} -1} \leq \epsilon\right)
\end{eqnarray*}
For any $0<\epsilon < \min(1-\tau, \tfrac{1}{2})$ we get 
\begin{eqnarray*}
    \prob\left( \abs{\tfrac{\norm{\vh}_2}{\sqrt{n}} - 1}  \leq \epsilon\right) & \stackrel{(a)}{\geq} & 1-2\e^{-c n \epsilon^2}\\
    \E \left\{ \left(\tfrac{\norm{\vh}_2}{\sqrt{n}}-\tau\right)_+ 
    \ \Big|\ 
    \abs{\tfrac{\norm{\vh}_2}{\sqrt{n}}-1} \leq \epsilon
    \right\}
    &\geq&  1-\tau-\epsilon > 0\\
    \prob \left(\tfrac{\norm{\vh}_2}{\sqrt{n}}>\tau
    \ \Big|\ 
    \abs{\tfrac{\norm{\vh}_2}{\sqrt{n}}-1} \leq \epsilon
    \right)
    &=&1\\
    \E \left(\tfrac{\norm{\vh}_2}{\sqrt{n}}-\tau\right)_+ &\geq&  (1 - \tau-\epsilon) (1 - 2e^{-c n \epsilon^2})\notag\\
    \prob\left(\tfrac{\norm{\vh}_2}{\sqrt{n}}>\tau\right) & \geq & (1-2e^{-c n \epsilon^2})\notag
\end{eqnarray*}
where $(a)$ follows from Lemma \ref{lem:concentration_results}\ref{lemitem:gausnormabs} by noting that for $0 < \epsilon \leq \tfrac{1}{2}$, 
$\sqrt{\tfrac{1}{1 - \epsilon}} \leq 1 + \epsilon$. Hence, with $\epsilon = (1-\tau)/n^{1/4}$, 
\begin{eqnarray}
F_{1}(\tau;n) & \geq & (1-\tau)^{2} (1-n^{-1/4})^{2} (1-2e^{-c (1-\tau)^{2} n^{1/2}})^{2}\label{eq:lowerBoundF1}\\
F_{2}(\tau;n) & \geq & \tau(1-\tau) (1-n^{-1/4}) (1-2e^{-c (1-\tau)^{2}n^{1/2}})^{2}\label{eq:lowerBoundF2}
\end{eqnarray}

For an upper bound, denote $c_n \bydef \tfrac{1}{\sqrt{n}} \E \norm{\vh}_2$. Since $\tfrac{n}{\sqrt{n+1}} \leq \E \norm{\vh}_{2} \leq \sqrt{n}$ (cf~\eqref{eq:lowerBoundExpNormH} for the lower bound, Jensen's inequality for the upper bound) we have $1 \geq c_{n} \geq \sqrt{\tfrac{n}{n+1}} \geq 1 - \tfrac{1}{2n}$ so that $0 \leq 2(1 - c_n) \leq 1/n$. By Jensen's inequality, for $0<\tau<1$,
\begin{eqnarray*}
\E \left(\tfrac{\norm{\vh}_2}{\sqrt{n}}-\tau\right)_+ 
&\leq &
\sqrt{\E \left(\tfrac{\norm{\vh}_2}{\sqrt{n}}-\tau\right)_+^2}
\leq
\sqrt{\E \left(\tfrac{\norm{\vh}_2}{\sqrt{n}}-\tau\right)^2}
=
\sqrt{1-2c_n \tau + \tau^2}\\
&=& \sqrt{ (1-\tau)^2+2(1-c_n)\tau}
 \leq  \sqrt{(1-\tau)^{2} + \tau/n} = (1-\tau) \sqrt{1+\tfrac{\tau}{(1-\tau)^{2}n}} 
\end{eqnarray*}
so that 
\begin{eqnarray}
F_{1}(\tau;n) &\leq& (1-\tau)^{2} 
+\tau/n 
\label{eq:upperBoundF1}\\
F_{2}(\tau;n) &\leq& \tau(1-\tau) \sqrt{1+\tfrac{\tau}{(1-\tau)^{2}n}}
\leq \tau (1-\tau) \left(1+\tfrac{\tau}{(1-\tau)^{2}n}\right)\label{eq:upperBoundGtau}
\end{eqnarray}
Combining all of the above yields~\eqref{eq:limF1}-\eqref{eq:limF2}.
  \item {\bf Step 9. (Proof of Lemma~\ref{le:boundp1})}.
By~\eqref{eq:TMax} and~\eqref{eq:TMin} we have for any $n \geq 2/(1-\delta)$
    \[
    0<t_{\min}(\delta) \bydef \theta^{-1}(\tfrac{1+\delta}{2}) \leq t^{*}_{1}(\delta;n) \leq 2\theta^{-1}(\delta) \bydef t_{\max}(\delta).
    \]    
    By the continuity of $\theta$ and its strict monotonicity, we have
\begin{equation}\label{eq:DefAuxiliaryTheta}
V(\delta) \bydef \inf_{t \in [t_{\min}(\delta),t_{\max}(\delta)]} 
\left\{
\theta(t)-\theta(2t)
\right\} >0
\end{equation}
Applying~\eqref{eq:DpConvex} (convexity of $D_{1}(t)$) with $h= t = t^{*}_{1}$ yields 
$D_{1}(2t^{*}_{1}) \geq D_{1}(t^{*}_{1}) + t^{*}_{1} D'_{1}(t^{*}_{1})$, hence
\begin{eqnarray*}
-t^{*}_{1}D'_{1}(t^{*}_{1};n) 
\geq D_{1}(t^{*}_{1};n)-D_{1}(2t^{*}_{1};n) 
&\stackrel{\eqref{eq:UpperBoundDpBarDp} \&\eqref{eq:LowerBoundDpBarDp}}{\geq}&
 \bar{D}_{1}(t^{*}_{1};n)-\bar{D}_{1}(2t^{*}_{1};n)-\tfrac{1}{n}\\
&\stackrel{\eqref{eq:barD1theta}}{=}& 
\theta(t^{*}_{1})-\theta(2t^{*}_{1}) -\tfrac{1}{n}\\
&\geq& V(\delta) -\tfrac{1}{n}.
\end{eqnarray*}
For $n \geq N(\delta) \bydef \max(2/V(\delta),2/(1-\delta))$ we obtain $-t^{*}_{1}D'_{1}(t^{*}_{1};n) \geq V(\delta)/2 \bydef \gamma(\delta) > 0$ which establishes~\eqref{eq:UnivLowBoundDprime1}.

By~\eqref{eq:UpperBoundDpBarDp}-\eqref{eq:barD1theta}-\eqref{eq:LowerBoundDpBarDp} we have $\theta(t)-1/n \leq D_{1}(t;n) \leq \theta(t)$ for all $t$ and $n$. Hence, the sequence of convex differentiable functions $\{D_{1}(\, \cdot \, ;n)\}_{n}$ converges uniformly to the convex and smooth function $\theta(t)$. It is a classical exercise in convex analysis to show that this implies the convergence of derivatives: indeed, since $D_{1}(\cdot;n)$ is convex we have for $h_{n} = n^{-1/2}$ with $n>t^{-2}$ (so that $t-h_{n}>0$):
\begin{eqnarray*}
D'_{1}(t;n) 
& \leq & \frac{D_{1}(t+h_{n};n)-D_{1}(t;n)}{h_{n}} 
\leq \frac{\theta(t+h_{n})-\theta(t)+\tfrac{1}{n}}{h_{n}} 
= \frac{\theta(t+h_{n})-\theta(t)}{h_{n}}+\tfrac{1}{\sqrt{n}} \\
D'_{1}(t;n) 
& \geq & \frac{D_{1}(t;n)-D_{1}(t-h_{n};n)}{h_{n}} 
\geq \frac{\theta(t)-\theta(t-h_{n})-\tfrac{1}{n}}{h_{n}}
= \frac{\theta(t)-\theta(t-h_{n})}{h_{n}}-\tfrac{1}{\sqrt{n}} \\
\end{eqnarray*}
It follows that $\lim_{n \to \infty} D'_{1}(t;n) = \theta'(t)$. Just as for the case $p=2$ we have shown that
$\lim_{n \to \infty} \left\{D_{1}(t;n)-(t/2) D'_{1}(t;n)\right\} = \theta(t)-(t/2)\theta'(t)$.
By Lemma~\ref{lem:theta_explicit}, $\theta(t)-(t/2)\theta'(t) = \erfc(t / \sqrt{2})$ so the unique $t = t^*_{1}(\delta)$ such that $\theta(t)-(t/2)\theta'(t) = \delta$ is 
\[
t^{*}_{1}(\delta) = \sqrt{2}\erfc^{-1}(\delta).
\]
With the same reasoning as for the case $p=2$ above we get that
\begin{eqnarray*}
\lim_{n \to \infty} t^{*}_{1}(\delta;n) & = & t^{*}_{1}(\delta)\\
\lim_{n \to \infty} -t^{*}_{1} D'_{1}(t^{*}_{1};n) & = & -\tfrac{t^*_{1}(\delta)}{2} \theta'(t^*_{1}(\delta)) > 0\\
\lim_{n \to \infty} D_{1}(t^{*}_{1};n) & = & \theta(t^*_{1}(\delta))\\
\lim_{n \to \infty} \alpha^{*}_{1}(\delta;n) & = & \sqrt{\frac{\theta(t^*_{1}(\delta))}{\delta(\delta-\theta(t^*_{1}(\delta)))}}.
\end{eqnarray*}
Since $\erfc(t^{*}_{1}/\sqrt{2}) = \delta$ we have $\theta(t^{*}_{1}) = \delta -\left(\sqrt{\tfrac{2}{\pi}} \e^{-\tfrac{(t^{*}_{1})^{2}}{2}} t^{*}_{1} - \delta (t^{*}_{1})^{2}\right)$ and we conclude that
\begin{equation}
    \lim_{n \to \infty} \alpha_{1}^{*}(\delta;n) = \sqrt{\frac{\delta-(\delta-\theta(t^*_{1}))}{\delta(\delta-\theta(t^*_{1}))}} 
    = \sqrt{\frac{1}{\sqrt{\tfrac{2}{\pi}} \e^{-\tfrac{(t_1^*)^2}{2}} t^{*}_{1} - \delta(t_1^*)^2} -\tfrac{1}{\delta}}
    \end{equation}
\end{itemize}
\end{proof}

\begin{lemma}[Deterministic properties of $\kappa(\alpha, \beta)$]\label{le:ArgMinKappa}
Consider $1 \leq m<n$ two integers, $\delta \bydef (m-1)/n$, $1 \leq p \leq \infty$, $\kappa(\alpha,\beta)$ defined in~\eqref{eq:DefKappa}, $D_{p}(t)$ defined in \eqref{eq:DefDp}.
The following hold:
\begin{enumerate}
\item The function $\kappa(\alpha,\beta)$ is convex-concave and proper on $[0,\infty) \times [0,\infty)$, hence the function 
\begin{equation}
\kappa(\alpha) \bydef \sup_{0 \leq \beta \leq 1} \kappa(\alpha,\beta)
\end{equation}
is convex on $[0,\infty)$, and for any $A>0$ the function
\begin{equation}
\underline{\kappa}_{A}(\beta) \bydef \inf_{0 \leq \alpha \leq A} \kappa(\alpha,\beta)
\end{equation}
is concave on $[0,\infty)$.

\item The scalar $t^{*} = t^{*}_{p}(\delta;n)$ (cf. Lemma~\ref{lem:Dp_det_prop}-Property 7) is well defined, with $D_{p}(t^{*}) < \delta$.
\item Define
\begin{equation}\label{eq:kappaminimizerexplicit}
\alpha^{*} = \alpha^{*}(\delta;n) \bydef \sqrt{\tfrac{D_{p}(t^{*})}{\delta(\delta-D_{p}(t^{*}))}}
\end{equation} 
For $A > \alpha^{*}$ and $\lambda \leq t^{*}$ we have
\begin{equation}\label{eq:kappaminimizer}
\argmin_{\alpha : 0 \leq \alpha \leq A} \max_{\beta: 0 \leq \beta \leq 1} \kappa(\alpha,\beta) 
= \alpha^{*}.
\end{equation} 
The corresponding optimal $\beta$ is $\beta^{*} = \beta^{*}(\lambda,\delta;n) \bydef \lambda/t^{*}$.

\item For $A>\alpha^{*}$, $\lambda \leq t^{*}$, $0<\epsilon \leq \max(\alpha^{*},A-\alpha^{*})$ we have 
\begin{equation}\label{eq:DefOmega}
\inf_{\abs{\alpha-\alpha^{*}} \geq \epsilon} \kappa(\alpha)-\kappa(\alpha^{*}) 
= \inf_{\abs{\alpha-\alpha^{*}}=\epsilon} \kappa(\alpha)-\kappa(\alpha^{*}) \geq 
\omega(\epsilon) = \omega_{p}(\epsilon;n,\delta,\lambda) \bydef
\tfrac{\epsilon^2}{2} 
\tfrac{\lambda \delta/t^{*}}{(1 + \delta(\alpha^* + \epsilon)^2)^{3/2}}.
\end{equation}
Observe that for the considered range of $\lambda$ and $\epsilon$, we have $\omega(\epsilon) \leq 1/2$.
\end{enumerate}
\end{lemma}

\begin{proof}[Proof of Lemma~\ref{le:ArgMinKappa}]
\hfill
\begin{enumerate}
\item It is obvious that $\kappa$ is proper. One easily checks that $\alpha \mapsto \sqrt{\delta \alpha^{2}+1}$ is convex by checking the non-negativity of its second derivative, hence $\kappa$ is convex in $\alpha$. The concavity in $\beta$ follows from the convexity of $\beta \mapsto \Delta_{p}(\beta;\vh,\lambda)$ which is a distance to a convex set \cite[Example 3.16]{Boyd:2004uz}, and the fact that the expectation $\Delta_{p}(\beta,\lambda)$ of a convex function is convex.
As a result, $\kappa(\alpha)$ is convex and $\underline{\kappa}(\beta)$ is concave.
\item We have $\delta \bydef (m - 1) / n < 1$, and by Lemma~\ref{lem:Dp_det_prop}, Property 4, $D_{p}(0) \geq \tfrac{n}{n+1}$. Because we consider the underdetermined case, $1 \leq m < n$, we have $n \geq 2$ and
\[
\delta \leq (n - 2) / n \leq n / (n + 1) \leq D_p(0).
\]

\item Since $\kappa(\alpha,\beta)$ is convex-concave and proper, and the constraint sets in~\eqref{eq:kappaminimizer} are convex and compact, we can change the order of maximization and minimization \cite[Corollary 3.3]{Sion:1958jm}. 

\begin{itemize}
\item Consider the minimization over $\alpha$ first. 

For $\beta = 0$ and any $\lambda \geq 0$ we have $\Delta_{p}(\beta,\lambda) = 0$ hence $\underline{\kappa}_{A}(0) = \inf_{0 \leq \alpha \leq A} \kappa(\alpha,0) = 0$.

For $\beta>0$, observing that $\Delta_{p}(\beta,\lambda) = \beta \sqrt{D_{p}(\lambda/\beta)}$ we rewrite 
\[
\kappa(\alpha,\beta) = \beta(\sqrt{\delta\alpha^{2}+1}-\alpha \sqrt{D_{p}(\lambda/\beta)}).
\]
With $D_{p}^{-1}(y) \bydef \inf \set{t: D_{p}(t) \leq y}$ the reciprocal of the strictly decreasing function $D_{p}$, 
we have  
\[
D_{p}(\lambda/\beta) \leq \frac{A^{2}\delta^{2}}{1+\delta A^{2}}
\]
if and only if $0<\beta \leq \tilde{\beta} \bydef \tfrac{\lambda}{D_{p}^{-1}\left(A^{2}\delta^{2}/(1+\delta A^{2})\right)}$. Since $\tfrac{A^{2}\delta^{2}}{1+\delta A^{2}} < \delta$, this implies that for $0<\beta \leq \tilde{\beta}$ we can define
\begin{equation}
     \tilde{\alpha}(\beta) \bydef \sqrt{\frac{D_p(\lambda/\beta)}{\delta(\delta - D_{p}(\lambda/\beta))}}.
\end{equation}
and check that $\tilde{\alpha}(\beta) \leq A$. By studying the sign of $\parder{\kappa(\alpha,\beta)}{\alpha} = \beta\left(\tfrac{\delta\alpha}{\sqrt{\delta \alpha^{2}+1}}-\sqrt{D_{p}(\lambda/\beta)}\right)$, we get that $\alpha \mapsto\kappa(\alpha,\beta)$ has a unique minimizer on $[0,A]$ which is precisely $\alpha^*(\beta) = \tilde{\alpha}(\beta)$. It follows that for $0<\beta \leq \tilde{\beta}$ we have:
 \begin{equation}\label{eq:ArgMinAlphaGivenBeta}
   \underline{\kappa}_{A}(\beta) = \min_{0 \leq \alpha \leq A} \kappa(\alpha,\beta) = \kappa(\alpha^*(\beta), \beta) = \beta \sqrt{\frac{\delta - D_p(\lambda/\beta)}{\delta}}. 
\end{equation}

\item Consider now the maximization over $\beta$.  For $0 < \beta \leq \tilde{\beta}$ the expression of $\underline{\kappa}_{A}(\beta)$ is given by~\eqref{eq:ArgMinAlphaGivenBeta}.
The sign of $\underline{\kappa}'_{A}(\beta)$ is that of $\delta-g(\lambda/\beta)$ with $g(t) \bydef D_{p}(t)-\tfrac{t}{2}D'_{p}(t)$ as in Lemma~\ref{lem:Dp_det_prop}-Property~7. Hence, we have: $\underline{\kappa}'_{A}(\beta) >0$ if $t=\lambda/\beta>t^{*}$ (that is to say if $\beta < \beta^{*} \bydef \lambda/t^{*}$); $\underline{\kappa}'_{A}(\beta) <0$ if $\beta > \beta^{*}$; and $D_{p}(t^{*}) < \delta$.

We check that $A>\alpha^{*}$ implies $D_{p}(t^{*}) < A^{2}\delta^{2}/(1+\delta A^{2})$ i.e. $\beta^{*} < \tilde{\beta}$. Combined with the fact that $\underline{\kappa}(0) = 0$, this shows that the supremum of $\underline{\kappa}(\beta)$ over $[0,\tilde{\beta}]$ is indeed achieved uniquely at $\beta^{*}$. This also implies that $\underline{\kappa}(\beta)$ is strictly decreasing for $\beta^{*} < \beta \leq \tilde{\beta}$. Being concave, $\underline{\kappa}_{A}(\beta)$ must be also strictly decreasing for $\beta \geq \tilde{\beta}$, so the supremum over $[0,\infty)$ is indeed achieved at $\beta^{*}$.
Since $\lambda \leq t^{*}$ we further have $\beta^{*} \leq 1$ hence this is also the supremum over $\beta \in [0,1]$.
\end{itemize}
To summarize, the optimal $\beta$ is $\beta^{*} = \lambda/t^{*}$, and the corresponding optimal $\alpha$ is given as
\begin{equation}
    \alpha^*(\beta^*) = \tilde{\alpha}(\lambda/t^{*}) = \sqrt{\frac{D_p(t^*)}{\delta(\delta - D_p(t^*))}}.
\end{equation}

\item The assumption $\epsilon \leq \max(\alpha^{*},A-\alpha^{*})$ ensures that the set $\set{\alpha: \abs{\alpha-\alpha^{*}} \geq \epsilon,\ 0 \leq \alpha \leq A}$ is not empty. Since $\kappa(\alpha)$ is convex on $[0,\infty)$ with its minimum at $\alpha^{*}$, we have
    \begin{align*}
       \omega(\epsilon)
       &\bydef \inf_{\substack{\alpha : \abs{\alpha-\alpha^*} \geq \epsilon\\ 0 \leq \alpha \leq A}} \kappa(\alpha)-\kappa(\alpha^{*})
       = \min_{\substack{\alpha : \abs{\alpha - \alpha^*} = \epsilon\\ 0 \leq \alpha \leq A}} \kappa(\alpha) - \kappa(\alpha^*)
    \end{align*}

    Since $A>\alpha^{*}$ and $\lambda \leq t^{*}$, we have $0 < \beta^* < 1$. The second derivative of $\alpha \mapsto \kappa(\alpha, \beta^*)$ with respect to $\alpha$ reads:
    \[
        \parderr{\kappa(\alpha, \beta^*)}{\alpha} = \frac{\beta^* \delta}{(1 + \delta \alpha^2)^{3/2}} > 0.
    \]
    This implies that on $[0,A] \cap [\alpha^* - \epsilon, \alpha^* + \epsilon]$ the function $\alpha \mapsto \kappa(\alpha, \beta^*)$ is strongly convex with strong convexity modulus $\frac{\beta^* \delta}{(1 + \delta(\alpha^* + \epsilon)^2)^{3/2}}$. Since $\alpha \mapsto \kappa(\alpha,\beta^{*})$ is minimum at $\alpha^{*}$, it holds that
    \[
        \kappa(\alpha^* \pm \epsilon, \beta^*) \geq \kappa(\alpha^*, \beta^*) + \frac{\epsilon^2}{2} \frac{\beta^* \delta}{(1 + \delta(\alpha^* + \epsilon)^2)^{3/2}}
    \]

    Furthermore, from the definition of $\kappa(\alpha)$ and $\beta^*$, we have that $\kappa(\alpha) \geq \kappa(\alpha, \beta^*)$ for any $\alpha$, with equality for $\alpha = \alpha^*$. The claim therefore follows.
    
    Since $\lambda \leq t^{*}$ and $\epsilon \leq \alpha^{*}$ we have $\omega(\epsilon) \leq \tfrac{\delta (\alpha^{*})^{2}}{2(1+\delta (\alpha^{*})^{2})^{3/2}} = f\left((1+\delta(\alpha^{*})^{2})^{-1/2}\right)$ with $f(u) \bydef \tfrac{1}{2} (1/u^{2}-1)u^{3} = \tfrac{u-u^{3}}{2} \leq 1/2$ for $0 < u \leq 1$. Hence, $\omega(\epsilon) \leq 1/2$.
        \end{enumerate}
\end{proof}
Invoking a lemma from \cite{Hjort:1993tq} we show that $\argmin_\alpha \phi(\alpha)$ concentrates around $\argmin_{\alpha}
\kappa(\alpha)$.

\begin{lemma}[{\cite[Lemma 2]{Hjort:1993tq}}] \label{le:hjortpollard}
Let $f(t)$ be a random convex function on some
    open set $S \subset \R^p$, and let $t_f$ be (one of) its minimizer(s). Consider
    another function $g(t)$ (which we interpret as approximating $f$), such
    that it has a unique argmin $t_g$. Then for each $\epsilon > 0$, we have
    that:
    \begin{equation}
        \label{eq:hjortpollard}
        \prob[\norm{t_f - t_g} \geq \epsilon] \leq \prob[\sup_{\norm{s - t_g}
        \leq \epsilon} \norm{f(s) - g(s)} \geq \tfrac{1}{2} \inf_{\norm{s - t_g} = \epsilon} (g(s) - g(t_g))]
    \end{equation}
\end{lemma}
The role of $f(t)$ (resp. $t_{f}$ will be played by $\phi(\alpha)$ (resp $\alpha^{*}_{\phi}$), and the role of $g(t)$ (resp. $t_{g}$) by $\kappa(\alpha)$ (resp. $\alpha^{*}$).
\begin{lemma}\label{le:empiricalminmax}
        Let $A > \alpha^{*}$, $\lambda \leq t^{*}$, with $\alpha^{*}$ defined as in Lemma~\ref{le:ArgMinKappa}-Equation~\eqref{eq:kappaminimizer} and $t^{*} = t^{*}_{p}(\delta;n)$ as in Lemma~\ref{lem:Dp_det_prop}-Property 7. Consider the random function $\phi(\alpha)$ defined as in Corollary~\ref{cor:empiricalmaxbeta}, and
        \begin{equation}
        \alpha^*_\phi \bydef \argmin_{0 \leq \alpha \leq A} \phi(\alpha).
        \end{equation} 
        For $0<\epsilon\leq \max(\alpha^{*},A-\alpha^{*})$ 
       we have
        \begin{equation}
            \prob \left[ \abs{\alpha^*_\phi - \alpha^*} \geq \epsilon \right] 
            \leq\ \zeta\big(n,\tfrac{\omega(\epsilon)}{2}\big) 
        \end{equation}
        with $\zeta(n,\epsilon) = \zeta(n,\epsilon;A,\delta)$ defined in~\eqref{eq:DefZeta} and $\omega(\epsilon)$ defined in~\eqref{eq:DefOmega}.
\end{lemma}

\begin{proof}
Since $0<\epsilon\leq \max(\alpha^{*},A-\alpha^{*})$ the set $\set{\alpha: 0 \leq \alpha \leq A,\ \abs{\alpha-\alpha^{*}} \leq \epsilon}$ is a non-empty subset of $[0,\ A]$, and 
\[
\sup_{\alpha: 0 \leq \alpha \leq A,\  \abs{\alpha - \alpha^{*}} \leq \epsilon} \abs{\phi(\alpha) - \kappa(\alpha)}
\leq 
\sup_{0 \leq \alpha \leq A} \abs{\phi(\alpha) - \kappa(\alpha)}.
\]
Since $\phi$ is a random convex function, we can apply Lemma~\ref{le:hjortpollard} to obtain
\begin{eqnarray*}
\prob \left[ \abs{\alpha^{*}_{\phi}-\alpha^{*}} \geq \epsilon\right]
& \leq &    
\prob \left[ \sup_{\alpha: 0 \leq \alpha \leq A,\ \abs{\alpha - \alpha^{*}} \leq \epsilon} \abs{\phi(\alpha) - \kappa(\alpha)} \geq  
\tfrac{\omega(\epsilon)}{2} 
\right]
    \leq
    \prob \left[ \sup_{0 \leq \alpha \leq A} \abs{\phi(\alpha) - \kappa(\alpha)} \geq \tfrac{\omega(\epsilon)}{2} \right]\\
    & \leq & \zeta(n,\omega(\epsilon)/2),
\end{eqnarray*}
where we used that $\omega(\epsilon) \leq \inf_{\alpha: 0 \leq \alpha \leq A,\ \abs{\alpha-\alpha^{*}}=\epsilon} \kappa(\alpha)-\kappa(\alpha^{*})$, and the last inequality follows from Corollary~\ref{cor:empiricalmaxbeta} -- which we can use since $\omega(\epsilon)/2 \leq 1/4 < 2$. 
\end{proof}
  
\begin{lemma}
    \label{lem:phi_eps_separation}
    Let $A > \alpha^{*}$, $\lambda \leq t^{*}$ with $\alpha^{*}$ defined as in Lemma~\ref{le:ArgMinKappa}-Equation~\eqref{eq:kappaminimizer} and $t^{*} = t^{*}_{p}(\delta;n)$ as in Lemma~\ref{lem:Dp_det_prop}-Property 7. 
    For $0 \leq \epsilon\leq \max(\alpha^{*},A-\alpha^{*})$, consider the optimal cost of the auxiliary optimization~\eqref{eq:aux_dual_swap} with an altered order of minimization and
    maximization and $\norm{\vz}_{2}$ further restricted to be $\epsilon$-away from $\alpha^{*}$
        \begin{equation}
        \label{eq:opt_for_phi_eps}
        \phi_{\epsilon}(\vg, \vh) \bydef \max_{\substack{\beta : 0 \leq \beta \leq 1\\\vw:\norm{\vw}_{p^*} \leq 1}} \min_{\substack{\vz : \norm{\vz} \leq A\\\norm{\vz} \notin (\alpha^* - \epsilon, \alpha^* + \epsilon)}}  \max_{\vu:\norm{\vu} = \beta} \frac{1}{\sqrt{n}} \norm{\vz} \vg^\T \vu - \frac{1}{\sqrt{n}} \norm{\vu} \vh^\T \vz 
    - \vu^\T \ve_1 + \frac{\lambda}{\sqrt{n}} \vw^\T \vz.
    \end{equation}
 With $\vg \sim \mathcal{N}(\vzero, \mI_m)$, $\vh \sim \mathcal{N}(\vzero, \mI_n)$, $1 \leq m < n$, $\delta = (m-1)/n$ we have
for any $0<\eta<2$: with probability at least $1-\zeta(n, \eta)$,
\begin{equation}
\begin{aligned}
  \phi_{0}(\vg,\vh) & < & \kappa(\alpha^{*}) + \eta,\\
  \phi_\epsilon(\vg, \vh) & > & \kappa(\alpha^{*}) + \omega(\epsilon) - \eta,
\end{aligned}
\end{equation}
with $\zeta(n,\epsilon) = \zeta(n,\epsilon;A,\delta)$ defined in~\eqref{eq:DefZeta} and $\omega(\epsilon)$ defined in~\eqref{eq:DefOmega}.
\end{lemma}

\begin{proof}
Since $0<\epsilon\leq \max(\alpha^{*},A-\alpha^{*})$ the set $S_{\epsilon} \bydef \set{\alpha: 0 \leq \alpha \leq A,\ \abs{\alpha-\alpha^{*}} \leq \epsilon}$ is a non-empty subset of $[0,\ A]$. Denote  $S^-_{\epsilon} \bydef [0, \alpha^* - \epsilon]$ and $S^+_{\epsilon} = [\alpha^* + \epsilon, A]$ its two convex components (at most one of them may be empty).
Let $\phi_{S^+_{\epsilon}}(\vg, \vh)$ the value of
\eqref{eq:opt_for_phi_eps}, but with $\norm{\vz}$ constrained to lie in $S^{+}_{\epsilon}$ (by convention, this is $+\infty$ when $S_{\epsilon}^{+} = \emptyset$). Similarly define $\phi_{S^-_{\epsilon}}(\vg, \vh)$. When $S^+_{\epsilon}$ is non-empty, since it is convex, we can
effect the same simplifications and min-max swaps as in the proof of Lemma
\ref{lem:concentration_column} (from \eqref{eq:aux_dual_swap_first_line} to \eqref{eq:opt_two_scalars}) to arrive at
\begin{equation*}
  \phi_{S^+_{\epsilon}}(\vg,\vh) = \min_{\alpha : \alpha^* + \epsilon \leq \alpha \leq A} \max_{\beta:0 \leq \beta \leq 1} \phi(\alpha, \beta; \vg, \vh)
\end{equation*}
and similary with $S^{-}_{\epsilon}$ we get when it is non-empty that
\begin{equation*}
  \phi_{S^-_{\epsilon}}(\vg,\vh) = \min_{\alpha : 0 \leq \alpha \leq \alpha^* - \epsilon} \max_{\beta:0 \leq \beta \leq 1} \phi(\alpha, \beta; \vg, \vh)
\end{equation*}
This shows that $\phi_{\epsilon}(\vg,\vh) = \min(\phi_{S-_{\epsilon}}(\vg,\vh),\phi_{S^+_{\epsilon}}(\vg,\vh)) = \phi_{S_{\epsilon}}(\vg,\vh)$ where the notation $\phi_{S_{\epsilon}}(\vg,\vh)$ matches that used in Corollary~\ref{cor:empiricalmaxbeta}.
Moreover by definition (see Lemma~\ref{le:ArgMinKappa}) we have 
\begin{equation*}
    \kappa_{S_\epsilon} \bydef \min_{\alpha \in S_{\epsilon}} \max_{0 \leq \beta \leq 1} \kappa(\alpha, \beta) = \kappa(\alpha^{*}) + \omega(\epsilon).
\end{equation*} 
By  Corollary \ref{cor:empiricalmaxbeta} we have, for $0<\eta<2$, with probability at least $1-\zeta(n, \eta)$: for all $S \subset [0,A]$, $\abs{\phi_{S}(\vg,\vh)-\kappa_{S}} < \eta$. Specializing to $S = [0,A]$ and $S = S_{\epsilon}$ and combining the above yields 
\begin{align*}
    \phi_{0}(\vg,\vh) & = \phi_{[0,A]}(\vg,\vh) < \kappa_{[0,A]} + \eta = \kappa(\alpha^{*}) + \eta\\
    \phi_{\epsilon}(\vg,\vh) & = \phi_{S_{\epsilon}}(\vg,\vh) > \kappa_{S_{\epsilon}}-\eta = \kappa(\alpha^*) + \omega(\epsilon) - \eta.
\end{align*}
\end{proof}

\begin{lemma}
    \label{lem:minmaxswap}
    Let $K > \alpha^{*}/\sqrt{n}$, $\lambda \leq t^{*}$ with $\alpha^{*}$ defined as in Lemma~\ref{le:ArgMinKappa}-Equation~\eqref{eq:kappaminimizer} and $t^{*} = t^{*}_{p}(\delta;n)$ as in Lemma~\ref{lem:Dp_det_prop}-Property 7.  Denote by $\wt{\vx}_K$ any optimal solution of \eqref{eq:po_bounded}.
    For $0 < \epsilon\leq \max(\alpha^{*},K\sqrt{n}-\alpha^{*})$ we have
    \[
        \prob \big[ \abs{\norm{\wt{\vx}_K} - \tfrac{\alpha^*}{\sqrt{n}}} \geq \tfrac{\epsilon}{\sqrt{n}}  \big]  \leq 4 \zeta\big(n, \tfrac{\omega(\epsilon)}{2}\big).
    \]
    with $\zeta(n,\epsilon)=\zeta(n,\epsilon;K\sqrt{n},\delta)$ defined in~\eqref{eq:DefZeta} and $\omega(\epsilon)$ defined in~\eqref{eq:DefOmega}.
\end{lemma}
\begin{proof}
    Denote $\Phi$ the optimal cost of
    \eqref{eq:po_bounded} and $\Phi_\epsilon$ the corresponding cost when
    $\vz = \vx\sqrt{n}$ is further restricted to $\setS_{\epsilon} \bydef \set{\vz \ : \norm{\vz}_{2} \leq A\ \text{and}\ \abs{\norm{\vz}_2 - \alpha^\star} \geq \epsilon}$, with $A \bydef K\sqrt{n}$:
    \begin{equation}
 \label{eq:value_po_bounded}
  \Phi_{\epsilon} = \min_{\substack{\vx : \norm{\vx} \leq K,\\ \abs{\sqrt{n}\norm{\vx}-\alpha^{*}} \geq \epsilon}} \max_{\vu: \norm{\vu} \leq 1} \vu^\T \mA
 \vx + \lambda \norm{\vx}_p - \vu^\T \ve_1.
\end{equation}
    We now want to show that for $\epsilon>0$ we have with high probability $\Phi_{\epsilon} > \Phi = \Phi_{0}$, because this is equivalent
    to $\wt{\vz}_K = \wt{\vx}_K \sqrt{n} \in \setS'_{\epsilon} \bydef \set{\vz \ : \norm{\vz}_{2} \leq A\ \text{and}\ \abs{\norm{\vz}_2 - \alpha^\star} < \epsilon}$.

    By CGMT part~\ref{thmpart:1} and part~\ref{thmpart:cgmt_optval_concentrate}  (see Theorem~\ref{thm:cgmt} in Appendix~\ref{append:CGMT}), 
    denoting $\phi^{P}_{\epsilon}$ (resp. $\phi^{D}_{\epsilon}$) the optimum value of the ``primal'' (resp. ``dual'') auxiliary optimization problem associated to the principal optimization problem~\eqref{eq:value_po_bounded},
       we have for any $c \in \R$ 
    \begin{equation}\label{eq:minmaxswap_step1}
        \prob [\Phi_{\epsilon} < c] \leq 2 \prob [\phi^P_{\epsilon} \leq c] 
        \quad \text{and} \quad
        \prob [\Phi_{0} > c] \leq 2 \prob [\phi^D_{0} \geq c].
    \end{equation}
For the second one, we use additionally that $\phi^D_{0} = \phi^P_{0}$ since  we optimize over convex sets (being a ball, $\setS_{0}$ is convex) and the penalty $\psi(\vx,\vu) = \lambda \norm{\vx}_{p}-\vu^{\T}\ve_{1}$ is convex-concave (see e.g. \cite[Corollary 3.3]{Sion:1958jm}).

Let $C(\vz, \vu, \vw) = C(\vz,\vu,\vz; \vg,\vh)$ be the objective function in
    \eqref{eq:aux_dual_swap} and~\eqref{eq:opt_for_phi_eps}, so that~\eqref{eq:opt_for_phi_eps}  becomes 
    (with $A \bydef K\sqrt{n}$)
    \[
        \phi_{\epsilon} = \max_{\substack{\beta: 0 \leq \beta \leq 1\\\vw: \norm{\vw}_{p^*} \leq 1}} \min_{\substack{\vz : \norm{\vz} \leq A\\ \vz \in \setS_{\epsilon}}} \max_{\vu: \norm{\vu} = \beta} C(\vz, \vu, \vw),
    \]
    and the optimal cost of \eqref{eq:aux_dual_swap} reads $\phi = \phi_{0}$. With these notations, we have (noting that $\max \min \leq \min \max$ is always true)
    \begin{equation}
    \begin{aligned}
        \phi^P_{\epsilon} 
        &\bydef \min_{\substack{\vz: \norm{\vz} \leq A\\\vz \in \setS_{\epsilon}}} \max_{\substack{\vu: \norm{\vu} \leq 1, \\ \vw: \norm{\vw}_{p^*} \leq 1}} C(\vz, \vu, \vw) 
        = \min_{\substack{\vz: \norm{\vz} \leq A\\\vz \in \setS_{\epsilon}}} \max_{\substack{\beta: 0 \leq \beta \leq 1, \\ \vw: \norm{\vw}_{p^*} \leq 1}} \max_{\vu: \norm{\vu}_2 = \beta} C(\vz, \vu, \vw) \\
        &\geq
        \max_{\substack{\beta: 0 \leq \beta \leq 1, \\ \vw: \norm{\vw}_{p^*} \leq 1}} \min_{\substack{\vz: \norm{\vz} \leq A\\\vz \in \setS_{\epsilon}}} \max_{\vu: \norm{\vu}_2 = \beta} C(\vz, \vu, \vw) \ = \phi_{\epsilon}
    \end{aligned}
    \end{equation}
 and

    \begin{equation}
    \begin{aligned}
        \phi^D_{0}
        & \bydef  \max_{\substack{\vu: \norm{\vu} \leq 1, \\ \vw: \norm{\vw}_{p^*} \leq 1}} \min_{\vz: \norm{\vz} \leq A} C(\vz, \vu, \vw)
        = \max_{\substack{\beta: 0 \leq \beta \leq 1, \\ \vw: \norm{\vw}_{p^*} \leq 1}} \max_{\vu: \norm{\vu}_2 = \beta} \min_{\norm{\vz} \leq A} C(\vz, \vu, \vw) \\
        &\leq
        \max_{\substack{\beta: 0 \leq \beta \leq 1, \\ \vw: \norm{\vw}_{p^*} \leq 1}} \min_{\vz: \norm{\vz} \leq A} \max_{\vu: \norm{\vu}_2 = \beta} C(\vz, \vu, \vw) \ = \phi_{0}.
    \end{aligned}
    \end{equation}
    Denote $\underline{\phi}_{\epsilon} \bydef \kappa(\alpha^{*})+\omega(\epsilon)$ and $\bar{\phi} \bydef \kappa(\alpha^{*})$ and use the above with $c_{1} =  \underline{\phi}_{\epsilon}- \eta$ and $c_{2} =  \bar{\phi} + \eta$ where $\eta>0$ is arbitrary to get  
    \begin{equation}
    \begin{aligned}
        &\prob[\Phi_{\epsilon} 
        <  \underline{\phi}_{\epsilon} - \eta]
        &\leq &\ 2 \prob[\phi^P_{\epsilon} \leq  \underline{\phi}_{\epsilon} - \eta]
        &\leq&\ 2 \prob[\phi_{\epsilon} \leq  \underline{\phi}_{\epsilon} - \eta],\\
        &\prob[\Phi > \bar{\phi} + \eta]
        &\leq &\ 2 \prob[\phi^D \geq \bar{\phi} + \eta]
        &\leq&\ 2 \prob[\phi_{0} \geq \bar{\phi} + \eta].
    \end{aligned}
    \end{equation}
    Consider the event ${\cal E} = \set{\Phi_{\epsilon} \geq
    \underline{\phi}_{\epsilon} - \eta \text{ and } \Phi \leq \bar{\phi} + \eta}$.
     For $0< \eta < (\underline{\phi}_{\epsilon}-\bar{\phi})/2 = \omega(\epsilon)/2$ we have $c_{1} > c_{2}$ hence this
    event implies that $\wt{\vz}_{K} \in \setS'_{\epsilon}$, which is what we wanted to prove. For such $\eta$, since $\omega(\epsilon)/2 \leq 1/4 <2$ we can use Lemma~\ref{lem:phi_eps_separation} and a union bound to obtain that this event happens with probability at least $1 - 4 \zeta(n, \eta)$. Hence, for any $0<\eta<\omega(\epsilon)/2$ we have
    \(
    \prob(\wt{\vz}_{K} \notin \setS_{\epsilon}) \leq 4\zeta(n,\eta).
    \)   
    By continuity of $\eta \mapsto \zeta(n,\eta)$ we take the limit when $\eta$ tends to $\omega(\epsilon)/2$. 
\end{proof}

\begin{lemma}
    \label{lemma:bounded_equals_unbounded}
    Let $K > \alpha^{*}/\sqrt{n}$, $\lambda \leq t^{*}$  with $\alpha^{*}$ defined as in Lemma~\ref{le:ArgMinKappa}-Equation~\eqref{eq:kappaminimizer} and $t^{*} = t^{*}_{p}(\delta;n)$ as in Lemma~\ref{lem:Dp_det_prop}-Property 7.  Denote by $\wt{\vx}_K$ any optimal solution of the random bounded problem \eqref{eq:po_bounded} and $\tilde{\vx}$ any optimal solution of the random unbounded problem~\eqref{eq:lasso}.
    For $0 < \epsilon\leq \min(\alpha^{*},K\sqrt{n}-\alpha^{*})$ (NB: here the upper bound on $\epsilon$ is the $\min$, not the $\max$) 
    we have 
    \[
      \prob [ \tilde{\vx}_K \neq \tilde{\vx}] \leq 
      4 \zeta\big(n, \tfrac{\omega(\epsilon)}{2}\big),
    \]  
    and
    \[
        \prob \big[ \abs{\norm{\wt{\vx}} - \tfrac{\alpha^*}{\sqrt{n}}} \geq \tfrac{\epsilon}{\sqrt{n}}  \big]  \leq 4 \zeta\big(n, \tfrac{\omega(\epsilon)}{2}\big).
    \]
      with $\zeta(n,\epsilon)=\zeta(n,\epsilon;K\sqrt{n},\delta)$ defined in~\eqref{eq:DefZeta} and $\omega(\epsilon)$ defined in~\eqref{eq:DefOmega}.
\end{lemma}

\begin{proof}
To handle the case of non-unique solutions, $\tilde{\vx}$ (resp. $\tilde{\vx}_{K}$) may denote the convex set of solutions of the respective convex optimization problems. The property $\tilde{\vx} \neq \tilde{\vx}_K$ then means that the sets do not intersect, and inequalities such as $f(\tilde{\vx}) > c$ are meant to hold for all elements of the set $\tilde{\vx}$. 

We first prove, by contradiction, that if $\tilde{\vx} \neq \tilde{\vx}_K$, then necessarily $\norm{\tilde{\vx}_{K}} = K$.
Suppose the opposite:  despite the fact that $\tilde{\vx} \neq \tilde{\vx}_K$, we have $\norm{\tilde{\vx}_{K}} < K$. Since $\tilde{\vx} \neq \tilde{\vx}_{K}$ we have $\norm{\tilde{\vx}} > K$. Denoting the lasso objective in \eqref{eq:lasso} by $\theta(\vx)$, this means that $\theta(\tilde{\vx}) < \theta(\tilde{\vx}_K)$. By convexity of $\theta$ it follows that all points on the line segment $\vx_\nu = \nu \tilde{\vx}_K + (1 - \nu) \tilde{\vx}$, $0\leq \nu \leq 1$ satisfy 
    \begin{equation}
        \label{eq:lemma_bounded_is_ok_cvx}
        \theta(\vx_\nu) \leq \theta(\tilde{\vx}_K).         
    \end{equation}   
Since $\norm{\vx_{0}} > K$ and $\norm{\vx_{1}}<K$, by continuity there exists $\nu \in (0,1)$ such that $\norm{\vx_\nu} = K$. Further, by \eqref{eq:lemma_bounded_is_ok_cvx}, $\vx_\nu$ is optimizing the bounded problem \eqref{eq:po_bounded}, contradicting our assumption. 

The contraposition of what we just established is that, if $\norm{\tilde{\vx}_{K}} < K$ then $\tilde{\vx} = \tilde{\vx}_{K}$. 
In particular, since we assume $K\sqrt{n}> \alpha^{*}+\epsilon$, we have: 
if  $\abs{\sqrt{n}\norm{\tilde{\vx}_K}-\alpha^{*}} < \epsilon$ then $\norm{\tilde{\vx}_{K}} < K$ 
hence $\tilde{\vx} = \tilde{\vx}_{K}$ and $\abs{\sqrt{n}\norm{\tilde{\vx}}-\alpha^{*}} = \abs{\sqrt{n}\norm{\tilde{\vx}_K}-\alpha^{*}} < \epsilon$.
It follows that $\prob[ \tilde{\vx}_K \neq \tilde{\vx}] \leq \prob [ \abs{\sqrt{n}\norm{\tilde{\vx}_K}-\alpha^{*}} \geq \epsilon]$ and
\[
\prob [ \abs{\sqrt{n}\norm{\tilde{\vx}}-\alpha^{*}} \geq \epsilon] 
\leq
\prob [ \abs{\sqrt{n}\norm{\tilde{\vx}_K}-\alpha^{*}} \geq \epsilon].
\]
We conclude using Lemma~\ref{lem:minmaxswap}.
\end{proof}

\begin{lemma}
    \label{lem:sqrtn_lipschitz}
    With $\Delta_{p}$ defined as in Lemma~\ref{lem:concentration2}, the function $f_{\vh}(\beta) \bydef \abs{
    \Delta_p(\beta; \vh,\lambda) - \Delta_p(\beta;n,\lambda)}$ is $\max \set{\norm{\vh}/\sqrt{n},1}$-Lipschitz in $\beta$.
\end{lemma}

\begin{proof}
    We omit the dependency in $\lambda$ for brevity and write by definition 
    \[
        \abs{\Delta_p(\beta_1; \vh) - \Delta_p(\beta_2; \vh)}
        = \tfrac{1}{\sqrt{n}}\ \abs{\dist(\beta_1 \vh, {\cal C}) - \dist(\beta_2 \vh, {\cal C})}
        \leq \tfrac{1}{\sqrt{n}} \ \norm{\beta_1 \vh - \beta_2 \vh} \leq \tfrac{\norm{\vh}}{\sqrt{n}} \abs{\beta_1 - \beta_2},
    \]
    where we used the fact that the Euclidean distance to a convex set is
    1-Lipschitz with respect to the Euclidean metric. Further, we have

    \begin{eqnarray*}
        \abs{\Delta_p(\beta_1) - \Delta_p(\beta_2)} 
        &=& \abs{ \E (\Delta_p(\beta_1; \vh)  - \Delta_p(\beta_2; \vh))}
        \leq \E \abs{ (\Delta_p(\beta_1; \vh)  - \Delta_p(\beta_2; \vh))}\\
        & \leq& \tfrac{\abs{\beta_1 - \beta_2}}{\sqrt{n}}\ \E \norm{\vh}
        \leq \abs{\beta_1 - \beta_2}.
    \end{eqnarray*}
    To conclude, note that the Lipschitz constant of the difference of two Lipschitz functions does not exceed the largest of the two Lipschitz constants, and that taking the absolute value does not change it.
\end{proof}

\bibliographystyle{IEEEtranSA}

\bibliography{pseudo,pseudo_supp}

\end{document}